\documentclass[aps,pra,twocolumn,superscriptaddress,longbibliography]{revtex4-2}

\usepackage{amssymb,amsmath,amsthm,bbm,bbold}
\usepackage{amsfonts}
\usepackage{graphicx}
\usepackage{mathtools}
\usepackage[breaklinks,colorlinks,allcolors=blue]{hyperref}
\usepackage[normalem]{ulem}
\usepackage[T1]{fontenc}
\usepackage{color}
\usepackage{xcolor}
\usepackage{tcolorbox}

\usepackage{pifont}
\usepackage{amssymb}
\usepackage{tikz}
\usepackage{booktabs}
\usepackage{tabularx}
\usepackage{array}
\newcolumntype{L}[1]{>{\raggedright\arraybackslash}p{#1}}
\newcolumntype{Y}{>{\raggedright\arraybackslash}X}

\usepackage{makecell}

\def\ZZ{\mathbbm{Z}}
\def\RR{\mathbbm{R}}
\def\CC{\mathbbm{C}}

\newcommand{\bd}[1]{\boldsymbol{#1}}

\usepackage{accents}

\newcommand{\Tr}{\mathrm{Tr}}

\newcommand{\bra}[1]{\mbox{$\langle #1 |$}}
\newcommand{\ket}[1]{\mbox{$| #1 \rangle$}}

\newtheorem{thm}{Theorem}
\newtheorem{lem}[thm]{Lemma}
\newtheorem{cor}[thm]{Corollary}

\begin{document}
\title{Operational Impact of Quantum Resources in Chemical Dynamics}
\author{Julia Liebert}
\affiliation{Department of Chemistry, Princeton University, Princeton, NJ 08540, United States}
\author{Gregory D. Scholes}
\email{gscholes@princeton.edu}
\affiliation{Department of Chemistry, Princeton University, Princeton, NJ 08540, United States}

\begin{abstract}
Quantum coherence and other non-classical features are widely discussed in chemical dynamics, yet it remains difficult to quantify when such resources are operationally relevant for a given process and observable. While quantum resource theories provide a comprehensive framework for comparing free and resourceful settings, existing approaches typically rely on resource monotones or on performance bounds under free operations, and do not directly quantify the maximal influence a chosen resource can exert on a fixed chemical dynamics.
Here, we introduce task specific, process level quantifiers that upper bound the largest change a quantum resource can induce in a target figure of merit. Central is a resource impact functional $\mathcal{C}_M(\Lambda)$, defined by comparing a state with its paired resource-free counterpart under the same quantum channel $\Lambda$, which admits an operational interpretation in binary hypothesis testing.
We derive variation and time bounds that constrain how rapidly a resource can modify a target signal, providing resource-aware analogues of quantum speed limits. Moreover, we show that open system dynamics can be decomposed into free and resourceful components such that only the resourceful component contributes to $\mathcal{C}_M(\Lambda)$, thereby isolating the parts of a generator responsible for resource-induced changes in the observable.
We illustrate the framework exemplary for energy transfer in a donor-acceptor dimer in two analytically solvable regimes. 
Our results provide a general toolbox for diagnosing and benchmarking quantum resource effects in molecular processes.
\end{abstract}
\date{\today}   
\maketitle

\section{Introduction\label{sec:intro}}
 
The interface of quantum information science and chemistry offers a unique opportunity to advance quantum technologies while at the same time providing new insights into chemical mechanisms \cite{WFMKSYZB20, S23, WS24, ZOHYSW24, SOCMNHF25}. Quantum-information measures such as entanglement, quantum discord, and coherence have increasingly been used as descriptors of chemical structure and dynamics, including chemical bonding and bond formation \cite{RNW06,BTLR12,ZBDLA15, DTBBLA15,BT15, BVSBPL19, DMSSSZS21,DKZS22, DGS24, LDS24, faraday24, ZRLA25,DMS25,GDSKCMR25}, vibrational spectra \cite{CBLR24, WBLCS25}, and critical points on potential energy surfaces \cite{EMPD15}. 
Moreover, experimental and theoretical evidence for long-lived electronic and vibronic coherences in molecular systems, such as photosynthetic complexes and other molecular aggregates \cite{ECRAMCBF07,PH08, CWWCBS10,PHFCHBE10, IF10, Mukamel10, OS11, Engel11, CKPM12, HE12, RANFTZG14, CS15, DPCATM17, SAP17,SFCAGBC17,CCCDHKJM20,CB20, SFS21,JRS23,CCB23,BZA25}, as well as spin-correlated radical pairs \cite{HW21,Mani22,FLMH20}, has led to an ongoing debate over whether such coherences are functionally exploited or are merely by-products of the underlying dynamics. 
In all these settings, a central challenge is to distinguish the contribution of genuinely quantum effects from that of classical structure and environmental noise.

However, even when a positive correlation is observed between a quantum quantity, such as coherence or entanglement, and a chemical figure of merit, such as a yield or transport efficiency, this does not by itself establish that the quantum resource \textit{causes} the observed enhancement as a third structural feature may influence both the resource and the observable. 
Using coherence or entanglement merely as static descriptors is therefore not sufficient to determine the \textit{operational} role of quantum effects in chemical systems, in particular for dynamical processes. 
A meaningful and comprehensive assessment requires comparing the effect of a given quantum feature to an appropriate ``classical'' baseline in which that feature has been removed.

Quantum resource theories (QRTs) provide a natural framework for such comparisons \cite{CG19,Gour2025QRT}. 
In a QRT, both quantum states and operations are partitioned into free and resourceful ones: free objects model states and controls available at negligible cost, while resourceful ones correspond to genuinely quantum features such as coherence, asymmetry, or entanglement. 
A central theme in this literature is to identify resources via improved performance in operational tasks compared to the best free strategy. 
In particular, several works have shown that resourceful states, measurements, and channels can outperform all free counterparts in suitable state, subchannel, or channel discrimination problems, and that standard resource monotones such as the (generalised) robustness or robustness of coherence quantitatively capture this advantage~\cite{NBCPJA16,TRBLA19,TR19,OB19,LBL20}. 
These results compare optimised resourceful and free strategies, typically via
\[
\max_{\rho\in\mathcal{R}}(\text{task performance})
\quad\text{vs.}\quad
\max_{\sigma\in\mathcal{F}}(\text{task performance})\,,
\]
where $\mathcal{R}$ and $\mathcal{F}$ denote the resourceful and free sets, respectively.
In contrast, to address the functional chemical question pursued here, namely, how sensitive a specific chemical figure of merit is to a quantum feature for a fixed process, it is important to keep each state paired with its own free counterpart. In QRTs, such a paired quantum-classical baseline is provided by resource-destroying maps $\mathcal{G}$, which map every state to a corresponding free state by ``removing'' the chosen resource \cite{LHL17,Gour17,CG19}. 
In this setting, we compare the task performance of $\rho$ directly with that of $\mathcal{G}(\rho)$ under the same dynamics, rather than optimising independently over all resourceful and all free inputs. 
As we will show, this paired comparison is crucial for quantifying the operational impact of quantum resources in concrete chemical processes.

Furthermore, recent works have begun to connect QRTs to concrete chemical scenarios, for example by linking non-stabilizerness to bond formation \cite{ST25}, or by deriving thermodynamic resource-theoretic bounds on photoisomerization yields \cite{YHL20, SPSHP25}. These approaches typically focus on characterising what free operations can achieve. 
In the photoisomerization case \cite{YHL20, SHP22, SPSHP25, TSHP25}, for instance, one asks for the maximal yield compatible with thermodynamic constraints, largely independent of the detailed molecular dynamics.

In contrast, in molecular settings the dynamics is often effectively fixed, for example by a specific open system Hamiltonian and its associated master equation, and is usually not a free operation for the resource under consideration. 
If one is interested in the functional role of a given quantum resource for a particular chemical process, analysing restrictions under free operations alone therefore does not directly address the relevant question. 
Instead, the central issue becomes: \textit{for this fixed process, to what extent can a chosen quantum resource actually influence a chemically meaningful observable such as a yield or another figure of merit?} 
Answering this process-level question requires shifting the emphasis from optimising over free operations to analysing the impact of quantum resources for the actual dynamics. 
While not phrased in the language of resource theories, this type of resource-induced operational advantage already appears, for example, in enhanced coherent control via entangled reactant states \cite{PFB18,DTB25}.

To address this question when and where quantum effects are operational in chemical processes and can be harnessed to improve target yields, we develop in this paper a theoretical framework tailored to the setting with specific processes and readouts. 
Building on the paired quantum-classical baseline provided by a resource-destroying map $\mathcal{G}$, we introduce a so-called resource impact functional that provides a rigorous upper bound on the maximal modification of a chosen figure of merit that a given process can draw from the resource. 
This functional satisfies pre- and post-processing properties and admits a clear operational meaning in terms of a binary state discrimination task. 
To bound how quickly a resource can modify a yield, we further define an instantaneous advantage rate whose time integral yields explicit minimal time guarantees to achieve a prescribed resource-induced change in the target yield, thus providing resource-aware analogues of quantum speed limits. 
In addition, we show that general open system dynamics can be decomposed into free and resourceful components relative to a resource-destroying map, with only the latter contributing to a resource-induced influence on the observable. This decomposition allows us to isolate those parts of the evolution that are genuinely responsible for changing the readout.

Returning to the role of coherence in molecular systems raised at the beginning of this section, we illustrate the calculation of the resource impact functional and the instantaneous advantage rate using donor-acceptor excitation transfer in the single excitation manifold. This setting allows for an analytic characterization of how the resource impact functional identifies the time windows over which the dynamics and, thus, the chosen readout, is most sensitive to coherence.
 
In the longer term, we expect that this framework can be exported beyond the examples considered here to other quantum chemical phenomena where the functional role of quantum resources remains unclear.

The remainder of the paper is structured as follows: Sec.~\ref{sec:scientific-approach} formulates the scientific problem, highlights the key open questions, and outlines the resource theoretic approach at a non-technical level.
Sec.~\ref{sec:defs} then introduces the theoretical framework for resource theories with a resource-destroying map. In Sec.~\ref{subsec:capacity} we define the resource impact functional, establish core properties such as convexity and behaviour under pre- and post-processing, and discuss their implications for when quantum effects are operationally relevant. 
Sec.~\ref{subsec:dynamics} introduces a dynamical version of the framework and derives variation bounds on how fast the resource impact functional can change, and Sec.~\ref{subsec:time-bound} uses these bounds to obtain minimal time constraints on achieving a prescribed resource-induced modification of a target yield, as well as feasibility bounds on yield improvements for a fixed time budget. 
In Sec.~\ref{sec:dissect-dyn} we show how Lindblad dynamics can be dissected into free and resourceful components, with only the latter contributing to a resource-induced advantage. 
Since not all resource theories admit a resource-destroying map, Sec.~\ref{sec:RT-wo-RDM} explains how several of the concepts from Sec.~\ref{sec:defs} can be generalized to this more general setting. 
Finally, Sec.~\ref{sec:DAM} illustrates how the resource impact functional and instantaneous rate can be calculated analytically for the donor-acceptor model in the single excitation manifold. 

\section{Problem statement and conceptual approach\label{sec:scientific-approach}}

We seek to determine, for a fixed process and readout in chemical dynamics, whether a chosen quantum resource can influence the quantity of interest, by how much, and under which conditions its influence is necessarily absent. This leads to three concrete questions:

\begin{enumerate}
\item[(i)] \textbf{Paired classical benchmark.}
The main difficulty is not detecting a quantum feature but defining the appropriate classical baseline. For a given initial state, dynamics, and readout, one would like to compare the state with its corresponding free counterpart under identical conditions. Standard ``unpaired'' benchmarks, which optimize independently over all states and over free states, generally compare unrelated inputs and therefore do not isolate the effect of the resource itself.
\item[(ii)] \textbf{Process-specific functional role.}
In chemical applications, the dynamics is typically fixed by a Hamiltonian or master equation and need not be free with respect to the resource of interest. The relevant question is therefore process- and readout-specific: for the given dynamics, to what extent can the chosen resource influence the chosen readout?
\item[(iii)] \textbf{No-go criteria.}
The field lacks general criteria for when a given quantum resource cannot affect a specified readout. Such no-go statements are essential for interpreting experiments, because they allow one to exclude a resource-based explanation rather than merely observe the presence of nonclassical features.
\end{enumerate}

These questions require a framework that defines the free baseline for a chosen quantum feature in a way that is not tied to a specific chemical model. Quantum resource theories (QRTs) provide such a framework. In a QRT, one specifies which states are free, i.e., do not possess the resource of interest, and regards all others as resourceful. The same distinction can be made for operations, which must be compatible with the set of free states in the sense that free operations map free states to free states. Different choices of quantum resources recover different notions of nonclassicality, such as coherence, entanglement, or asymmetry \cite{CG19, Gour2025QRT}. QRTs therefore provide a common language for analysing quantum effects in chemical dynamics without committing to a single resource from the beginning. We refer to Appendix~\ref{app:QRT} for a brief review of the main concepts of QRTs used in this paper.

The guiding idea is to compare a given initial state with the corresponding resource-free reference state under otherwise identical conditions. For resource theories that admit a resource-destroying map, this reference state is obtained by applying the map to the initial state. Evaluating both states under the same dynamics and the same readout provides a physically meaningful baseline and avoids the ambiguity of unpaired optimizations. In this way, the framework addresses the first two questions above: it provides a paired resource-free benchmark and quantifies the maximal functional role that a non-zero resource in the initial state can play in the selected readout.
More precisely, for a fixed dynamics and readout, this paired construction induces a state resolved change in the selected readout between an initial state and its resource-free reference. Its supremum over admissible initial states defines the corresponding process-level resource impact, namely the largest readout change that the fixed process can generate due to the chosen resource. If this quantity vanishes, the resource cannot affect the selected readout, providing the desired no-go criterion.

Sections~\ref{sec:defs} and \ref{sec:RT-wo-RDM} develop the framework in detail, while Sec.~\ref{sec:DAM} illustrates the main concepts on a minimal donor-acceptor model.

\section{Theoretical concepts\label{sec:defs}}

Let $\mathcal{H}$ be a finite-dimensional complex Hilbert space, $\mathcal{B}(\mathcal{H})$ the set of bounded linear operators on $\mathcal{H}$ and $\mathcal{D}(\mathcal{H})\subset \mathcal{B}(\mathcal{H})$ the set of density operators on $\mathcal{H}$. Here, we restrict to finite-dimensional Hilbert spaces, both for simplicity and because this setting closely reflects most computational frameworks in quantum chemistry.

To investigate under which circumstances chemical processes can benefit from quantum resources, we first focus on resource theories equipped with a resource-destroying map, such as those for coherence or asymmetry \cite{LHL17,Gour17,CG19} (see Appendix~\ref{app:QRT} for a brief recap of resource theories including resource-destroying maps). 
In this setting, the resource-destroying map $\mathcal{G}:\mathcal{D}(\mathcal{H})\to\mathcal{D}(\mathcal{H})$ provides a precise notion of the corresponding ``classical'' behaviour obtained when the resource is removed, and thereby singles out a canonical free counterpart of any state or process. 
This will serve as the starting point for our later generalization to resource theories without a resource-destroying map, discussed in Sec.~\ref{sec:RT-wo-RDM}.

\subsection{Resource impact functional \label{subsec:capacity}}

Let $\{M_i\}$ be a positive operator-valued measure (POVM) on $\mathcal{H}$. For a fixed POVM element $M$, we define the corresponding yield by
\begin{equation}\label{eq:yield}
Y(\rho) := \Tr[M\rho]\,,\qquad \rho\in \mathcal{D}(\mathcal{H})\,.
\end{equation}
In most of what follows it is not essential that $M$ be part of a POVM and it suffices that $M$ is a bounded self-adjoint operator on $\mathcal{H}$. 
For generality, we will therefore take $M\in \mathcal{B}_{\mathrm{sa}}(\mathcal{H})$ and impose the POVM constraints (i.e., $0\leq M\leq \mathbbm{1}$, and $\sum_i M_i=\mathbbm{1}$ for a full measurement) only when explicitly required.

Consider a resource theory equipped with a resource-destroying map $\mathcal{G}$ (see Appendix~\ref{app:RDP}). A common example of a resource-destroying map in chemical physics is the dephasing map, which removes coherences in the site basis. We use this map in Sec.~\ref{sec:DAM} to illustrate the framework developed in this section for coherence-induced yield changes in a minimal energy transfer model.

In general, $\mathcal{G}$ need not be linear. In fact, for non-convex resource theories (such as quantum discord~\cite{LHL17}) any resource-destroying map that is idempotent and fixes all free states cannot be linear, since the image of the convex set of density operators $\mathcal{D}(\mathcal{H})$ under a linear map is convex, whereas the free set of a non-convex resource theory is not.
Because an adjoint $\mathcal{G}^\dagger$ is defined only for linear maps, one cannot, in general, transfer statements to the Heisenberg picture via $\mathcal{G}^\dagger$ when $\mathcal{G}$ is non-linear.
However, the Heisenberg picture will turn out to be useful in the following. 
Accordingly, we first develop results that do not rely on an adjoint (allowing potentially non-linear $\mathcal{G}$), and then present refinements for the case in which $\mathcal{G}$ is linear. 
Additional simplifications arise when $\mathcal{G}$ is self-adjoint, i.e., $\mathcal{G} = \mathcal{G}^\dagger$ with respect to the Hilbert-Schmidt inner product. 
Most examples of linear resource-destroying maps in the literature, including dephasing maps and group twirls, are self-adjoint. For an example of a resource theory with a linear resource-destroying map that is not self-adjoint, see Appendix \ref{app:RDP}. 

Given a fixed quantum channel $\Lambda$ and an input state $\rho\in\mathcal{D}(\mathcal{H})$, we define the yield difference relative to the resource-free baseline as
\begin{equation} \label{eq:yield-change}
\Delta Y(\rho) := \Tr \left[M \left(\Lambda(\rho) -\Lambda(\mathcal{G}(\rho))\right)\right]\,.
\end{equation}
If $\mathcal{G}:\mathcal{B}(\mathcal{H})\to\mathcal{B}(\mathcal{H})$ is linear, then by moving to the Heisenberg picture we can write
\begin{equation}
\Delta Y(\rho) = \Tr \left[(\mathrm{id}-\mathcal{G}^\dagger)\left(\Lambda^\dagger(M)\right)\,\rho\right]
\qquad (\mathcal{G}\ \text{linear})\,,
\end{equation}
where $\mathrm{id}:\mathcal{B}(\mathcal{H})\to\mathcal{B}(\mathcal{H})$ is the identity map. 
If in addition $\mathcal{G}$ is self-adjoint, $\mathcal{G}^\dagger=\mathcal{G}$, and the linear functional $\Delta Y(\rho)$ can be written as the expectation value of the observable $X := (\mathrm{id}-\mathcal{G})(\Lambda^\dagger(M))$. 

\subsubsection{Definition\label{subsubsec:C-def}}

The yield change $\Delta Y$ defined above is state-specific.
To quantify how strongly a quantum channel describing some physical process can, in principle, exploit the presence of a quantum resource, we introduce the \textit{resource impact functional} with respect to a POVM element $M$ as
\begin{equation}\label{eq:capacity-sup}
\mathcal{C}_M(\Lambda):= \sup_{\rho\in\mathcal{D}(\mathcal{H})}  \left\vert \Tr\left[M\left(\Lambda(\rho) - \Lambda(\mathcal{G}(\rho))\right)\right]\right\vert  \geq 0\,.
\end{equation}
In other words, $\mathcal{C}_M(\Lambda)$ is the maximal change, over all density operators, in the figure of merit encoded by $M$ that can be induced by a single application of $\Lambda$ when comparing a state to its free counterpart. In particular, because the supremum is taken over all density operators $\rho\in \mathcal{D}(\mathcal{H})$, $\mathcal{C}_M(\Lambda)$ provides a state-independent upper bound on the resource-induced change, valid uniformly for any initial preparation.

For a fixed choice of $M$, the state-dependent yield change $\Delta Y(\rho)$ can be decomposed into its maximal enhancement $\mathcal{C}_M^+(\Lambda):=\sup_{\rho}\Delta Y\geq 0$ and maximal suppression $\mathcal{C}_M^-(\Lambda):=\sup_\rho(-\Delta Y) = -\inf_{\rho}\Delta Y$,
such that 
\begin{equation} 
-\mathcal{C}_M^-(\Lambda)  \leq \Delta Y\leq \mathcal{C}_M^+(\Lambda) 
\end{equation} 
and $\mathcal{C}_M(\Lambda) = \max\{\mathcal{C}_M^-(\Lambda), \mathcal{C}_M^+(\Lambda)\}$. This decomposition into $\mathcal{C}_M^\pm(\Lambda)$ can be useful in chemical dynamics because it distinguishes whether a given resource enhances or suppresses the chosen figure of merit. In excitation energy transfer, for example, coherence can either promote or hinder transport, depending on the initial superposition, interference conditions, and environmental regime \cite{OCLFJ08, Jang09, DCHP10, HC12, CMHP12, CHP12}.

In resource theories where the resource-destroying map $\mathcal{G}$ is not unique, the quantity $\mathcal{C}_M(\Lambda)$ implicitly depends on the choice of $\mathcal{G}$. 
In the following we fix one such $\mathcal{G}$ and omit it from the notation, i.e.\ $\mathcal{C}_M(\Lambda)\equiv \mathcal{C}_{M,\mathcal{G}}(\Lambda)$; all statements are to be understood with respect to this fixed choice.

Since for finite-dimensional $\mathcal{H}$ the state space $\mathcal{D}(\mathcal{H})$ is compact and the functional in Eq.~\eqref{eq:capacity-sup} is continuous, so the supremum is in fact a maximum. 
We now derive a first specialization of $\mathcal{C}_M(\Lambda)$, which relies on the linearity of $\mathcal{G}$.

\begin{thm}\label{thm:capacity-G-linear}
Let $M\in\mathcal{B}_{\mathrm{sa}}(\mathcal{H})$, $\dim(\mathcal{H})<\infty$, $\Lambda$ be a quantum channel, and let $\mathcal{G}:\mathcal{B}(\mathcal{H})\to\mathcal{B}(\mathcal{H})$ be a linear resource-destroying map. 
Define $\mathcal{C}_M(\Lambda)$ as in Eq.~\eqref{eq:capacity-sup}. Then:
\begin{itemize}
\item[(i)] The supremum in Eq.~\eqref{eq:capacity-sup} can be restricted to pure states,
\begin{equation}
\mathcal{C}_M(\Lambda) 
= \sup_{\substack{[\Psi]\in\mathbbm{P}(\mathcal{H})\\ \|\Psi\|=1}}
\left\vert\Tr\!\left[ (\mathrm{id}-\mathcal{G}^\dagger)\left(\Lambda^\dagger(M)\right)\, \ket{\Psi}\!\bra{\Psi}\right]\right\vert\,,
\end{equation}
where $\mathbbm{P}(\mathcal{H})$ denotes the complex projective Hilbert space.
\item[(ii)] If $\mathcal{G}$ is Hermiticity preserving, then
\begin{equation}\label{eq:capacity}
\mathcal{C}_M(\Lambda) = \left\|(\mathrm{id} - \mathcal{G}^\dagger)\left(\Lambda^\dagger(M)\right)\right\|_{\infty}\,.
\end{equation}
\item[(iii)] If, in addition, $\mathcal{G}$ is idempotent (i.e., $\mathcal{G}\circ\mathcal{G}=\mathcal{G}$), then a necessary and sufficient condition for $\mathcal{C}_M(\Lambda)=0$ is
\begin{equation}
\Lambda^\dagger(M)\in \mathrm{im}(\mathcal{G}^\dagger)\,,
\end{equation}
where $\mathrm{im}(\mathcal{G}^\dagger)$ denotes the image of $\mathcal{G}^\dagger$.
\end{itemize}
\end{thm}
We provide the proof of Theorem~\ref{thm:capacity-G-linear} in Appendix~\ref{app:proof-thm-1}. 

In the non-linear case, the adjoint-based criterion (iii) does not apply. 
A simple sufficient, but generally stronger than necessary, condition ensuring $\mathcal{C}_M(\Lambda)=0$ is a kind of ``resource blindness'' of the quantum channel $\Lambda$, namely $\Lambda(\rho) = \Lambda(\mathcal{G}(\rho))$ for all $\rho\in\mathcal{D}(\mathcal{H})$. 
In this situation the resource under consideration affords no yield advantage for the observable $M$ and is, in this sense, not operationally functional for the task defined by $(\Lambda,M)$. 
For linear $\mathcal{G}$, the minimal necessary and sufficient condition is given by (iii), which is strictly weaker than full resource blindness of $\Lambda$. 
Moreover, (iii) together with (i) implies that $\mathcal{C}_M(\Lambda)$ vanishes for linear $\mathcal{G}$ whenever $\Lambda$ is $\mathcal{G}$-covariant and $M$ is $\mathcal{G}$-invariant, as expected as in that case the dynamics and the readout are insensitive to the resource-destroying operation.

For non-linear $\mathcal{G}$, the objective $f(\rho)= |\Delta Y(\rho)|$ is in general non-convex in $\rho$ because of the non-linear dependence $\rho\mapsto\mathcal{G}(\rho)$. 
As a consequence, the supremum in Eq.~\eqref{eq:capacity-sup} need not be attained at an extremal point (i.e., pure state) of $\mathcal{D}(\mathcal{H})$ and there can be settings in which a mixed input state outperforms every pure state input. 
This should be contrasted with the case of linear $\mathcal{G}$, where $\Delta Y(\rho)$ is affine in $\rho$ and hence $f(\rho)$ is convex. Thus, mixing (e.g., due to non-zero temperature or environmental noise) cannot increase $f(\rho)$ beyond its value for the best component state.
In the future, it will be interesting to characterize the maximizers in Eq.~\eqref{eq:capacity-sup} for given pairs $(\Lambda,\mathcal{G})$, in order to identify which properties of a quantum state have operational impact on the target yield and how they interact with the dynamics described by $\Lambda$.

At this point, it is important to explain how the resource impact functional $\mathcal{C}_M(\Lambda)$ differs from related notions in the literature~\cite{ZZF00, GPPZ13, MK15,NBCPJA16, LHL17, BKZW17, BDGMW17,ZSS18, SCVZ18, ZCV18, SC19, TELP19,CG19, LBL20, MNRJJB20}. 
A widely studied, and resource theory specific, example is the coherence-generating power~\cite{MK15,BKZW17,BDGMW17,SCVZ18,ZCV18,LBL20}, which quantifies how much coherence a channel can create from free (incoherent) states, as evaluated by a chosen coherence monotone. 
By contrast, when applied to the resource theory of coherence, $\mathcal{C}_M(\Lambda)$ does not measure the increase of coherence itself. Instead, it captures how strongly the quantum channel $\Lambda$ can exploit a resource to change a task-specific readout described by $M$, i.e., the anticipated target yield. This shift from resource creation to task-level impact is essential for analysing the operational role of quantum resources in physical processes described by $\Lambda$. 
Furthermore, the yield change $\Delta Y (\rho)$ in Eq.~\eqref{eq:yield-change} enables a state-resolved analysis when required.  
Finally, the definition of $\mathcal{C}_M(\Lambda)$ in Eq.~\eqref{eq:capacity-sup} only assumes the existence of a resource-destroying map and is otherwise universally applicable to any resource theory with such a map. 
Extensions to resource theories without a resource-destroying map will be discussed in Sec.~\ref{sec:RT-wo-RDM}.

Furthermore, suppose the system is prepared in a (possibly mixed) state $\rho$ and undergoes the transformation $\Lambda(\rho)$. 
If $\Lambda$ is not a free operation in the chosen resource theory, then the achievable change in the target yield $Y$ can exceed what is possible under free operations, i.e., it can be larger than the maximal yield change obtainable when $\Lambda$ is restricted to the free set of quantum channels. 
The task specific resource impact functional $\mathcal{C}_M(\Lambda)$ therefore measures the maximum modification that the single-step process $\Lambda$ can draw from the resource with respect to $M$, while $\mathcal{C}_M(\Lambda)=0$ means that the quantum resource is effectively irrelevant for the process described by $\Lambda$ and the observable $M$. 
In the case where $\mathcal{G}$ is linear, Theorem \ref{thm:capacity-G-linear} provides a precise condition for when this happens.

To further emphasise the distinct operational meaning of $\mathcal{C}_M(\Lambda)$, consider a resource theory with a set of free states $\mathcal{F}$ and define
\begin{equation}\label{eq:def-PiML}
\Pi_M(\Lambda):=  \sup_{\rho\in\mathcal{D}(\mathcal{H})} \big|\Tr[M\,\Lambda(\rho)]\big|- \sup_{\sigma\in \mathcal{F}}\big|\Tr[M\,\Lambda(\sigma)]\big|\geq 0\,,
\end{equation}
which is the difference between the optimal performance achievable with arbitrary input states and the optimal performance achievable when the input is restricted to free states.
The quantity $\Pi_M(\Lambda)$ thus provides a meaningful measure of the operational advantage enabled by resourceful states, and functionals which consider the optimum over all states versus an optimum over free states are common in the literature \cite{NBCPJA16, TR19, TRBLA19, KTAY24}. We discuss in Appendix \ref{app:state-disc} how $\Pi_M(\Lambda)$ can be interpreted as the maximal advantage of resourceful states in an appropriate binary decision (discrimination-type) task, in the sense of Ref.~\cite{KTAY24}. 

\begin{figure}[tb]
\includegraphics[width=\linewidth]{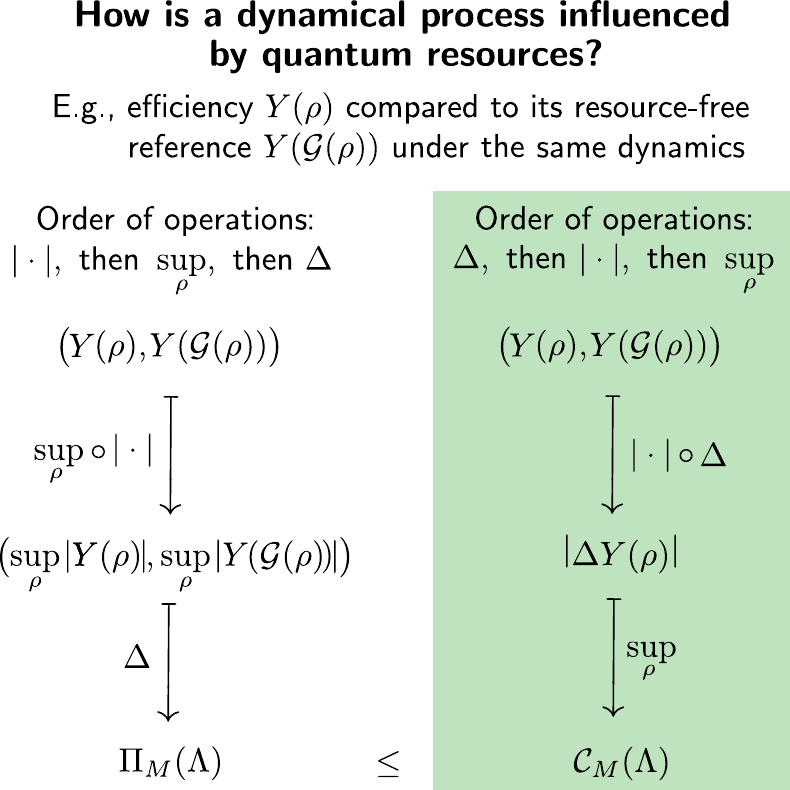}
\caption{Comparison of $\mathcal{C}_M(\Lambda)$ and $\Pi_M(\Lambda)$ constructed from the same readout pair $(Y(\rho),Y(\mathcal{G}(\rho)))$, with $\mathcal{G}(\rho)$ the resource-free counterpart of $\rho$. The inequality (Corollary \ref{cor:CM-PiM}) follows from the non-commutativity of the supremum over $\rho$ and the difference map $\Delta:\mathbbm{R}^2\to\mathbbm{R}$, $\Delta(a,b)=a-b$ in combination with the absolute value $|\cdot|$. The paired quantum-classical construction introduced in Sec.~\ref{subsec:capacity} is highlighted in green.
 \label{fig:summary-CM-PiM}}
\end{figure}

However, unlike $\mathcal{C}_M(\Lambda)$, the definition in Eq.~\eqref{eq:def-PiML} does not compare a state $\rho$ with a specific free baseline obtained from it (such as $\mathcal{G}(\rho)$ in the presence of a resource-destroying map). Instead, the two optimisations are carried out independently over all density operators $\rho\in\mathcal{D}(\mathcal{H})$ and over free states $\sigma\in\mathcal{F}$. As a result, the explicit pairing between each state and its resource-free counterpart, which is crucial for the interpretation of $\mathcal{C}_M(\Lambda)$, is lost. This paired quantum-classical baseline is especially relevant for chemical physics problems such as, for instance, excitation energy transfer \cite{IF12, CS15, JRWS18, JZMTMD26} and radical pair magnetoreception \cite{HBH12, FLMH20}, where the operational role of quantum effects under fixed dynamics has long been debated.
\begin{cor}\label{cor:CM-PiM}
Let $M\in\mathcal{B}_{\mathrm{sa}}(\mathcal{H})$, $\Lambda$ a linear completely positive and trace-preserving (CPTP) map, and $\mathcal{F}$ denote the set of free states of a given resource theory. Then, 
\begin{equation}\label{eq:rel-Pi-C}
\Pi_M(\Lambda) \leq \mathcal{C}_M(\Lambda)
\end{equation}
where $\mathcal{C}_M(\Lambda)$ and $\Pi_M(\Lambda)$ are defined in Eqs.~\eqref{eq:capacity-sup} and \eqref{eq:def-PiML}, respectively. 
\end{cor}
The proof of Corollary \ref{cor:CM-PiM} is provided in Appendix \ref{app:proof-cor-2}.

We can interpret Corollary \ref{cor:CM-PiM} as follows: In the definition of $\Pi_M(\Lambda)$ in Eq.~\eqref{eq:def-PiML}, the second term selects the free input state that maximises the yield for the given task specified by $(\Lambda,M)$. Using this globally optimal free strategy as a benchmark makes the free theory look as strong as possible and can therefore only decrease the observed advantage of resourceful inputs. 
In contrast, $\mathcal{C}_M(\Lambda)$ compares each state $\rho$ with its own ``free counterpart'' $\mathcal{G}(\rho)$, so the resourceful input and the baseline are not optimised independently. Since $\mathcal{G}(\rho)$ is just one particular free state, the contrast per state appearing in $\mathcal{C}_M(\Lambda)$ can only be larger than or equal to the advantage measured against the globally optimal free input. This explains why $\Pi_M(\Lambda)$ is always a lower bound to $\mathcal{C}_M(\Lambda)$. Figure~\ref{fig:summary-CM-PiM} summarizes the difference between $\Pi_M(\Lambda)$ and the resource impact functional $\mathcal{C}_M(\Lambda)$. In particular, it contrasts the unpaired comparison encoded by $\Pi_M(\Lambda)$ with the paired quantum-classical baseline introduced in this work, which is highlighted by the green shaded region. Here, $\Delta$ denotes the standard difference map that sends the tuple $(Y(\rho),\,Y(\mathcal{G}(\rho)))$ of paired quantum and classical readouts, for instance, an efficiency, to the difference $Y(\rho)-Y(\mathcal{G}(\rho))$. The inequality in the bottom row reflects the fact that taking the supremum over density operators $\rho \in \mathcal{D}(\mathcal{H})$ does not commute with applying the difference map $\Delta$.

\subsubsection{Key properties\label{subsec:properties-CM}}

We now analyse the properties of the resource impact functional $\mathcal{C}_M(\Lambda)$ introduced in the previous section. In particular, below each theorem we explain the physical meaning of the corresponding property of $\mathcal{C}_M(\Lambda)$ and its relevance for chemical applications. These structural properties are important both for theoretical consistency and for chemical interpretability, as they ensure that $\mathcal{C}_M(\Lambda)$ remains robust and meaningful under operations that naturally arise in modelling and experiment.

The next two theorems make no linearity assumption on $\mathcal{G}$. 
\begin{thm}\label{thm:capacity-properties}
Let $\mathcal{H}$ be a finite-dimensional Hilbert space, $M\in\mathcal{B}_{\text{sa}}(\mathcal{H})$, $\mathcal{G}:\mathcal{B}(\mathcal{H})\to\mathcal{B}(\mathcal{H})$ a (possibly non-linear) resource-destroying map, and $\mathcal{C}_M(\Lambda)$ as defined in Eq.~\eqref{eq:capacity-sup}. 
\begin{itemize}
\item[(i)] (Convexity) For any $p_k\geq 0$, $\sum_k p_k=1$, 
\begin{equation}\label{eq:CM-pullback} 
\mathcal{C}_M\left(\sum_{k}p_k \Lambda_k\right)\leq \sum_{k}p_k \mathcal{C}_M(\Lambda_k)\,.
\end{equation}
\item[(ii)] (Lipschitz continuity) For any two linear CPTP maps $\Lambda_1, \Lambda_2$, the resource impact functional $\mathcal{C}_M(\Lambda)$ is Lipschitz continuous,
\begin{equation}\label{eq:Lipschitz}
\left\vert \mathcal{C}_M(\Lambda_1) - \mathcal{C}_M(\Lambda_2)\right\vert \leq K \|\Lambda_1 - \Lambda_2\|_{\diamond}\,,
\end{equation}
with Lipschitz constant $K : = 2\|M\|_{\infty}$. If $\mathcal{G}$ is a linear unital CPTP map, Eq.~\eqref{eq:Lipschitz} holds with $K$ replaced by $K^\prime=\|M\|_{\infty}<K$. 
\item[(iii)] (Seminorm) Let $M_1, M_2\in\mathcal{B}_{\text{sa}}(\mathcal{H})$ and $a, b\in \RR$. Then, 
\begin{equation}
\mathcal{C}_{a M_1+bM_2}(\Lambda)\leq |a|\mathcal{C}_{M_1}(\Lambda) +|b|\mathcal{C}_{M_2}(\Lambda)\,.
\end{equation}
\end{itemize}
\end{thm}
The proof of Theorem \ref{thm:capacity-properties} is shown in Appendix \ref{app:proof-thm-3}.

Convexity of $\mathcal{C}_M(\Lambda)$ implies that mixing channels, i.e., classical randomness in the dynamics, cannot increase the potential advantage provided by a resource. In addition, when optimizing $\mathcal{C}_M(\Lambda)$ over a convex set of admissible quantum channels, the maximum is attained at an extremal element. 
Furthermore, Lipschitz continuity is important as it ensures that $\mathcal{C}_M(\Lambda)$ is insensitive with respect to small implementation errors of the quantum channel $\Lambda$. Furthermore, the map $M\mapsto\mathcal{C}_M(\Lambda)$ is absolutely homogeneous and subadditive as a consequence of (iii) and, thus, a seminorm on the real vector space of Hermitian operators. 

The next theorem collects several data-processing inequalities that hold for an arbitrary resource-destroying map $\mathcal{G}$. We then discuss their implications afterwards and explain the geometric picture underlying $\mathcal{C}_M(\Lambda)$ and its connection to the post-processing bound in Sec.~\ref{sec:CM-geometry}.
\begin{thm}\label{thm:data-processing}
Let $\mathcal{H}$ be a finite-dimensional Hilbert space, $M\in\mathcal{B}_{\text{sa}}(\mathcal{H})$, $\mathcal{G}:\mathcal{B}(\mathcal{H})\to\mathcal{B}(\mathcal{H})$ a (possibly non-linear) resource-destroying map, and $\mathcal{C}_M(\Lambda)$ as defined in Eq.~\eqref{eq:capacity-sup}. 
\begin{itemize}
\item[(i)] (Post-processing inequality) Let $\Lambda_{1/2}: \mathcal{B}(\mathcal{H})\to\mathcal{B}(\mathcal{H})$ be two linear CPTP maps. The resource impact functional $\mathcal{C}_M(\Lambda)$ satisfies a Heisenberg pullback identity
\begin{equation}\label{eq:CM-composition}
\mathcal{C}_M(\Lambda_2\circ \Lambda_1) = \mathcal{C}_{\Lambda_2^\dagger(M)}(\Lambda_1)\,.
\end{equation}
Moreover, the following post-processing bounds hold:
\begin{align}\label{eq:CM-postprocessing}
\mathcal{C}_M(\Lambda_2\circ \Lambda_1) &\leq  \mathcal{C}_M(\Lambda_1) +\mathcal{C}_{\Delta M_{\Lambda_2}}(\Lambda_1)\nonumber\\
&\leq \mathcal{C}_M(\Lambda_1) +  \|\Delta M_{\Lambda_2}\|_{\infty} R_{\mathcal{G}}\,,
\end{align}
where $\Delta M_{\Lambda_2}:= \Lambda_2^\dagger(M)-M$ and the ``resource radius''
\begin{equation}\label{eq:res-radius}
R_{\mathcal{G}} : = \max_{\rho\in\mathcal{D}(\mathcal{H})}\|\rho - \mathcal{G}(\rho)\|_1\,.
\end{equation} 
describes how far a state can be from its resource-destroyed counterpart in trace distance. 
\item[(ii)] (Pre-processing with a $\mathcal{G}$-covariant channel) Let $\Lambda_1$ be a linear CPTP map with $\mathcal{G}\circ \Lambda_1 = \Lambda_1\circ\mathcal{G}$. Then, 
\begin{equation}
\mathcal{C}_M(\Lambda_2\circ \Lambda_1)\leq \mathcal{C}_M(\Lambda_2)\,.
\end{equation}
\item[(iii)] (Classical coarse-graining) Let $\{M_i\}_i$ be a POVM, and $\widetilde{M}_j =\sum_i V_{ji}M_i$ for a column stochastic matrix $V$. Then, for any quantum channel $\Lambda$ and measurement outcome $j$, 
\begin{equation}\label{eq:class-DPI}
\mathcal{C}_{\widetilde{M}_j}(\Lambda)\leq \sum_i V_{ji}\mathcal{C}_{M_i}(\Lambda)\,.
\end{equation}
\end{itemize}
\end{thm}
Theorem \ref{thm:data-processing} is proven in Appendix \ref{app:proof-thm-4}. 

Parts (i)-(iii) in Theorem \ref{thm:data-processing} show how the resource impact functional $\mathcal{C}_M(\Lambda)$ changes under pre-processing or post-processing with a second quantum channel. Eq.~\eqref{eq:CM-composition} says that post-processing by a quantum channel $\Lambda_2$ acts as a readout deformation $M\mapsto \Lambda_2^\dagger(M)$. Then, the first inequality in Eq.~\eqref{eq:CM-postprocessing} allows us to decompose the maximal yield change as described by $\mathcal{C}_{M}(\Lambda)$ into a part which corresponds to initial $M$ and a second term associated with the readout deformation induced by $\Lambda_2$. 
The second term vanishes when $\Lambda_2^\dagger(M)=M$, i.e., for $\|\Delta M_{\Lambda_2}\|_{\infty} =0$, recovering the standard post-processing monotonicity. Furthermore, the second inequality in Eq.~\eqref{eq:CM-postprocessing} shows that small readout deformations as quantified by $\|\Delta M_{\Lambda_2}\|_{\infty}$ cannot lead to large changes in $\mathcal{C}_M(\Lambda)$. 

Property (ii) in Theorem \ref{thm:data-processing} states that free pre-processing with a $\mathcal{G}$-covariant quantum channel cannot increase the maximally possible yield change. Importantly, while every $\mathcal{G}$-covariant map is a free operation for the given resource theory, not every free operation is $\mathcal{G}$-covariant. Consequently, one should not expect a general monotonicity statement under all free operations and (ii) isolates the subset of free operations under which such a data-processing inequality holds.

Classical coarse-graining plays an important role if we would like to apply $\mathcal{C}_M(\Lambda)$ to non-idealised experimental settings, where coarse-graining is naturally incorporated in the readout model \cite{CDGMS10, Hawton10, HSDOK18}. For this, recall that $M$ is an element of a POVM $\{M_i\}_i$. Measuring $\{M_i\}_i$ and then classically coarse-graining the outcomes by a classical channel described by a column-stochastic matrix $V$ (i.e., $V_{ji}\geq 0$ for all $i,j$ and $\sum_j V_{ji}=1$ for each $i$) is equivalent to measuring the POVM $\{\widetilde{M}_j\}_j$ with
\begin{equation}\label{eq:M-coarse-graining}
\widetilde{M}_j = \sum_i V_{ji} M_i\,,
\end{equation}
where $V_{ji}$ is the conditional probability to map the ``ideal'' outcome $i$ to the recorded outcome $j$ \cite{LeppS21,HJN22}, and similar constructions hold for continuous outcome measurements \cite{Haapasalo2015, HMZ16, Beneduci16}.

\subsubsection{Geometric picture\label{sec:CM-geometry}}

Next, we show that the resource impact functional admits a natural geometric interpretation: it is the support function of a convex set and, thus, the post-processing inequality in Eq.~\eqref{eq:CM-postprocessing} becomes the well-known subadditivity of support functions. 
For this, we first define the set
\begin{equation}\label{eq:dev-set}
\mathcal{M}_{\mathcal{G}, \Lambda}:= \left\{\Lambda(\rho) - \Lambda(\mathcal{G}(\rho))\,\middle\vert\, \rho\in\mathcal{D}(\mathcal{H})\right\}
\subset \mathcal{B}_{\mathrm{sa}}(\mathcal{H})
\end{equation}
and introduce the convex hull of the union of $\mathcal{M}_{\mathcal{G}, \Lambda}$ with $-\mathcal{M}_{\mathcal{G}, \Lambda}$ as
\begin{equation}\label{eq:set-tilde-M}
\widetilde{\mathcal{M}}_{\mathcal{G}, \Lambda}
:= \mathrm{conv}\left(\mathcal{M}_{\mathcal{G}, \Lambda} \cup\left(-\mathcal{M}_{\mathcal{G}, \Lambda}\right)\right)\,.
\end{equation}

\begin{cor}\label{cor:capacity-geometry}
Let $M\in\mathcal{B}_{\mathrm{sa}}(\mathcal{H})$. Then the resource impact functional $\mathcal{C}_M(\Lambda)$ is equal to the support function of the convex set $\widetilde{\mathcal{M}}_{\mathcal{G}, \Lambda}$ evaluated at $M$:
\begin{equation}
\mathcal{C}_M(\Lambda)
= \sup_{O\in\widetilde{\mathcal{M}}_{\mathcal{G}, \Lambda}}\langle M, O\rangle_{\mathrm{HS}}
=: h_{\widetilde{\mathcal{M}}_{\mathcal{G}, \Lambda}}(M)\,.
\end{equation}
\end{cor}
The proof of Corollary~\ref{cor:capacity-geometry} is provided in Appendix~\ref{app:proof-cor-5}. 

\begin{figure}[tb]
\includegraphics[width=\linewidth]{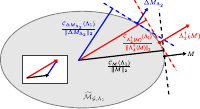}
\caption{Schematic illustration of the geometric picture behind the post-processing bound in Eq.~\eqref{eq:CM-postprocessing} and Corollary~\ref{cor:capacity-geometry}. (See main text for details.)
 \label{fig:support-function}}
\end{figure}

We now describe the geometric picture of the post-processing inequality in Eq.~\eqref{eq:CM-postprocessing}, as illustrated in Fig.~\ref{fig:support-function}. For a non-zero $M$, introduce the unit vector $m:=M/\|M\|_2$
and define the supporting hyperplane
\begin{equation}
H_{\widetilde{\mathcal{M}}_{\mathcal{G}, \Lambda}}(M)
:= \left\{O\in\mathcal{B}_{\mathrm{sa}}(\mathcal{H})\,\middle\vert
\langle M, O\rangle_{\mathrm{HS}}
= h_{\widetilde{\mathcal{M}}_{\mathcal{G}, \Lambda}}(M)\right\}\,.
\end{equation}
Since $\widetilde{\mathcal{M}}_{\mathcal{G},\Lambda}$ is compact and convex, the supremum defining $h_{\widetilde{\mathcal{M}}_{\mathcal{G},\Lambda}}(M)$ is attained, so $H_{\widetilde{\mathcal{M}}_{\mathcal{G}, \Lambda}}(M)$ is indeed a supporting hyperplane.

Then, it follows as a direct consequence of Corollary~\ref{cor:capacity-geometry} that
\begin{align}\label{eq:CM-dist}
\mathcal{C}_M(\Lambda)
&= \|M\|_2\, h_{\widetilde{\mathcal{M}}_{\mathcal{G}, \Lambda}}(m) \\
&= \|M\|_2\,\mathrm{dist}\left(0,H_{\widetilde{\mathcal{M}}_{\mathcal{G}, \Lambda}}(M)\right)\,,\nonumber
\end{align}
where $\mathrm{dist}(\cdot, \cdot)$ denotes the distance induced by the Hilbert-Schmidt inner product. In other words, $\mathcal{C}_M(\Lambda)$ is proportional to the distance from the origin $0$ to the supporting hyperplane $H_{\widetilde{\mathcal{M}}_{\mathcal{G}, \Lambda}}(M)$.

In Fig.~\ref{fig:support-function}, we depict the supporting hyperplanes (dashed lines) corresponding to $\mathcal{C}_M(\Lambda_1)$ (black), $\mathcal{C}_{\Delta M_{\Lambda_2}}(\Lambda_1)$ (blue), and $\mathcal{C}_{\Lambda_2^\dagger(M)}(\Lambda_1)$ (red) appearing in Eq.~\eqref{eq:CM-postprocessing}. Their distances from the origin equal $\mathcal{C}_M(\Lambda_1)/\|M\|_2$, $\mathcal{C}_{\Delta M_{\Lambda_2}}(\Lambda_1)/\|\Delta M_{\Lambda_2}\|_2$, and $\mathcal{C}_{\Lambda_2^\dagger(M)}(\Lambda_1)/\|\Lambda_2^\dagger(M)\|_2$, respectively.
Thus, the subadditivity of support functions \cite{Rockafellar97} then reads
\begin{equation}
h_{\widetilde{\mathcal{M}}_{\mathcal{G}, \Lambda_1}}(\Lambda_2^\dagger(M))\leq h_{\widetilde{\mathcal{M}}_{\mathcal{G}, \Lambda_1}}(M)+h_{\widetilde{\mathcal{M}}_{\mathcal{G}, \Lambda_1}}(\Delta M_{\Lambda_2})\,.
\end{equation}
This is precisely the geometric content of the post-processing inequality \eqref{eq:CM-postprocessing}.

In addition, Corollary \ref{cor:capacity-geometry} also admits a dual formulation in terms of polar sets, which provides further insights into the meaning of $\mathcal{C}_M(\Lambda)$ as we explain below.
In general, given a subset $\mathcal{X}\subset\mathcal{B}_{\mathrm{sa}}(\mathcal{H})$ that contains the origin, its polar is defined as \cite{Barvinok2002}
\begin{equation}\label{eq:polar-def}
\mathcal{X}^{\circ}:=\left\{M\in\mathcal{B}_{\mathrm{sa}}(\mathcal{H})\,\middle\vert\, \langle M,O\rangle_{\mathrm{HS}}\leq 1\,\, \forall\,O\in\mathcal{X}\right\}\,,
\end{equation}
i.e., it is the intersection of the closed half-spaces $\{M:\langle M,O\rangle_{\mathrm{HS}}\leq 1\}$ generated by $O\in\mathcal{X}$.
Since we work in finite-dimensional settings, the polar is closed and convex, even if $\mathcal{X}$ is not \cite{Barvinok2002}.

\begin{cor}[Polar characterization of $\mathcal{C}_M(\Lambda)$] \label{cor:CM-polar}
Let $M\in\mathcal{B}_{\mathrm{sa}}(\mathcal{H})$. Then
\begin{equation}\label{eq:CM-as-gauge}
\mathcal{C}_M(\Lambda)
= h_{\widetilde{\mathcal{M}}_{\mathcal{G},\Lambda}}(M)
= \inf\left\{t\ge 0\,\middle\vert\,
M\in t\,\widetilde{\mathcal{M}}_{\mathcal{G},\Lambda}^{\circ}\right\}.
\end{equation}
Equivalently, the polar set is the unit ball of the functional $M\mapsto \mathcal{C}_M(\Lambda)$:
\begin{equation}\label{eq:polar-unit-ball}
\widetilde{\mathcal{M}}_{\mathcal{G},\Lambda}^{\circ}
= \left\{M\in\mathcal{B}_{\mathrm{sa}}(\mathcal{H})\,\middle\vert\, \mathcal{C}_M(\Lambda)\le 1\right\}.
\end{equation}
\end{cor}
The short proof of Corollary \ref{cor:CM-polar} follows from the definition of the polar set in Eq.~\eqref{eq:polar-def} and is provided in Appendix \ref{app:cor-CM-polar}. 
Thus, Corollary~\ref{cor:CM-polar} yields a complementary, dual interpretation of the resource impact functional:
$\mathcal{C}_M(\Lambda)$ is the smallest factor by which one must rescale the observable $M$ so that it belongs to the polar set
$\widetilde{\mathcal{M}}_{\mathcal{G},\Lambda}^{\circ}$, i.e., to the class of normalized observables whose response to all resource-induced output deviations is uniformly bounded. In other words, $M\in\widetilde{\mathcal{M}}_{\mathcal{G},\Lambda}^{\circ}$ if and only if
\begin{equation}
\left|\Tr \big[M\left(\Lambda(\rho)-\Lambda(\mathcal{G}(\rho))\right)\big]\right|\leq 1
\quad \forall\,\rho\in\mathcal{D}(\mathcal{H})\,.
\end{equation}
Accordingly, elements of $\widetilde{\mathcal{M}}_{\mathcal{G},\Lambda}^{\circ}$ may be viewed as bounded sensitivity probes, while $\mathcal{C}_M(\Lambda)$ quantifies the largest magnitude of the yield change induced by the deviations $\Lambda(\rho)-\Lambda(\mathcal{G}(\rho))$.

Because the definition of $\mathcal{C}_M(\Lambda)$ involves an absolute value, $\mathcal{C}_M(\Lambda)$ also admits a description in terms of the absolute polar
\begin{equation}\label{eq:absolute-polar-def}
\mathcal{M}_{\mathcal{G},\Lambda}^{\circ_a} :=\left\{M\in\mathcal{B}_{\mathrm{sa}}(\mathcal{H})\,\middle\vert\,
\left|\langle M,O\rangle_{\mathrm{HS}}\right|\leq 1\,\, \forall\,O\in\mathcal{M}_{\mathcal{G},\Lambda}\right\}\,.
\end{equation}
Geometrically, $\mathcal{M}_{\mathcal{G},\Lambda}^{\circ_a}$ is the intersection of the slabs
$\{M:|\langle M,O\rangle_{\mathrm{HS}}|\leq 1\}$ induced by $O\in\mathcal{M}_{\mathcal{G},\Lambda}$. This is further reflected by the symmetrization in Eq.~\eqref{eq:set-tilde-M}, which implies that the polar of $\widetilde{\mathcal{M}}_{\mathcal{G},\Lambda}$ equals the absolute polar, $\widetilde{\mathcal{M}}_{\mathcal{G},\Lambda}^{\circ}= \mathcal{M}_{\mathcal{G},\Lambda}^{\circ_a}$
and, thus, Eq.~\eqref{eq:CM-as-gauge} holds also with $\widetilde{\mathcal{M}}_{\mathcal{G},\Lambda}^{\circ}$ replaced by $\mathcal{M}_{\mathcal{G},\Lambda}^{\circ_a}$.

Finally, since $\widetilde{\mathcal{M}}_{\mathcal{G},\Lambda}$ is closed, convex, and contains the origin,
the bipolar theorem \cite{Barvinok2002} implies that it can be recovered from its polar, $(\widetilde{\mathcal{M}}_{\mathcal{G},\Lambda}^{\circ})^{\circ} = \widetilde{\mathcal{M}}_{\mathcal{G},\Lambda}$. Thus, while $\widetilde{\mathcal{M}}_{\mathcal{G},\Lambda}$ and $\widetilde{\mathcal{M}}_{\mathcal{G},\Lambda}^{\circ}$ are equivalent descriptions (via polarity), the polar viewpoint can be more convenient for questions posed in observable space, such as certifying that an experimentally accessible class of measurements is insensitive to the resource-induced deviations, or optimizing $\mathcal{C}_M(\Lambda)$ over a constrained family of observables.

In chemical applications, the deviation set $\mathcal{M}_{\mathcal{G}, \Lambda}$ \eqref{eq:dev-set} collects all output-state differences generated by removing the chosen resource from the input while keeping the same dynamics fixed, and $\widetilde{\mathcal{M}}_{\mathcal{G},\Lambda}$ \eqref{eq:set-tilde-M} is its symmetrized convex hull. Its polar then characterizes those observables whose response to all such deviations is uniformly bounded. This is particularly useful when the accessible measurements are restricted to physically natural classes, such as site populations, trap occupations, or coarse-grained product yields $M=\sum_i w_i\ket{i}\!\bra{i}$, with the appropriate normalization constraint. If an admissible measurement class $\mathfrak{M}$ satisfies $\mathfrak{M}\subseteq t\,\widetilde{\mathcal{M}}_{\mathcal{G},\Lambda}^{\circ}$, then $\mathcal{C}_M(\Lambda)\leq t$ for every $M\in\mathfrak{M}$, so no larger resource-induced response can be observed within that class. Conversely, the support function formulation turns the search for the most sensitive chemically admissible readout into an optimization problem over the allowed observable set. The dual formulation therefore provides a practical tool both for certifying insensitivity of experimentally accessible observables and for optimizing chemically meaningful readouts under measurement constraints.

\subsubsection{Operational interpretation\label{sec:op-interp}}

In this section, we collect several properties of the resource impact functional $\mathcal{C}_M(\Lambda)$ introduced in Secs.~\ref{subsec:capacity} and~\ref{subsec:properties-CM} into a unified operational interpretation. 

Consider a two-outcome POVM $\{M,\mathbbm{1}-M\}$ with $0\leq M\leq \mathbbm{1}$ and the following binary hypothesis testing task. Hypothesis $H_0$ is that the input state before applying $\Lambda$ (and then measuring $M$) was the resource-free state $\mathcal{G}(\rho)$, while hypothesis $H_1$ is that the input state was the potentially resourceful $\rho$, with prior probabilities $p_0=p_1=1/2$. Moreover, we label the outcome associated with $M$ by $X=1$ and the complementary outcome by $X=0$. Then, for a given pair $(\rho,\mathcal{G}(\rho))$,
\begin{align}
P(X=1|H_0) &= p_0(\rho) := \Tr[M\Lambda(\mathcal{G}(\rho))]\,,\\
P(X=1|H_1) &= p_1(\rho) := \Tr[M\Lambda(\rho)]\,,
\end{align}
such that
\begin{equation}
p_1(\rho) - p_0(\rho) = \Delta Y(\rho)
\end{equation}
with $\Delta Y(\rho)$ as in Eq.~\eqref{eq:yield-change}. Since $M$ is a POVM element in this section, $\mathcal{C}_M(\Lambda)\in [0, 1]$. This contrasts with the more general case where $M$ is an arbitrary observable, for which $\mathcal{C}_M(\Lambda)$ can exceed $1$. 

For fixed $\rho$, and given the measurement $\{M,\mathbbm{1}-M\}$, the optimal single-shot decision rule (for equal priors) is to decide $H_1$ on the outcome that is more likely under $H_1$ than under $H_0$ and $H_0$ otherwise. The corresponding success probability follows as
\begin{equation}
p_{\text{succ}}^{(1)}(\rho)
= \frac{1}{2} + \frac{1}{2}\vert\Delta Y(\rho)\vert\,.
\end{equation}
Thus, the maximal bias above random guessing in distinguishing whether the input state was $\rho$ or its free counterpart $\mathcal{G}(\rho)$ is $|\Delta Y(\rho)|/2$. Optimizing over all input states,
\begin{equation}
p_{\text{succ,max}}^{(1)} = \sup_{\rho\in\mathcal{D}(\mathcal{H})} p_{\text{succ}}^{(1)}(\rho) = \frac{1}{2} + \frac{1}{2}\mathcal{C}_M(\Lambda)\,.
\end{equation}
Equivalently, $\mathcal{C}_M(\Lambda)$ equals twice the maximal improvement over random guessing in distinguishing a resourceful input from its paired free counterpart solely from the statistics of measuring $M$ after applying the channel $\Lambda$.

This leads to a simple dichotomy. First, $\mathcal{C}_M(\Lambda)=0$ if and only if $p_{\text{succ}}^{(1)}(\rho)=1/2$ for all $\rho\in\mathcal{D}(\mathcal{H})$, i.e., independent of the input state, $\Lambda$ followed by $M$ cannot distinguish $\rho$ from $\mathcal{G}(\rho)$. Second, $\mathcal{C}_M(\Lambda)>0$ whenever there exists a $\rho$ such that $p_{\text{succ}}^{(1)}(\rho)>1/2$, i.e., at least one state for which the resource makes a statistically detectable difference in the readout associated with $M$. In this sense, the resource is then operationally relevant for the task defined by $(\Lambda,M)$, in line with the discussion in Sec.~\ref{subsec:capacity}. To summarize, $\mathcal{C}_M(\Lambda)$ answers the binary question of whether a resource under the dynamics can affect the chosen measurement at all, and at the same time quantifies how strongly it can do so.

For $\mathcal{G}$ linear, Theorem~\ref{thm:capacity-G-linear} shows that
\begin{equation}
\mathcal{C}_M(\Lambda) = \|B_{M,\Lambda}\|_{\infty}\,,\quad B_{M,\Lambda} := (\mathrm{id}-\mathcal{G}^\dagger)(\Lambda^\dagger(M))\,.
\end{equation}
The states that are optimal for the single-shot hypothesis test are precisely the eigenstates of $B_{M,\Lambda}$ associated with an extremal eigenvalue (of largest magnitude). Accordingly, for a given triple $(\Lambda,M,\mathcal{G})$, reconstructing (or approximating) $B_{M,\Lambda}$ identifies which state preparations are maximally sensitive to the resource for that specific process and readout. This operational viewpoint is equivalent to the geometric picture in Sec.~\ref{sec:CM-geometry}: fixing $M$ selects a direction in operator space, and the optimal eigenstates correspond to preparations whose deviation $O=\Lambda(\rho)-\Lambda(\mathcal{G}(\rho))$ lies on the supporting face of $\widetilde{\mathcal{M}}_{\mathcal{G},\Lambda}$ in that direction. 

We can lift this discussion to a multi-shot setting, which shows that $\mathcal{C}_M(\Lambda)$ not only determines whether a resource can affect a task at all, but also sets a natural scale for the statistical rate at which repeated experiments can reveal that effect. 
To this end, consider $n$ independent and identically distributed (i.i.d.) repetitions of the above single-shot experiment with the same input state $\rho$. This allows us to upper bound the total error probability $P_{\text{err}}^{(n)}(\rho)$, which is the average probability of making an error over the two hypotheses $H_0$ and $H_1$ when occurring with equal probability in terms of the yield change $|\Delta Y(\rho)|$, as shown in Appendix \ref{app:CM-operat-int}. 

To connect $P_{\text{err}}^{(n)}(\rho)$ to the resource impact functional $\mathcal{C}_M(\Lambda)$, we consider a state $\rho^\ast$ that (approximately) maximizes $|\Delta Y(\rho)|$. In fact, for any $\varepsilon>0$ there exists a $\rho^\ast$ with
\begin{equation}
|\Delta Y(\rho^\ast)| \geq \mathcal{C}_M(\Lambda) - \varepsilon\,.
\end{equation}
For this state, the total error probability is then bounded by (recall Eq.~\eqref{eq:Perr-app} in Appendix \ref{app:CM-operat-int})
\begin{equation}
P_{\text{err}}^{(n)}(\rho^\ast)\leq \exp\left(-\frac{n}{2}\big(\mathcal{C}_M(\Lambda)-\varepsilon\big)^2\right)\,.
\end{equation}
Thus, whenever $\mathcal{C}_M(\Lambda)>0$, there exist preparations $\rho^\ast$ for which the error probability decays exponentially fast in $n$, with an exponent that can be made arbitrarily close to $\mathcal{C}_M(\Lambda)^2/2$. In contrast, if $\mathcal{C}_M(\Lambda)=0$, then $\Delta Y(\rho)=0$ for all $\rho$, so that $p_0(\rho)=p_1(\rho)$ and the outcome distributions under $H_0$ and $H_1$ coincide. In that case no decision rule can outperform random guessing, and hence $P_{\text{err}}^{(n)}(\rho)=1/2$ for all $n$ and all $\rho$.

\subsection{Instantaneous rate\label{subsec:dynamics}}

In the following, we introduce a dynamical version of the resource impact functional in order to quantify whether a given resource can provide a fixed target yield improvement within a given time window.

Let $\{\Lambda_t\}_{t\geq 0}$ be a family of dynamical maps. At each fixed time $t$, the corresponding resource impact functional is defined exactly as in Eq.~\eqref{eq:capacity-sup}, with $\Lambda$ replaced by $\Lambda_t$. This yields the time-resolved quantity $t\mapsto \mathcal{C}_M(\Lambda_t)$ and requires no additional structural assumption, such as a time-local master equation. Accordingly, as for the distinction between general resource-destroying maps $\mathcal{G}$ and linear $\mathcal{G}$ in Sec.~\ref{subsec:capacity}, we first formulate the results at the level of general dynamical maps and only afterwards specialize to time-local master equations where this is useful. In particular, the resource impact functional remains well defined for non-Markovian reduced dynamics at each fixed time directly in terms of the map $\Lambda_t$.

To define an instantaneous rate, one needs temporal differentiability. Assume therefore that $t\mapsto \Lambda_t$ is differentiable, or more generally absolutely continuous so that the derivative exists almost everywhere. Then the time derivative of the yield change $\Delta Y(\rho;t)$ from Eq.~\eqref{eq:yield-change} is given by
\begin{align}
\frac{\mathrm d}{\mathrm dt}\Delta Y(\rho;t)&= \Tr \left[M   \dot{\Lambda}_t \left(\rho - \mathcal{G}(\rho)\right)\right]\nonumber\\
&= \Tr \left[\dot{\Lambda}_t^\dagger(M) \left(\rho-\mathcal{G}(\rho)\right)\right]\,,
\end{align}
where $\dot{\Lambda}_t := \mathrm d\Lambda_t/\mathrm dt$.
This motivates the general instantaneous rate
\begin{align}\label{eq:Gamma-M-sup-general}
\Gamma_M(t) := \sup_{\rho\in\mathcal{D}(\mathcal{H})} \left| \Tr\left[M \dot{\Lambda}_t \left(\rho-\mathcal{G}(\rho)\right)\right] \right|\,,
\end{align}
which quantifies the largest possible instantaneous change of the resource-induced yield difference associated with $M$.
If $\mathcal{G}$ is linear and Hermiticity-preserving, the same duality argument as in the static case gives
\begin{equation}
\Gamma_M(t) = \left\| (\mathrm{id} - \mathcal{G}^\dagger) \left(\dot{\Lambda}_t^\dagger(M)\right) \right\|_\infty\,.
\end{equation}

A particularly important special case is that of a time-local master equation. If $\Lambda_t$ admits a possibly time-dependent generator $\mathcal{L}_t$ such that $\tfrac{\mathrm d}{\mathrm dt}\Lambda_t = \mathcal{L}_t\circ \Lambda_t$, 
then the adjoint family satisfies $\tfrac{\mathrm{d}}{\mathrm{d}t}\Lambda_t^\dagger = \Lambda_t^\dagger\circ \mathcal{L}_t^\dagger$ (see Appendix~\ref{app:time-local} for a brief recap and Refs.~\cite{HCLA14, SLHK19}).
In this case, Eq.~\eqref{eq:Gamma-M-sup-general} becomes
\begin{align}\label{eq:Gamma-M-sup}
\Gamma_M(t) &= \sup_{\rho\in\mathcal{D}(\mathcal{H})} \left| \Tr \left[ M \mathcal{L}_t \left(\Lambda_t \left(\rho - \mathcal{G} (\rho)\right)\right) \right] \right| \\
&= \left\| (\mathrm{id}-\mathcal{G}^\dagger) \left(\Lambda_t^\dagger(\mathcal{L}_t^\dagger(M))\right) \right\|_\infty\,,
\end{align}
where the second equality holds only for $\mathcal{G}$ linear and Hermiticity-preserving.
Note that $\Gamma_M(t)$ is not, in general, the same as the time derivative of $\mathcal{C}_{M}(\Lambda_t)$ since the absolute value and the operator norm are in general non-smooth and, thus, $\mathcal{C}_{M}(\Lambda_t)$ may not be everywhere differentiable. In addition it may not equal $\Gamma_M(t)$ when its time derivative exists. 
Still, we can derive a variation bound for $\mathcal{C}_M(\Lambda_t)$:

\begin{thm}[Variation bound and Dini derivatives]\label{thm:variation-bound}
Let $\mathcal{H}$ be a finite dimensional Hilbert space, let $M\in\mathcal{B}_{\mathrm{sa}}(\mathcal{H})$, and let
$\mathcal{G}:\mathcal{D}(\mathcal{H})\to\mathcal{D}(\mathcal{H})$ be a (possibly non-linear) resource-destroying map.
Let $\mathcal{I}\subset\RR$ be an interval and let $\{\Lambda_t\}_{t\in\mathcal{I}}$ be a family of CPTP maps on $\mathcal{B}(\mathcal{H})$.
Assume that there exists a strongly measurable family of bounded linear maps
$\{\dot\Lambda_t\}_{t\in\mathcal{I}}$ such that $t\mapsto \|\dot\Lambda_t\|_{1\to 1}$ is locally integrable on $\mathcal{I}$ and
\begin{equation}\label{eq:Lambda-integral}
\Lambda_t=\Lambda_{t_0}+\int_{t_0}^t \dot\Lambda_s\,\mathrm ds \quad \forall\, t,t_0\in\mathcal{I}\,.
\end{equation}
Define $\mathcal{C}_M(\Lambda_t)$ and $\Gamma_M(t)$ as in Eqs.~\eqref{eq:capacity-sup} and \eqref{eq:Gamma-M-sup-general}, respectively.
Then for all $t_1\leq t_2$ in $\mathcal{I}$,
\begin{equation}\label{eq:variation-bound}
\big|\mathcal{C}_M(\Lambda_{t_2})-\mathcal{C}_M(\Lambda_{t_1})\big| \leq \int_{t_1}^{t_2}\Gamma_M(s)\,\mathrm{d}s\,.
\end{equation}
In particular, $t\mapsto \mathcal{C}_M(\Lambda_t)$ is absolutely continuous on compact subintervals of $\mathcal{I}$, and its right Dini derivatives satisfy
\begin{equation}\label{eq:Dini}
D^+\mathcal{C}_M(\Lambda_t)\le \Gamma_M(t)\,, \quad D_+\mathcal{C}_M(\Lambda_t)\geq -\Gamma_M(t)
\end{equation}
for Lebesgue almost every $t\in\mathcal{I}$.
The analogous left-sided bounds hold as well.
If $t\mapsto \Gamma_M(t)$ is continuous, then Eq.~\eqref{eq:Dini} holds for all $t\in\mathcal{I}$.
\end{thm}
We prove Theorem \ref{thm:variation-bound} in Appendix \ref{app:proof-thm-6}. For chemical applications, Theorem \ref{thm:variation-bound} bounds how rapidly the maximal possible resource-induced contribution to a chosen readout can change over a given time window. For example, if $M$ represents an acceptor population, trap occupation, or product-yield observable, and $\mathcal{G}$ removes site-basis coherence, then $\mathcal{C}_M(\Lambda_t)$ is the largest coherence-induced change that any initial preparation can produce in that observable at time $t$. The theorem shows that this maximal coherence contribution cannot vary between $t_1$ and $t_2$ by more than $\int_{t_1}^{t_2}\Gamma_M(s)\,\mathrm ds$. In practice, this provides a state-independent criterion for analysing whether a coherence-assisted effect could emerge or disappear on the timescale of transfer, trapping, or recombination.

Moreover, at any time $t$ where $\Gamma_M(t)$ is finite and continuous, Theorem \ref{thm:variation-bound} implies that all four one-sided Dini derivatives of $t\mapsto\mathcal{C}_M(\Lambda_t)$ are finite and all lie in the interval $[-\Gamma_M(t),\Gamma_M(t)]$.
If at some $t_0$ the upper right and upper left Dini derivatives differ, $D^{+}\mathcal{C}_M(\Lambda_{t_0})\neq D^{-}\mathcal{C}_M(\Lambda_{t_0})$, then $\mathcal{C}_M(\Lambda_t)$ is not differentiable at $t_0$, since an ordinary derivative (if it existed) would have to coincide with all four one-sided Dini derivatives.
Moreover, if $t_0$ is a local maximum of $\mathcal{C}_M(\Lambda_t)$, then the one-sided Dini derivatives obey
\begin{align}
D^{+}\mathcal{C}_M(\Lambda_{t_0})&\leq  0\leq D^{-}\mathcal{C}_M(\Lambda_{t_0})\,,
\end{align}
and similarly for $D_{+}\mathcal{C}_M(\Lambda_{t_0}), D_{-}\mathcal{C}_M(\Lambda_{t_0})$, even if $\mathcal{C}_M(\Lambda_t)$ is not differentiable at $t_0$. In this sense, the Dini derivatives capture the behaviour of the maximal advantage described by $\mathcal{C}_M(\Lambda_t)$ around local extrema, independent of differentiability.

Furthermore, a vanishing $\Gamma_M(t)$ over a given time interval $\mathcal{I}\in [t_1, t_2]$ implies that $\mathcal{C}_M(\Lambda_t)$ stays constant within $\mathcal{I}$. The condition under which $\Gamma_M(t)=0$ follows immediately from the definition of $\Gamma_M(t)$ in Eq.~\eqref{eq:Gamma-M-sup} and is stated in the following Lemma:
\begin{lem}\label{lem:Gamma-M-zero}
Let $M\in\mathcal{B}_{\mathrm{sa}}(\mathcal{H})$, let $\mathcal{G}$ be a resource-destroying map, and let $\Gamma_M(t)$ be defined as in Eq.~\eqref{eq:Gamma-M-sup}. Then, for any $t$,
\begin{equation}\label{eq:cond-Gamma-0}
 \Gamma_M(t)=0 \quad\Leftrightarrow\quad \dot{\Lambda}_t^\dagger (M)\in \mathcal{V}_{\mathcal{G}}^\perp\,,
\end{equation}
where 
\begin{equation}
\mathcal{V}_{\mathcal{G}}^\perp: = \left\{O\in\mathcal{B}_{\mathrm{sa}}(\mathcal{H})\,\middle\vert \,\langle O, X\rangle_{\mathrm{HS}} =0\quad \forall X \in \mathcal{V}_{\mathcal{G}}\right\}
\end{equation}
with 
\begin{equation}
\mathcal{V}_{\mathcal{G}}: = \mathrm{span}_{\RR}\left(\left\{\rho - \mathcal{G}(\rho)\,\middle\vert\, \rho\in\mathcal{D}(\mathcal{H})\right\}\right)\subset \mathcal{B}_{\mathrm{sa}}(\mathcal{H})\,.
\end{equation}
\end{lem}
We provide the proof of Lemma \ref{lem:Gamma-M-zero} in Appendix \ref{app:proof-lem-7}. 
Consequently, in case Eq.~\eqref{eq:cond-Gamma-0} is satisfied for all $t\in [t_1, t_2]$, the variation bound in Eq.~\eqref{eq:variation-bound} enforces $|\mathcal{C}_M(\Lambda_{t_2}) - \mathcal{C}_M(\Lambda_{t_1})|=0$. 

In summary, Eq.~\eqref{eq:variation-bound} shows that the change of $\mathcal{C}_M(\Lambda_t)$ between times $t_1$ and $t_2$ is bounded by the time integral of the instantaneous rate $\Gamma_M(t)$. In other words, the maximal resource induced advantage encoded in $\mathcal{C}_M(\Lambda_t)$ cannot increase or decrease faster than allowed by $\Gamma_M(t)$.
We will show in the next section that this leads to a \textit{resource-aware quantum speed limit} on how fast the maximal resource based advantage for the task $M$ can vary. In contrast to standard quantum speed limits \cite{TEDMF13, dCEPH13, MSZ16, DC17}, which bound changes in terms of distances in state space, our bound is explicitly task and resource dependent as it is written directly in terms of the target yield associated with $M$, and free contributions are filtered out through the resource-destroying map $\mathcal{G}$.

\subsection{Time and feasibility bounds\label{subsec:time-bound}}

We next show that Theorem~\ref{thm:variation-bound} yields explicit lower bounds on the time required to generate a prescribed change in the task specific quantity $\mathcal{C}_M(\Lambda_t)$, that is, in the maximal resource-induced modification of the yield associated with $M$. Conversely, for a fixed time budget $T_{\max}$, it gives a feasibility bound on the largest possible change in $\mathcal{C}_M(\Lambda_t)$. These two complementary statements are summarized in the following corollary.
\begin{cor}[Time and feasibility bounds]\label{cor:time-bound}
Let the assumptions of Theorem~\ref{thm:variation-bound} hold, and let
\begin{equation}
\Delta\mathcal{C}_M(\Lambda_{t_2},\Lambda_{t_1}) := \mathcal{C}_M(\Lambda_{t_2}) - \mathcal{C}_M(\Lambda_{t_1})
\end{equation}
for $t_1\leq t_2$ in $\mathcal{I}$.
Define the resource radius $R_{\mathcal{G}}$ as in Eq.~\eqref{eq:res-radius} and
\begin{equation}
c_{M,\mathcal{G}} := \|M\|_\infty R_{\mathcal{G}} \leq 2\|M\|_\infty\,.
\end{equation}
Then, for every interval $[t_1,t_2]\subset\mathcal{I}$,
\begin{equation}\label{eq:uniform-bound}
\left|\Delta\mathcal{C}_M(\Lambda_{t_2},\Lambda_{t_1})\right| \leq
c_{M,\mathcal{G}}\,\Delta\tau\, \operatorname*{ess\,sup}_{s\in [t_1,t_2]}\|\dot{\Lambda}_s\|_{1\to 1}\,,
\end{equation}
where $\Delta\tau := t_2-t_1$. If $\dot\Lambda_t$ admits a bounded representative on $[t_1,t_2]$, one may equivalently replace the essential supremum by the ordinary supremum.
Moreover, if $\|\dot{\Lambda}_s\|_{1\to 1}\leq K_{\max}$ for Lebesgue almost every $t\in\mathcal{I}$,
then any evolution achieving a non-zero change $\Delta\mathcal{C}_M(\Lambda_{t_2},\Lambda_{t_1})$ must satisfy the minimal-time constraint
\begin{equation}\label{eq:time-QSL}
\Delta\tau \geq
\frac{\left|\Delta\mathcal{C}_M(\Lambda_{t_2},\Lambda_{t_1})\right|}{c_{M,\mathcal{G}} K_{\max}} \quad (K_{\max}>0)\,.
\end{equation}
Conversely, given a fixed maximal time budget $T_{\max}<\infty$ and a uniform bound $\|\dot{\Lambda}_s\|_{1\to 1}\leq K_{\max}$ for almost every $s\in[t_1,t_1+T_{\max}]$, the achievable change in the advantage is bounded by
\begin{equation}\label{eq:feasibility-bound}
\left|\Delta\mathcal{C}_M(\Lambda_{t_2}, \Lambda_{t_1})\right| \leq T_{\max}\,c_{M,\mathcal{G}} K_{\max}\,,
\end{equation}
where $0\leq t_2-t_1\leq T_{\max}$.
If, in addition, $\Lambda_t$ is generated by a time-local master equation
$\dot{\Lambda}_t=\mathcal{L}_t\circ\Lambda_t$ , then
\begin{equation}
\|\dot{\Lambda}_t\|_{1\to 1} \leq\|\mathcal{L}_t\|_{1\to 1}\,.
\end{equation}
Therefore Eqs.~\eqref{eq:time-QSL} and \eqref{eq:feasibility-bound} reduce to the previous generator-based form with $K_{\max}$ replaced by any uniform bound on $\|\mathcal{L}_t\|_{1\to 1}$.
\end{cor}
The proof of Corollary \ref{cor:time-bound} is provided in Appendix \ref{app:proof-cor-8}. For chemical applications, Corollary \ref{cor:time-bound} turns the abstract variation bound into operational constraints on both the timescale and the magnitude of a possible resource-induced modification of a measurable outcome. For example, $M$ may represent an observable associated with trap population, acceptor occupation, or product yield. Under the stated uniform bound, Eq.~\eqref{eq:time-QSL} provides a lower bound on how long any interval must be in order for the maximal resource-induced change to reach a prescribed value, while Eq.~\eqref{eq:feasibility-bound} bounds the largest such change achievable within a finite experimental or dynamical time window. In this way, the corollary links resource-theoretic quantities to chemically relevant questions about whether the maximal possible coherence-induced, or more generally nonclassical, effect can become significant before relaxation, trapping, or recombination takes over.

\begin{figure}[tb]
\centering
\includegraphics[width=0.85\linewidth]{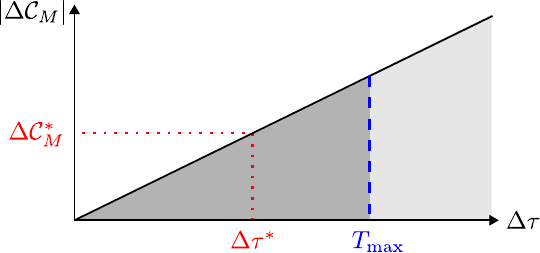}
\caption{Schematic illustration of the uniform time and feasibility bounds in Corollary~\ref{cor:time-bound}. The horizontal axis shows the evolution time $\Delta\tau$ and the vertical axis the change $|\Delta\mathcal{C}_M|$. For a uniform generator bound $\|\dot{\Lambda}_s\|_{1\to 1}\leq K_{\max}$, the light and dark gray regions below the line $|\Delta\mathcal{C}_M| = c_{M,\mathcal{G}} K_{\max}\Delta\tau$ contain all admissible pairs. A time budget $T_{\max}$ limits the achievable change to $|\Delta\mathcal{C}_M|\leq T_{\max}c_{M,\mathcal{G}} K_{\max}$ (dark gray), while a target change $\Delta\mathcal{C}_M^*$ defines a minimal time $\Delta\tau^*$ via Eq.~\eqref{eq:min-time}.
 \label{fig:time-bound}}
\end{figure}

Suppose we want to achieve a prescribed change $\Delta\mathcal{C}_M^*\geq 0$ of $\mathcal{C}_M(\Lambda_t)$.
We define the corresponding minimal time as
\begin{align}\label{eq:min-time}
\Delta\tau^*(\Delta\mathcal{C}_M^*) &:= \inf\big\{t_2-t_1\geq 0 \,\big|\, t_1,t_2\in\mathcal{I}, \nonumber\\
&\qquad\qquad\ \left|\Delta\mathcal{C}_M(\Lambda_{t_2},\Lambda_{t_1})\right| \geq \Delta\mathcal{C}_M^* \big\}\,,
\end{align}
with the convention $\Delta\tau^*(\Delta\mathcal{C}_M^*)=\infty$ if the set is empty.
In words, $\Delta\tau^*(\Delta\mathcal{C}_M^*)$ is the shortest time interval over which the given dynamics can produce a change in $\mathcal{C}_M$ of at least $\Delta\mathcal{C}_M^*$.
The uniform bound in Eq.~\eqref{eq:uniform-bound} then implies a \textit{resource-aware quantum speed limit} in the sense that, for any interval $[t_1,t_2]$ with $\Delta\tau:=t_2-t_1>0$,
\begin{equation}
\frac{\left|\Delta\mathcal{C}_M(\Lambda_{t_2},\Lambda_{t_1})\right|}{\Delta\tau} \leq c_{M,\mathcal{G}}\,\operatorname*{ess\,sup}_{s\in [t_1,t_2]}\|\dot{\Lambda}_s\|_{1\to 1}\,.
\end{equation}
Thus, the average rate of change of $\mathcal{C}_M$ is bounded from above by a quantity that depends only on the task specified by $(M, \mathcal{G})$ and on the family of dynamical maps.
Equivalently, if $\|\dot{\Lambda}_s\|_{1\to 1}\leq K_{\max}$ for almost every $s\in\mathcal{I}$, then the minimal time satisfies
\begin{equation}
\Delta\tau^*(\Delta\mathcal{C}_M^*) \geq \frac{\Delta\mathcal{C}_M^*}{c_{M,\mathcal{G}} K_{\max}}\,.
\end{equation}
Conversely, for a given maximal time budget $T_{\max}$ and the same bound on $\|\dot{\Lambda}_s\|_{1\to 1}$, Corollary~\ref{cor:time-bound} implies
\begin{equation}\label{eq:LCM-uni}
\left|\Delta\mathcal{C}_M(\Lambda_{t_2},\Lambda_{t_1})\right| \leq T_{\max}c_{M,\mathcal{G}}K_{\max}
\end{equation}
for $0\leq t_2-t_1\leq T_{\max}$.
Hence $T_{\max}$ and $\Delta\tau^*$ play dual roles: $T_{\max}$ constrains how much $\mathcal{C}_M$ can change, whereas $\Delta\tau^*$ constrains how fast a prescribed change in $\mathcal{C}_M$ can be achieved.

At the same time, Eq.~\eqref{eq:uniform-bound} is model independent and relies on two worst-case estimates.
First, one replaces the exact instantaneous rate $\Gamma_M(s)$ by the upper bound $\Gamma_M(s)\leq c_{M,\mathcal{G}}\|\dot\Lambda_s\|_{1\to 1}$.
Second, one bounds the time integral by the essential supremum over the interval.
The resulting bound can therefore be loose, for example when $\Gamma_M(s)$ is much smaller than $c_{M,\mathcal{G}}\|\dot\Lambda_s\|_{1\to1}$ for most times, or when $\|\dot\Lambda_s\|_{1\to1}$ is sharply peaked on a small subset of the interval.
Tighter, model-dependent bounds can be obtained by working directly with the integral $\int_{t_1}^{t_2}\Gamma_M(s)\,\mathrm ds$ (see Sec.~\ref{sec:DAM} for an analytical example).
In the time-local case $\dot\Lambda_s =  \mathcal{L}_s\circ\Lambda_s$ almost everywhere, one may further use $\|\dot\Lambda_s\|_{1\to 1}\leq \|\mathcal{L}_s\|_{1\to1}$.

We illustrate the simple concept underlying the uniform time and feasibility bounds in Fig.~\ref{fig:time-bound}. The horizontal axis shows the time interval $\Delta\tau = t_2 - t_1$ and the vertical axis the achievable change $|\Delta\mathcal{C}_M(\Lambda_{t_2},\Lambda_{t_1})|$. The solid line represents the uniform bound for a dynamical map satisfying $\|\dot{\Lambda}_s\|_{1\to 1}\leq K_{\max}$ as in Eq.~\eqref{eq:LCM-uni}. The light-gray region below this line indicates the admissible pairs $(\Delta\tau,|\Delta\mathcal{C}_M|)$. A finite time budget $T_{\max}$ imposes the additional vertical cut-off at $\Delta\tau = T_{\max}$, such that that the dark-gray rectangle highlights the maximal achievable change $|\Delta\mathcal{C}_M|\leq T_{\max} c_{M,\mathcal{G}} K_{\max}$. Conversely, for a fixed target change $\Delta\mathcal{C}_M^*$ (recall Eq.~\eqref{eq:min-time}), the intersection with the bound line defines the minimal required time $\Delta\tau^*$, illustrating the resource-aware speed limit.

A common form of the Lindblad generator $\mathcal{L}_t$ is the Gorini-Kossakowski-Sudarshan-Lindblad (GKSL) form (see Eq.~\eqref{eq:time-local-master}). In that case, we can use
\begin{equation}
\Gamma_M(t)\leq R_{\mathcal{G}}\|\mathcal{L}_t^\dagger (M)\|_{\infty}
\end{equation}
to obtain an explicit bound on $\Delta\tau^*$ in terms of the Hamiltonian $\widetilde{H}(t) = H_S(t)+H_L(t)$, which consists of the system's Hamiltonian $H_S$ and a Lamb shift contribution $H_L$, and the Lindblad jump operators $L_{k}(t)$. 
Moreover, the Heisenberg adjoint of $\mathcal{L}_t$ reads 
\begin{align}
\mathcal{L}_t^\dagger(M) &= i[\widetilde{H}(t), M]+\sum_k\big[L_k^\dagger(t)ML_k(t)\nonumber\\
&\quad - \frac{1}{2}\{L_k^\dagger(t)L_k(t), M\}\big]
\end{align}
such that
\begin{equation}
\Gamma_M(t)\leq 2R_{\mathcal{G}}\|M\|_{\infty}\left(\|\widetilde{H}(t)\|_{\infty} + \sum_k\|L_k(t)\|_{\infty}^2\right)\,.
\end{equation}
Together with the variation bound Eq.~\eqref{eq:variation-bound} this implies that increasing $\Delta\mathcal{C}_M(\Lambda_{t_2}, \Lambda_{t_1})$ (or, equivalently, lowering the lower bound on $\Delta\tau$) requires either allocating more time or increasing the Hamiltonian or dissipator strengths.

\subsection{Dissecting dynamics into free and resourceful components \label{sec:dissect-dyn}}

In this section, we eventually show that for linear, Hermiticity-preserving resource-destroying maps $\mathcal{G}$, any channel (and, in the setting of time-local master equations, also its generator) admits a canonical decomposition into a $\mathcal{G}$-commuting ``free'' part and a ``resource'' part. This decomposition is useful for chemical applications because it separates a resource-neutral background from the part of the dynamics that can generate a resource-induced change in the chosen yield. The following theorem explains how we determine the free and resource parts and shows that only the resource part can affect $\mathcal{C}_M(\Lambda_t)$.

\begin{thm}[Decomposition with respect to a resource-destroying map]\label{thm:dissect-channel}
Let $\mathcal{H}$ be a finite-dimensional Hilbert space, and let
$\mathcal{G}:\mathcal{B}(\mathcal{H})\to\mathcal{B}(\mathcal{H})$ be a linear idempotent map.
Let $\{\Lambda_t\}_{t\in[0,T]}$ be a family of CPTP maps with $\Lambda_0=\mathrm{id}$.

\begin{itemize}
\item[(i)] For each $t \in[0,T]$ there is a canonical decomposition
\begin{equation}\label{eq:Lambda-t-dec}
\Lambda_t=\Lambda_{t,\mathrm{free}}+\Lambda_{t,\mathrm{res}},
\end{equation}
defined by
\begin{equation}\label{eq:channel-decomp}
\Lambda_{t,\mathrm{free}}:=\mathcal{G}\circ\Lambda_t\circ\mathcal{G}\,, \quad \Lambda_{t,\mathrm{res}}:=\Lambda_t-\Lambda_{t,\mathrm{free}}\,.
\end{equation}
This decomposition is uniquely determined by $\mathcal{G}$ and satisfies
\begin{align}
&\Lambda_{t,\mathrm{free}} = \mathcal{G}\circ \Lambda_{t,\mathrm{free}} = \Lambda_{t,\mathrm{free}}\circ \mathcal{G}\,,  \label{eq:decomp-properties-1}\\
&\mathcal{G}\circ \Lambda_{t,\mathrm{res}}\circ \mathcal{G} =0\,.\label{eq:decomp-properties-2}
\end{align}

Moreover, if $\{\Lambda_t\}$ satisfies a time-local master equation $\dot\Lambda_t=\mathcal{L}_t\circ \Lambda_t$ for some family of linear maps $\{\mathcal{L}_t\}_{t\in [0, T]}$, then one may likewise define
\begin{equation}\label{eq:Lt-dec}
\mathcal{L}_t=\mathcal{L}_{t,\mathrm{free}}+\mathcal{L}_{t,\mathrm{res}}
\end{equation}
by
\begin{equation}\label{eq:generator-decomp}
\mathcal{L}_{t,\mathrm{free}}:=\mathcal{G}\circ \mathcal{L}_t\circ \mathcal{G}\,, \quad \mathcal{L}_{t,\mathrm{res}}:=\mathcal{L}_t-\mathcal{L}_{t,\mathrm{free}}\,,
\end{equation}
with
\begin{align}
&\mathcal{L}_{t,\mathrm{free}} = \mathcal{G}\circ \mathcal{L}_{t,\mathrm{free}} = \mathcal{L}_{t,\mathrm{free}}\circ \mathcal{G}\,, \quad \\
&\mathcal{G}\circ \mathcal{L}_{t,\mathrm{res}}\circ\mathcal{G}=0\,.
\end{align}
\item[(ii)] For any observable $M\in\mathcal{B}_{\mathrm{sa}}(\mathcal{H})$, the resource induced advantage depends only on the resource part of the channel, in the sense that
\begin{equation}
\mathcal{C}_M(\Lambda_t)= \sup_{\rho\in\mathcal{D}(\mathcal{H})} \left|\Tr\left[M\,\Lambda_t(\rho-\mathcal{G}(\rho))\right]\right|= \mathcal{C}_M(\Lambda_{t,\mathrm{res}})\,.
\end{equation}
\item[(iii)] If $\mathcal{G}$ is CPTP, then the free part $\Lambda_{t,\mathrm{free}}$ is CPTP for every $t$.
\end{itemize}
\end{thm}
We provide the proof of Theorem \ref{thm:dissect-channel} in Appendix \ref{app:proof-thm-9}. Theorem~\ref{thm:dissect-channel} is useful for chemical applications because it isolates the part of the dynamics that can actually influence a resource-sensitive readout. In complex molecular models, where coherent evolution, relaxation, trapping, dephasing, or other processes act simultaneously, it is often not a priori clear which contributions are relevant for the observed resource-induced change in the yield. The decomposition above provides a systematic way to separate a resource neutral background from the genuinely resource active contribution. This is valuable both conceptually, because it sharpens the interpretation of $\mathcal{C}_M(\Lambda_t)$, and practically, because it allows one to focus the analysis on the part of the dynamics that can matter for the chosen resource and readout. In the time-local setting, the corresponding decomposition of the generator is particularly helpful because it allows one to inspect directly which terms in the master equation are capable of driving resource-sensitive changes in the yield and which terms remain invisible at the level of $\mathcal{C}_M(\Lambda_t)$.

Moreover, it is important to notice that $\Lambda_{t,\mathrm{free}}$ (and similarly $\mathcal{L}_{t,\mathrm{free}}$ in the case of time-local master equations), as defined in Eq.~\eqref{eq:channel-decomp}, is in general \textit{not} the same as the largest part of $\Lambda_t$ that commutes with $\mathcal{G}$. To see this, recall the decomposition of any linear map $T:\mathcal{B}(\mathcal{H})\to\mathcal{B}(\mathcal{H})$ from Eq.~\eqref{eq:T-decomp}. Writing $\mathcal{G}^\perp:=\mathrm{id}-\mathcal{G}$, we have
\begin{equation}\label{eq:T-decomp-G}
T = \mathcal{G}\circ T\circ\mathcal{G} + \mathcal{G}^\perp\circ T\circ\mathcal{G}^\perp+ \mathcal{G}^\perp\circ T\circ\mathcal{G} + \mathcal{G}\circ T\circ\mathcal{G}^\perp\,.
\end{equation}
The ``block-diagonal'' part
\begin{equation}
\widetilde{T} := \mathcal{G}\circ T\circ\mathcal{G} + \mathcal{G}^\perp\circ T\circ\mathcal{G}^\perp
\end{equation}
preserves both $\mathrm{im}(\mathcal{G})$ and $\ker(\mathcal{G})$ and commutes with $\mathcal{G}$, $[\widetilde T,\mathcal{G}]=0$. In contrast, the cross terms $\mathcal{G}^\perp\circ T\circ\mathcal{G}$ and $\mathcal{G}\circ T\circ\mathcal{G}^\perp$ need not commute with $\mathcal{G}$ for a generic $T$. Thus $\widetilde T$ can be viewed as the projection of $T$ onto the subspace of maps that commute with $\mathcal{G}$ (block-diagonal in the decomposition $\mathcal{B}(\mathcal{H})=\mathrm{im}(\mathcal{G})\oplus\ker(\mathcal{G})$), whereas in Theorem~\ref{thm:dissect-channel} we defined the ``free'' part of $\Lambda_t$ and $\mathcal{L}_t$ as the projection onto the image of $\mathcal{G}$ only:
\begin{equation}
\Lambda_{t,\mathrm{free}} = \mathcal{G}\circ\Lambda_t\circ\mathcal{G}\,,\quad
\mathcal{L}_{t,\mathrm{free}} = \mathcal{G}\circ\mathcal{L}_t\circ\mathcal{G}\,.
\end{equation}

This choice is motivated by operational considerations: $\mathcal{G}\circ T\circ\mathcal{G}$ is, by construction, the unique component of $T$ that acts non-trivially only on the ``free'' sector $\mathrm{im}(\mathcal{G})$. The residual part $T_{\mathrm{res}} := T - \mathcal{G}\circ T\circ\mathcal{G}$ then collects all contributions that involve the resourceful directions $\ker(\mathcal{G})$. In particular, Theorem \ref{thm:dissect-channel} part (ii) shows that $\Lambda_{t,\mathrm{res}}$ contains precisely those ``resourceful'' contributions of $\Lambda_t$ that can generate a non-zero response functional $\mathcal{C}_M(\Lambda_t)$. This is reflected in Eqs.~\eqref{eq:decomp-properties-1} and \eqref{eq:decomp-properties-2}, and yields a clean splitting between free and resourceful parts of the dynamics.

For $\mathcal{G}$ linear, the algebraic decompositions in Eqs.~\eqref{eq:B-dec} and \eqref{eq:Lt-dec} are always well defined. If $\mathcal{G}$ is non-linear, however, $\mathrm{im}(\mathcal{G})$ and $\ker(\mathcal{G})$ are no longer linear subspaces in general, so the direct-sum decomposition \eqref{eq:B-dec} is not available. Nevertheless, for both linear and non-linear idempotent $\mathcal{G}$ one can still define the residual map
\begin{equation}
\widetilde{\Lambda}_{\mathrm{res}} := \Lambda - \Lambda\circ\mathcal{G}\,,
\end{equation}
and a direct calculation using $\mathcal{G}^2=\mathcal{G}$ shows that
\begin{equation}
\mathcal{C}_M(\Lambda) = \mathcal{C}_M(\widetilde{\Lambda}_{\mathrm{res}})\,.
\end{equation}
When $\mathcal{G}$ is linear, Theorem \ref{thm:dissect-channel} part (ii) additionally gives $\mathcal{C}_M(\Lambda) = \mathcal{C}_M(\Lambda_{\mathrm{res}})$, such that
\begin{equation}
\mathcal{C}_M(\Lambda) = \mathcal{C}_M(\widetilde{\Lambda}_{\mathrm{res}}) = \mathcal{C}_M(\Lambda_{\mathrm{res}})\,,
\end{equation}
even though in general $\widetilde{\Lambda}_{\mathrm{res}}\neq \Lambda_{\mathrm{res}}$. Equality of the maps $\widetilde{\Lambda}_{\mathrm{res}},\Lambda_{\mathrm{res}}$ holds if and only if
\begin{equation}
\Lambda\circ\mathcal{G} = \mathcal{G}\circ\Lambda\circ\mathcal{G} = \Lambda_{\mathrm{free}}\,.
\end{equation}
A sufficient condition for this is $\mathcal{G}$-covariance of $\Lambda$, i.e. $\Lambda\circ\mathcal{G} = \mathcal{G}\circ\Lambda$, in which case $\mathcal{G}\circ\Lambda\circ\mathcal{G} = \Lambda\circ\mathcal{G}$.

The two cross terms in \eqref{eq:T-decomp} have natural interpretations in the resource theoretic setting. The block $\mathcal{G}^\perp\circ\Lambda\circ\mathcal{G}$ describes the part of $\Lambda$ that sends free inputs (in $\mathrm{im}(\mathcal{G})$) into resourceful sectors (i.e., elements of $\ker(\mathcal{G})$), while $\mathcal{G}\circ\Lambda\circ\mathcal{G}^\perp$ describes the converse direction. In the terminology of Ref.~\cite{LHL17}, a channel $\Lambda$ is \textit{resource non-generating} if and only if
\begin{equation}\label{eq:res-non-gen}
\mathcal{G}\circ\Lambda\circ\mathcal{G} = \Lambda\circ\mathcal{G}\,,
\end{equation}
which is equivalent to
\begin{equation}
\mathcal{G}^\perp\circ\Lambda\circ\mathcal{G} = 0\,.
\end{equation}
Thus, a resource non-generating quantum channel maps the set of free states into itself.
Similarly, $\Lambda$ is \textit{resource non-activating} if and only if
\begin{equation}\label{eq:res-non-act}
\mathcal{G}\circ\Lambda\circ\mathcal{G} = \mathcal{G}\circ\Lambda\,,
\end{equation}
which is equivalent to
\begin{equation}
\mathcal{G}\circ\Lambda\circ\mathcal{G}^\perp = 0\,.
\end{equation}
The above definition implies that a resource non-activating channel $\Lambda$ cannot exploit any resource present in a state $\rho$. Indeed, $\mathcal{G} \circ \Lambda$ is equivalent to first destroying the resource by applying \(\mathcal{G}\) and then evolving the state, so $\Lambda$ cannot ``activate'' a resource contained in $\rho$ \cite{LHL17}.
As shown in Ref.~\cite{LHL17}, $\Lambda$ is both resource non-generating and resource non-activating if and only if it is $\mathcal{G}$-covariant, $\Lambda\circ\mathcal{G} = \mathcal{G}\circ\Lambda$.

Finally, even when $\mathcal{G}$ is CPTP, the free component $\Lambda_{t,\mathrm{free}} = \mathcal{G}\circ\Lambda_t\circ\mathcal{G}$ is CPTP, but the residual map $\Lambda_{t,\mathrm{res}} = \Lambda_t - \Lambda_{t,\mathrm{free}}$ need not be CPTP. Likewise, in the time-local case the projected generators $\mathcal{L}_{t,\mathrm{free}}$ and $\mathcal{L}_{t,\mathrm{res}}$ are algebraic objects and are not guaranteed to be physical Lindblad generators in their own right. 
In Corollary \ref{cor:gen-free-compatibility} we make precise what we mean by $\mathcal{L}_{t,\mathrm{free}}$ being \textit{compatible} with $\Lambda_{t,\mathrm{free}}$ and give a simple criterion. For this, recall that in the time-independent case a linear map $\mathcal{L}$ on $\mathcal{B}(\mathcal{H})$ is a physical Lindblad generator if and only if it generates a quantum dynamical semigroup $\Lambda_t = e^{t\mathcal{L}}$, or equivalently a GKSL master equation \cite{GKS76, Lindblad1976}. In the time-dependent case, $\mathcal{L}_t$ is typically called physical if it generates a CPTP family via a time-local master equation, or equivalently if the corresponding dynamical map is completely positive divisible \cite{CHMS14, BLPV16} (see also Appendix \ref{app:time-local} for more information).

\begin{cor}\label{cor:gen-free-compatibility}
Let $\Lambda_{t,\mathrm{free}}$ and $\mathcal{L}_{t,\mathrm{free}}$ be defined as in Theorem~\ref{thm:dissect-channel}, and let $\mathcal{G}^\perp := \mathrm{id}-\mathcal{G}$. Then the following are equivalent:
\begin{enumerate}
\item[(i)] $\Lambda_{t,\mathrm{free}}$ solves the time-local master equation
\begin{equation}\label{eq:compatibility-cond-TLME}
\frac{\mathrm{d}}{\mathrm{d} t}\Lambda_{t,\mathrm{free}} = \mathcal{L}_{t,\mathrm{free}}\circ\Lambda_{t,\mathrm{free}} \quad \forall t\in[0,T]
\end{equation}
with $\Lambda_{0,\mathrm{free}}=\mathcal{G}$.
\item[(ii)] The following compatibility condition holds:
\begin{equation}\label{eq:compatibility-glg}
\mathcal{G}\circ\mathcal{L}_t\circ\mathcal{G}^\perp\circ\Lambda_t\circ\mathcal{G} = 0
\quad\forall t\in[0,T]\,.
\end{equation}
\end{enumerate}
In particular, each of the stronger conditions $\mathcal{G}\circ\mathcal{L}_t\circ\mathcal{G}^\perp = 0$ or $\mathcal{G}^\perp\circ\Lambda_t\circ\mathcal{G} = 0$ for all $t\in[0,T]$ is sufficient to guarantee \eqref{eq:compatibility-cond-TLME}.
\end{cor}
The short proof of Corollary \ref{cor:gen-free-compatibility} is shown in Appendix \ref{app:proof-cor-10}. 

Clearly, $\Lambda_t$ being $\mathcal{G}$-covariant, $\Lambda_t\circ\mathcal{G} = \mathcal{G}\circ\Lambda_t$ for all $t$, is a stronger condition than the compatibility condition in Corollary \ref{cor:gen-free-compatibility} and is therefore sufficient to guarantee that $\Lambda_{t,\mathrm{free}}$ and $\mathcal{L}_{t,\mathrm{free}}$ are compatible in the sense of Eq.~\eqref{eq:compatibility-cond-TLME}. In particular, $\mathcal{G}$-covariance implies that $\Lambda_t$ is both resource non-generating and non-activating \cite{LHL17}, such that $\mathcal{G}^\perp\circ\Lambda_t\circ\mathcal{G}=0$ for all $t$. As discussed above, this is already enough to ensure that Eq.~\eqref{eq:compatibility-glg} is satisfied and,  hence, that $\mathcal{L}_{t,\mathrm{free}}$ solves the time-local master equation for $\Lambda_{t,\mathrm{free}}$. More generally, for any family $\{\Lambda_t\}_{t\in[0,T]}$ such that each $\Lambda_t$ is resource non-generating (i.e.~$\mathcal{G}^\perp\circ\Lambda_t\circ\mathcal{G}=0$), the projected generator $\mathcal{L}_{t,\mathrm{free}}$ generates the projected dynamics $\Lambda_{t,\mathrm{free}}$ via Eq.~\eqref{eq:compatibility-cond-TLME}.

We now illustrate the splitting of the dynamics into free and resourceful components introduced in Theorem \ref{thm:dissect-channel} for some well-known Lindblad generators and investigate in each case whether $\Lambda_{t,\mathrm{free}}$ solves the corresponding projected time-local master equation. We present one example where the ``free'' components of the Lindblad generator and the quantum channel are compatible, and one where they are not.

As a concrete setting we consider a single qubit with Hilbert space $\mathcal{H}\cong\CC^2$ and choose as resource-destroying map the dephasing map in the computational basis,
\begin{equation}
\mathcal{G}(\rho) = \sum_{i=0,1}\ket{i}\!\bra{i}\rho\ket{i}\!\bra{i}\,,
\end{equation}
which is CPTP. In Liouville space $\mathcal{H}\otimes\mathcal{H}$ this dephasing map has the matrix representation
\begin{equation}\label{eq:G-dephasing-L}
\mathcal{G} = \frac{1}{2}\left(\sigma_z\otimes\sigma_z + \mathbbm{1}_{\mathcal{H}\otimes\mathcal{H}}\right)\,,
\end{equation}
where we use row-stacking (see Eq.~\eqref{eq:row-stacking} in Appendix~\ref{app:time-local}) to represent superoperators acting on vectorized density operators $|\rho\rangle\!\rangle$. For notational simplicity we use the same symbols $\mathcal{G}$, $\Lambda_t$ and $\mathcal{L}_t$ for both the maps and their matrix representations.

An example of a time-local generator $\mathcal{L}_t\in\mathcal{B}(\mathcal{H}\otimes\mathcal{H})$ for which $\mathcal{L}_{t,\mathrm{free}}$ is compatible with $\Lambda_{t,\mathrm{free}}$ is
\begin{equation}\label{eq:multi-decay-qubit}
\mathcal{L}_t = \sum_{i\in\{x,y,z\}}\gamma_i(t)\left(\sigma_i\otimes\sigma_i^\mathrm{T} - \mathbbm{1}_{\mathcal{H}\otimes\mathcal{H}}\right)\,,
\end{equation}
with $\gamma_i(t)\ge 0$ for all $t$ and $\sigma_i$ the Pauli matrices in the computational basis. The Lindblad superoperator~\eqref{eq:multi-decay-qubit} describes a qubit subject to three independent decoherence channels~\cite{BLPV16}. Since $\mathcal{G}$ is defined in the same basis, $\mathcal{L}_t$ preserves the decomposition into diagonal and off-diagonal operators, and one finds
\begin{align}
\mathcal{L}_{t,\mathrm{free}} &= (\gamma_x(t)+\gamma_y(t))\Big[\sigma_-\otimes\sigma_- + \sigma_+\otimes\sigma_+ \nonumber\\
&\quad\quad - \frac{1}{2}\big(\sigma_z\otimes\sigma_z + \mathbbm{1}_{\mathcal{H}\otimes\mathcal{H}}\big)\Big]\,.
\end{align}
In particular, $\mathcal{G}\circ\mathcal{L}_t\circ\mathcal{G}^\perp=0$, so by Corollary~\ref{cor:gen-free-compatibility} the projected generator $\mathcal{L}_{t,\mathrm{free}}$ indeed generates the projected dynamics $\Lambda_{t,\mathrm{free}}$.

In contrast, not every projected channel $\Lambda_{t,\mathrm{free}}$ is compatible with its projected generator. A simple example is provided by purely Hamiltonian evolution
\begin{equation}
\mathcal{L}(\rho) = -i[H,\rho]\,,\qquad H = \frac{\Omega}{2}\sigma_x\,.
\end{equation}
For the dephasing map $\mathcal{G}$ in Eq.~\eqref{eq:G-dephasing-L}, any diagonal state $\rho$ satisfies that $[H,\rho]$ is purely off-diagonal, such that 
\begin{equation}
\mathcal{L}_{\mathrm{free}}= \mathcal{G}\circ\mathcal{L}\circ\mathcal{G}= 0\,.
\end{equation}
At the same time, the projected dynamics $\Lambda_{t,\mathrm{free}} = \mathcal{G}\circ\Lambda_t\circ\mathcal{G}$ is non-trivial and can be written in Liouville form as
\begin{align}
\Lambda_{t,\mathrm{free}}&= \frac{1}{2}\cos^2 \left(\frac{\Omega t}{2}\right) (\sigma_z\otimes\sigma_z + \mathbbm{1}\otimes\mathbbm{1}) \nonumber\\
&\quad + \sin^2\left(\frac{\Omega t}{2}\right) (\sigma_-\otimes\sigma_- + \sigma_+\otimes\sigma_+)\,.
\end{align}
Since $\mathcal{L}_{\mathrm{free}} = 0$ but $\Lambda_{t,\mathrm{free}}$ is time-dependent for all $t\neq 2\pi m/\Omega$, $m\in\ZZ$, the projected channel does not satisfy the projected differential equation $\frac{d}{dt}\Lambda_{t,\mathrm{free}} = \mathcal{L}_{\mathrm{free}}\circ\Lambda_{t,\mathrm{free}}$ except at those isolated times. This illustrates that the ``free'' components of generator and channel need not be compatible in general.

\section{Resource theories without resource destroying map\label{sec:RT-wo-RDM}}

Although having a free-state set that forms an affine subset of the state space is a necessary condition for the existence of a linear resource-destroying map, it is not sufficient as there exist affine resource theories that do not admit any such map. In the absence of a resource-destroying map (linear or not), we can no longer use it to define the ``classical'' baseline that we relied on in Sec.~\ref{sec:defs} to quantify how much advantage a channel can extract from a resource. In what follows, we show how the concepts from Sec.~\ref{sec:defs} can be extended to resource theories without a resource-destroying map. This extension, however, comes with a subtle change in perspective and in the interpretation of the analogue of the resource impact functional $\mathcal{C}_M(\Lambda)$, which we discuss below.

We consider a resource theory with a set of free states $\mathcal{F}$, a set of free operations $\mathcal{O}$ (see Appendix~\ref{app:free-set-op}), and a divergence $D(\cdot\|\cdot)$ such that
\begin{equation}
D(\Phi(\rho)\|\Phi(\sigma)) \leq D(\rho\|\sigma) \quad\forall\,\rho,\sigma\in\mathcal{D}(\mathcal{H})\,, \forall\,\Phi\in\mathcal{O}\,,
\end{equation}
i.e., $D$ is contractive under all free operations. Moreover, we assume that $\mathcal{F}$ is compact and that $D(\rho\Vert\sigma)$ is lower semi-continuous in $\sigma$ such that, for each $\rho\in\mathcal{D}(\mathcal{H})$, the minimiser of $D(\rho\|\sigma)$ over $\sigma\in\mathcal{F}$ exists. We denote the set of minimisers by
\begin{equation}
\mathcal{P}(\rho) := \arg\min_{\sigma\in\mathcal{F}} D(\rho\|\sigma)\subset\mathcal{F}\,.
\end{equation}

In analogy to the definition of $\mathcal{C}_M(\Lambda)$ in Eq.~\eqref{eq:capacity-sup}, we then define the resource impact functional ($\mathcal{H}$ finite-dimensional)
\begin{equation}\label{eq:CM-general}
\mathfrak{C}_M(\Lambda) := \max_{\rho\in\mathcal{D}(\mathcal{H})} \max_{\pi(\rho)\in\mathcal{P}(\rho)}\big|\Tr\left[M\left(\Lambda(\rho) - \Lambda(\pi(\rho))\right)\right]\big|\,.
\end{equation}
Here $\pi(\rho)$ is chosen among the free states that are closest to $\rho$ with respect to $D(\cdot\|\cdot)$, and the extra maximisation over $\pi(\rho)\in\mathcal{P}(\rho)$ resolves the possible non-uniqueness of the closest free state. By construction, we adopt an ``optimistic'' convention \footnote{Alternatively, one could replace the inner $\max_{\pi(\rho)\in\mathcal{P}(\rho)}$ by a $\min$ (conservative scenario) or fix a particular selection rule for $\pi(\rho)$. In case the minimiser of $D(\cdot\|\cdot)$ is unique, all these conventions coincide on the level of the resource impact functional. However, for the variation bound in terms of $\Gamma_M(t)$ below it is required to either work with twice a maximum (i.e., over $\rho$ and $\pi(\rho)$) or twice a minimum}: for each $\rho$ we select, among all closest free baselines, the one that yields the largest task advantage. If $\mathcal{P}(\rho)$ consists of a single state for all $\rho$ (for instance under suitable strict convexity assumptions on $D(\cdot\|\cdot)$), the inner minimisation is trivial and $\mathfrak{C}_M(\Lambda)$ reduces to a direct analogue of $\mathcal{C}_M(\Lambda)$ with the resource-destroying map $\mathcal{G}$ replaced by $\rho\mapsto\pi(\rho)$.

In analogy to $\mathcal{C}_M(\Lambda)$ in Eq.~\eqref{eq:capacity-sup}, the functional $\mathfrak{C}_M(\Lambda)$ still pairs each density operator $\rho$ with a corresponding free counterpart. However, its interpretation is slightly different: while the resource-destroying map $\mathcal{G}$ in Sec.~\ref{sec:defs} provides an operationally defined free counterpart $\mathcal{G}(\rho)$ obtained by physically ``removing'' the resource, $\pi(\rho)$ selects a free baseline \textit{geometrically}, as a minimiser of a chosen quantum information divergence $D(\cdot\|\cdot)$. 
Consequently, different choices of $D(\cdot\|\cdot)$ lead to different notions of what it means for a state to be a ``closest'' free counterpart of $\rho$. For example, choosing $\sigma$ to minimise the trace distance $\|\rho-\sigma\|_1/2$ picks the free state that is hardest to distinguish from $\rho$ in a single-shot minimum-error discrimination task, as quantified by the Helstrom bound \cite{Helstrom1976, NielsenChuang, Watrous2018}. In contrast, choosing $\sigma$ to minimise the quantum relative entropy $S(\rho\Vert\sigma) = \Tr[\rho(\ln(\rho-\ln(\sigma))]$ yields a free state that is asymptotically hardest to distinguish from $\rho$ in a many-copy hypothesis testing scenario \cite{HP91, ON00}.

Based on the definition of $\mathfrak{C}_M(\Lambda)$ in Eq.~\eqref{eq:CM-general}, we can introduce an analogue of the instantaneous rate $\Gamma_M(t)$ in Eq.~\eqref{eq:Gamma-M-sup} for the distance-projected setting by
\begin{equation}
\mathfrak{T}_M(t):= \max_{\rho\in\mathcal{D}(\mathcal{H})} \max_{\pi(\rho)\in\mathcal{P}(\rho)}\left|\Tr\left[M\,\dot{\Lambda}_t(\rho - \pi(\rho))\right]\right|.
\end{equation}
Replacing $\rho - \mathcal{G}(\rho)$ by $\rho - \pi(\rho)$ throughout the proof of Theorem \ref{thm:variation-bound}, and taking the maximum over the choice of closest free state $\pi(\rho)\in\mathcal{P}(\rho)$, yields variation and time bounds for $\mathfrak{C}_M(\Lambda_t)$ that are completely analogous to Theorem \ref{thm:variation-bound} and Corollary \ref{cor:time-bound}. Since these modifications are purely notational, we omit the details here, and continue with the discussion of three chemically motivated examples. 

\section{Donor-acceptor model\label{sec:DAM}}

In this section, we illustrate how to compute the resource impact functional $\mathcal{C}_M(\Lambda)$ and the instantaneous advantage rate $\Gamma_M(t)$ for a simple donor-acceptor dimer. This system can be viewed as the minimal non-trivial model of excitonic energy transfer, where an electronic excitation initially localized on the donor, $\ket{D} :=  \ket{1_e}\otimes\ket{2_g}$,
can coherently hop to the acceptor, $\ket{A} := \ket{1_g}\otimes\ket{2_e}$,
which is subject to dephasing and spontaneous emission, where $\ket{1_g} \,(\ket{1_e})$ correspond to the ground (excited) state of the donor site, and similarly $\ket{2_g} \,(\ket{2_e})$ refer to the ground (excited) state of the acceptor site. 
In photosynthetic complexes and model exciton transport networks, the population transferred to an acceptor (or trap) is often interpreted as a transport efficiency or reaction yield \cite{MS10, Thilagam12, OROC14, YC20}. 

We work in the single-excitation subspace with Hilbert space
\begin{equation}
\mathcal{H} = \mathrm{span}_{\mathbbm{C}}\{\ket{g},\ket{D},\ket{A}\}\cong\mathbbm{C}^3,
\end{equation}
where $\ket{g} := \ket{1_g}\otimes\ket{2_g}$ denotes the electronic ground state and $\ket{D}$ and $\ket{A}$ are the donor and acceptor exciton states, respectively.

We consider the resource theory of coherence in the site basis $\{\ket{D},\ket{A}\}$. Free states are those that are diagonal in this basis, such that the corresponding set of free states is given by
\begin{align}
\mathcal{F} &= \big\{\rho \,\big|\,\rho = p_D\ket{D}\!\bra{D} + p_A\ket{A}\!\bra{A}\,,p_D,p_A\geq 0\,,\nonumber\\
&\quad \qquad p_D+p_A\leq 1 \big\}\,,
\end{align}
while the resource consists of donor-acceptor coherences $\rho_{DA} = \bra{D}\rho\ket{A}\neq 0$. As resource-destroying map we take complete dephasing in the site basis,
\begin{equation}\label{eq:G-DAM}
\mathcal{G}(\rho) = \sum_{i\in\{g,D,A\}}\ket{i}\!\bra{i}\rho\ket{i}\!\bra{i}\,,
\end{equation}
and we use $\mathcal{C}_M(\Lambda)$ and $\Gamma_M(t)$ as defined in Eqs.~\eqref{eq:capacity-sup} and~\eqref{eq:Gamma-M-sup}. For the Lindblad dynamics considered below, no coherence between $\ket{g}$ and the single-excitation manifold is generated and, hence, only the $\{\ket{D},\ket{A}\}$ block contributes. 

As a representative task we choose $M =\ket{A}\!\bra{A}$, such that $\Tr[M\Lambda(\rho)]$ is the acceptor population at the end of the process, i.e., the exciton transfer yield. In this language, $\mathcal{C}_M(\Lambda)$ quantifies the maximal increase in acceptor yield that any initial site coherence can provide, compared to its dephased counterpart $\mathcal{G}(\rho)$.

\subsection{Static coherence advantage from donor-acceptor mixing \label{sec:DAM-coherence}}

We first isolate the effect of coherent donor-acceptor mixing at the level of a single unitary channel associated with the excitonic mixing angle.
For this, we consider the Hamiltonian
\begin{equation}\label{eq:Hex}
H_{\mathrm{ex}}= \frac{\Delta}{2}\left(\ket{D}\!\bra{D}-\ket{A}\!\bra{A}\right)
+ J\left(\ket{D}\!\bra{A}+\ket{A}\!\bra{D}\right)\,,
\end{equation}
with detuning $\Delta$ and electronic coupling $J$. Its eigenstates are $\ket{+} = \cos\theta\ket{D} +\sin\theta\ket{A}$ and $\ket{-}= - \sin\theta\ket{D}+\cos\theta{A}$, where $\theta = \tfrac{1}{2}\arctan(2J/\Delta)$ is the donor-acceptor mixing angle.

The associated unitary channel $\Lambda_\theta(\rho) = U_\theta \rho U_\theta^\dagger$ with
\begin{align}\label{eq:Utheta}
U_{\theta} &= \ket{g}\!\bra{g} + \cos\theta\left(\ket{D}\!\bra{D} + \ket{A}\!\bra{A}\right)\nonumber\\
&\quad  + \sin\theta\left(\ket{A}\!\bra{D} -\ket{D}\!\bra{A}\right)
\end{align}
generates coherence in the site basis and is therefore not a free operation in the resource theory of coherence. A straightforward calculation shows that
\begin{equation}\label{eq:CM-Ltheta-main}
\mathcal{C}_{|A\rangle\!\langle A|}(\Lambda_\theta) = \left\|(\mathrm{id}-\mathcal{G})\Lambda_\theta^\dagger(M)\right\|_\infty = \frac{1}{2}\left|\sin(2\theta)\right|\,,
\end{equation}
which is maximised at $\theta=\pi/4$, corresponding to resonant, maximally delocalised excitonic eigenstates, as expected, and $\mathcal{C}_M(\Lambda_\theta)$ quantifies its maximal magnitude for a non-zero advantage in the exciton transfer task.

We stress that $\Lambda_\theta(\rho)=U_\theta\rho U_\theta^\dagger$ is interpreted here as an active one-step donor-acceptor mixing channel, not as a passive basis change. With the site basis, the dephasing map $\mathcal{G}$, and the acceptor readout $M = \ket A\!\bra A$ held fixed, $\mathcal{C}_{|A\rangle\!\langle A|}(\Lambda_\theta)$ quantifies the maximal change in the acceptor population caused by initial site coherence relative to the corresponding dephased preparation. Equivalently, in the Heisenberg picture, $U_\theta^\dagger \ket A\!\bra A\, U_\theta$ contains an off-diagonal term proportional to $\sin(2\theta)$, which is precisely what makes the readout sensitive to initial site coherence. In this sense, the ``static'' coherence advantage is the one-shot, task-dependent coherence sensitivity generated by donor-acceptor mixing.

\subsection{Dissipative dynamics, $\Gamma_M(t)$, and variation bounds\label{sec:DAM-bounds}}

\begin{figure}[tb]
\centering
\includegraphics[width=0.75\linewidth]{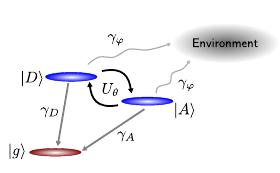}
\caption{Illustration of the three-site donor-acceptor model and the three processes described by the quantum channel in Eq.~\eqref{eq:Lambda-conc} (see text for more details). \label{fig:DAM-schematic}}
\end{figure}

To make the model more realistic, we include dephasing in the excited-state manifold and spontaneous emission to the electronic ground state. At the level of a single effective ``gate time'' $\Delta t$, we model these processes by a phase damping channel $\Lambda_{\mathrm{PD}}$ with parameter $\eta$ and an amplitude damping channel $\Lambda_{\mathrm{AD}}$ with decay probabilities $p_D,p_A$ from $\ket{D}$ and $\ket{A}$ to $\ket{g}$, respectively (explicit Kraus operators are given in Appendix~\ref{app:DAM}). Combined with the coherent donor-acceptor mixing $\Lambda_\theta(\rho)=U_\theta\rho U_\theta^\dagger$, the full one-step process is described by the CPTP map
\begin{equation}\label{eq:Lambda-conc}
\Lambda = \Lambda_{\mathrm{AD}}\circ\Lambda_{\mathrm{PD}}\circ\Lambda_\theta\,.
\end{equation}
For a general POVM element
\begin{equation}\label{eq:M-DAM}
M = \mu_g\ket{g}\!\bra{g} + \mu_D\ket{D}\!\bra{D} + \mu_A\ket{A}\!\bra{A} + \left(\nu\ket{D}\!\bra{A} + \mathrm{h.c.}\right)\,,
\end{equation}
with $0\leq M\leq \mathbbm{1}$, $\mu_i\in\RR$ ($i\in\{g,D,A\}$), and $\nu\in\CC$, one finds that (see Appendix~\ref{app:DAM})
\begin{align}
\mathcal{C}_M(\Lambda) &= \Bigg|\cos(2\theta)\,\eta\,\nu\sqrt{(1-p_D)(1-p_A)}\nonumber\\
&\qquad+\frac{1}{2}\sin(2\theta)\Big[\mu_A(1-p_A)-\mu_D(1-p_D)\nonumber\\
&\qquad +\mu_g(p_A-p_D)\Big]\Bigg|\,.
\end{align}
In particular, for the readout $M=\ket{A}\!\bra{A}$ this reduces to
\begin{equation}
\mathcal{C}_{|A\rangle\!\langle A|}(\Lambda) = \frac{1}{2}(1-p_A)\left|\sin(2\theta)\right|\,,
\end{equation}
i.e., a damped version of Eq.~\eqref{eq:CM-Ltheta-main}, as expected from amplitude damping on the acceptor. 

To analyse the time dependence of coherence-induced advantage, we additionally employ a continuous-time Markovian description and propagate observables in the Heisenberg picture under the adjoint Lindblad generator $\mathcal{L}^\dagger$. Specifically, we consider the semigroup $\Lambda_t^\dagger=\mathrm{exp}(t\mathcal{L}^\dagger)$ satisfying
\begin{equation}\label{eq:DAM-ME1}
\frac{\mathrm{d}}{\mathrm{d}t}\Lambda_t^\dagger(M) = \mathcal{L}^\dagger(\Lambda_t^\dagger(M))\,,
\end{equation}
with
\begin{align}\label{eq:DAM-ME2}
\mathcal{L}^\dagger(M) &= i[H_{\mathrm{ex}},M] + \sum_{k=1}^4\left(L_k^\dagger M L_k - \tfrac{1}{2}\{L_k^\dagger L_k,M\}\right)\,,
\end{align}
and jump operators
\begin{align}\label{eq:L-DAM}
L_1 &= \sqrt{\gamma_{\varphi}}\ket{D}\!\bra{D}\,,\quad
L_2 = \sqrt{\gamma_{\varphi}}\ket{A}\!\bra{A}\,, \nonumber\\
L_3 &= \sqrt{\gamma_D}\ket{g}\!\bra{D}\,,\quad
L_4 = \sqrt{\gamma_A}\ket{g}\!\bra{A}\,.
\end{align}
The finite-time parameters are related to the rates by $\eta(\Delta t)=e^{-\gamma_\varphi \Delta t}$ and $p_j(\Delta t)=1-e^{-\gamma_j \Delta t}$ for $j\in\{D,A\}$. We stress that the CPTP map in Eq.~\eqref{eq:Lambda-conc} and the semigroup $\Lambda_{\Delta t}=e^{\Delta t\mathcal{L}}$ are, in general, not identical for finite $\Delta t$ (since the coherent and dissipative generators do not commute when $J\neq 0$). Instead, Eq.~\eqref{eq:Lambda-conc} can be viewed as a convenient splitting approximation that agrees with $e^{\Delta t\mathcal{L}}$ to first order in $\Delta t$ and becomes exact in the limit of sufficiently small time steps.

The adjoint dynamics preserves the block form of $M$ in Eq.~\eqref{eq:M-DAM}, such that we can parametrise
\begin{align}
\Lambda_t^\dagger(M) &= x_g(t)\ket{g}\!\bra{g}  + x_D(t)\ket{D}\!\bra{D} + x_A(t)\ket{A}\!\bra{A} \nonumber\\
&\quad + (y(t)\ket{D}\!\bra{A} + \mathrm{h.c.})\,,
\end{align}
with $x_g(t),x_D(t),x_A(t)\in\RR$ and $y(t)\in\CC$. For the dephasing map $\mathcal{G}$ in Eq.~\eqref{eq:G-DAM} this implies (see Appendix~\ref{app:DAM} for more details)
\begin{equation}\label{eq:CM-GM-DM-main}
\mathcal{C}_M(\Lambda_t) = |y(t)|\,,\qquad  \Gamma_M(t) = \left|\frac{\mathrm{d}y(t)}{\mathrm{d}t}\right|\,.
\end{equation}
The coupled linear equations governing $x_D(t),x_A(t),y(t)$ are collected in Eq.~\eqref{eq:DA-CDE} in Appendix \ref{app:DAM}, where we also derive explicit solutions in several physically relevant regimes, which are discussed next.

\begin{figure}[tb]
\centering
\includegraphics[width=\linewidth]{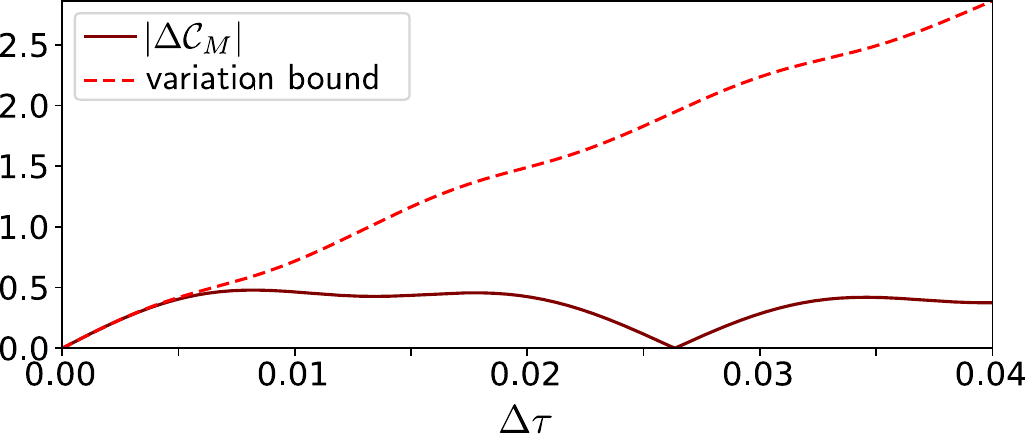}
\caption{Illustration of the coherence-induced task advantage and its variation bound in the donor--acceptor model. 
We plot the exact task advantage $|\Delta\mathcal{C}_M|$ (dark red, solid) and the integrated variation bound $\int_0^t\Gamma_M(s)\,\mathrm{d}s$ (light red, dashed) from Theorem \ref{thm:variation-bound} for $M = |A\rangle\!\langle A|$, starting at $t_1=0$, for parameters $\Delta=130$, $J=100$, and $\gamma=5$ (in arbitrary but consistent units). 
 \label{fig:DAM-bounds}}
\end{figure}

\subsubsection*{Zero dephasing and symmetric lifetimes}

As an illustrative example, consider first the case of zero dephasing ($\gamma_{\varphi}=0$) and symmetric lifetimes ($\gamma_D=\gamma_A=\gamma$). Derivation and explicit expressions of $\mathcal{C}_M(\Lambda)$, $\Gamma_M(\Lambda)$ and the corresponding variation bound (recall Theorem \ref{thm:variation-bound}) for general $M$ in Eq.~\eqref{eq:M-DAM} can be found it in Appendix \ref{sec:zero-deph-DA}. Here, we illustrate in Fig.~\ref{fig:DAM-bounds} the behaviour of $|\Delta\mathcal{C}_M(\Lambda_{t_2}, \Lambda_{t_1})|$ (dark red, solid) as well as the integrated variation bound (red, dashed) for $M=\ket{A}\!\bra{A}$. 
For short times, the bound $\int_0^t\Gamma_{|A\rangle\!\langle A|}(s)\,\mathrm{d}s$ closely follows $|\mathcal{C}_{|A\rangle\!\langle A|}(\Lambda_t)|$, and therefore provides a sharp estimate of the initial growth of the coherence-induced advantage. In the present regime, the Heisenberg-picture coherence takes the form (see Appendix \ref{sec:zero-deph-DA})
\begin{equation}
y(t)=e^{-\gamma t}\left[-\frac{J\Delta}{\Omega^2}\left(1 - \cos(\Omega t)\right) - i\frac{J}{\Omega}\sin(\Omega t)\right]\,,
\end{equation}
where $\Omega = \sqrt{\Delta^2+4J^2}$ is the generalized Rabi frequency. 
Therefore, the oscillatory structure of $|\Delta\mathcal{C}_M|  = |y(t)|$ is set by the generalized Rabi frequency $\Omega$. In particular, $|\Delta\mathcal{C}_M|$ vanishes at $t_n = 2\pi n/\Omega, n\in \ZZ$, where the Heisenberg evolved readout becomes diagonal in the site basis and the acceptor population is momentarily insensitive to initial site coherence. Between successive nodes, the local maxima mark optimal readout times and are controlled by the donor-acceptor mixing, through $\sin(2\theta) = 2J/\Omega$, while their envelope is damped by $e^{-\gamma t}$.
At longer times, the integrated variation bound becomes increasingly loose because $\int_0^t \Gamma_M(s)\,\mathrm{d}s = \int_0^t|\dot y(s)|\,\mathrm{d}s$ measures the total variation of the oscillatory function $y(t)$ and therefore counts each increase and decrease without allowing for cancellations. Its monotonic growth should thus be interpreted as a worst-case control quantity rather than as an accumulated task advantage.

\subsubsection*{Zero detuning and finite dephasing}

\begin{figure*}[tb]
\centering
\includegraphics[width=\linewidth]{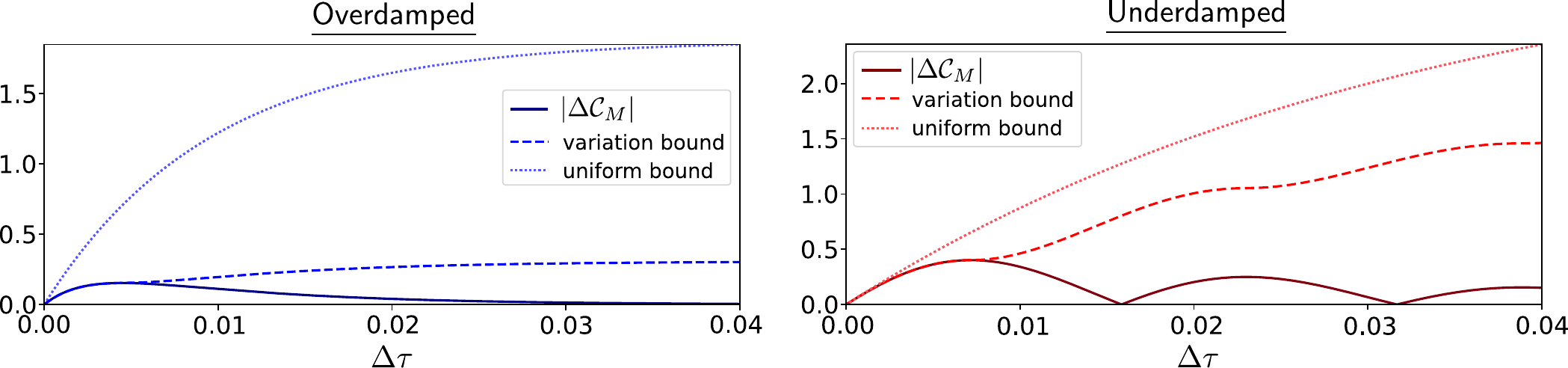}
\caption{Change in the coherence-induced resource impact functional $|\Delta\mathcal{C}_M|$ for $M=\ket{A}\!\bra{A}$ (solid) in the donor-acceptor model for zero detuning and $J=100$, $\gamma=5$ in the corresponding variation bound $\int_0^t\Gamma_M(s)\,\mathrm{d}s$ (dashed) and analytic uniform (dotted) bounds from Eq.~\eqref{eq:DC-Zdet} in the overdamped ($\gamma_\varphi=500$; left, blue) and underdamped ($\gamma_\varphi=50$; right, red) setting. We choose $t_1=0$ such that $\Delta\tau = t_2$. \label{fig:DAM-bounds-Delta-zero}}
\end{figure*}

For zero detuning, $\Delta=0$, but finite dephasing $\gamma_\varphi>0$, the dynamics of the coherence $y(t)$ crosses over from underdamped Rabi-like oscillations to overdamped exponential relaxation as $\gamma_\varphi$ increases. In Appendix~\ref{sec:zero-det-DA} we solve the corresponding equations of motion analytically and obtain explicit expressions for $\mathcal{C}_M(t)$, $\Gamma_M(t)$, and the induced variation bounds \eqref{eq:variation-bound} for a generic POVM element $M$ in Eq.~\eqref{eq:M-DAM} in the three regimes: underdamped ($\gamma_\varphi<4|J|$), critical ($\gamma_\varphi=4|J|$), and overdamped ($\gamma_\varphi>4|J|$). 
In both cases, the resource impact functional $\mathcal{C}_{|A\rangle\!\langle A|}(\Lambda_t)$ can be written in closed form. In the underdamped regime, $\gamma_\varphi<4|J|$,
\begin{equation}
\mathcal{C}_{|A\rangle\!\langle A|}(\Lambda_t) = e^{-(\gamma+\gamma_\varphi/2)t}\frac{J}{\omega}\,|\sin(\omega t)|\,,
\end{equation}
where $\omega:=\sqrt{4J^2-\gamma_\varphi^2/4}$, whereas in the overdamped regime, $\gamma_\varphi>4|J|$, we have 
\begin{equation}
\mathcal{C}_{|A\rangle\!\langle A|}(\Lambda_t) = e^{-(\gamma+\gamma_\varphi/2)t}\,\frac{J}{\kappa}\,\sinh(\kappa t)\,,
\end{equation}
where $\kappa:=\sqrt{\gamma_\varphi^2/4-4J^2}$. 
Thus, underdamped dynamics exhibits damped coherent oscillations with nodes at $t_n=n\pi/\omega, n\in \ZZ$, while overdamped dynamics removes these oscillations and replaces them by a single non-oscillatory peak followed by decay.

Fig.~\ref{fig:DAM-bounds-Delta-zero} illustrates this behaviour of the time-dependent change $|\Delta\mathcal{C}_{|A\rangle\!\langle A|}|$ in the resource impact functional for $M = \ket{A}\!\bra{A}$, both in the underdamped regime (right; dark red, solid) and in the overdamped regime (left; dark blue, solid). As representative parameters we choose $J=100$ and $\gamma=5$ (as in Fig.~\ref{fig:DAM-bounds}), together with $\gamma_\varphi=50$ in the underdamped case and $\gamma_\varphi=500$ in the overdamped case. We set $t_1=0$ such that $|\Delta\mathcal{C}_M| = |\mathcal{C}_M(\Lambda_{t_2})-\mathcal{C}_M(\Lambda_{0})|$ reduces to $|\mathcal{C}_M(\Lambda_{t_2})|$, and $\Delta\tau=t_2-t_1=t_2$. 

In Appendix~\ref{sec:zero-det-DA} we further show that, for $M=\ket{A}\!\bra{A}$, the maximal achievable change in the resource impact functional over any interval $[t_1,t_2]$ is bounded from above by simple functions of the system parameters $J$, $\gamma_\varphi$, and $\gamma$. These lead to the analytic uniform bounds given in Eq.~\eqref{eq:DC-Zdet}, which provide closed form upper bounds on $|\Delta\mathcal{C}_M|$ obtained by bounding the integral in Eq.~\eqref{eq:variation-bound}. In Fig.~\ref{fig:DAM-bounds-Delta-zero} these uniform bounds are shown by the light red dotted (underdamped, right panel) and light blue dotted (overdamped, left panel) curves. For comparison, the numerically integrated variation bounds $\int_0^t \Gamma_M(s)\,\mathrm{d}s$ are shown as medium red dashed and medium blue dashed curves.

In the underdamped case, the exact change $|\Delta\mathcal{C}_M|$ oscillates because coherent donor-acceptor exchange repeatedly changes the sign of the underlying coherence contribution before the absolute value is taken. The integrated variation bound is therefore most informative at small times. At later times it necessarily overestimates the exact change because it accumulates the magnitude of each oscillatory increase and decrease. In contrast, the overdamped case shows no such nodes, so the exact curve is smoother and the variation bound remains comparatively more representative over the same time window.
Moreover, stronger dephasing suppresses the coherent back-and-forth oscillations that produce the minima and zeros of the underdamped solution. As a result, a more strongly dephased curve can exceed a weakly dephased one at selected times, especially near the oscillatory minima of the latter. This is a redistribution in time of the coherence sensitivity, not a monotonic enhancement of the overall coherence-induced advantage by dephasing.

In the present zero-detuning setting, the universal variation bound from Theorem~\ref{thm:variation-bound} can be complemented by a sharper, model-specific estimate that preserves oscillatory cancellations. Since $\mathcal{C}_{|A\rangle\!\langle A|}(\Lambda_t)=|y(t)|$ and $y(t)$ is known in closed form, one has for any interval $[t_1,t_2]$
\begin{equation}
\left|\Delta\mathcal{C}_{|A\rangle\!\langle A|}(\Lambda_{t_2},\Lambda_{t_1})\right| = \left||y(t_2)|-|y(t_1)|\right| \leq |y(t_2)-y(t_1)|\,.
\end{equation}
Unlike the integrated bound $\int_{t_1}^{t_2}\Gamma_M(s)\,\mathrm ds$, this estimate retains the damped oscillatory structure of the exact Heisenberg solution and therefore does not discard cancellations. For the intervals used in Fig.~\ref{fig:DAM-bounds-Delta-zero}, where $t_1=0$, it reduces to the exact result $|\Delta\mathcal{C}_M| = \mathcal{C}_M(\Lambda_{t_2})$.

Finally, as discussed in Appendix~\ref{sec:zero-det-DA}, the uniform bounds on $|\Delta\mathcal{C}_M|$ imply feasibility bounds in Eq.~\eqref{eq:feasibility-bound-Zdet}, which further constrain which target values of $|\Delta\mathcal{C}_M|$ are reachable within the model within a finite time window. In particular, for fixed $J$ and $\gamma$, increasing the dephasing rate $\gamma_\varphi$ not only suppresses the maximal coherence-induced advantage but also shrinks the time interval over which a given $|\Delta\mathcal{C}_M|$ can be realised (see also Fig.~\ref{fig:DAM-bounds-Delta-zero}).
From the exciton transport perspective, site basis coherence can enhance the donor-to-acceptor transfer yield in this minimal dimer, but only within strict quantitative limits set by the coupling strength and the environmental rates. The quantity $\mathcal{C}_M(\Lambda_t)$ identifies the maximal coherence-assisted improvement of the transfer efficiency at time $t$, while $\Gamma_M(t)$ and the associated time and feasibility bounds constrain how rapidly such an advantage can build up or decay.

\section{Conclusions and outlook \label{sec:concl}}

\begin{table*}
\centering
\begin{tabular}{| c | c |}
\hline
\makecell{\textbf{Quantity}}& \makecell{\textbf{Properties}}\\
\hline
\makecell{Resource impact functional\\ $\mathcal{C}_M(\Lambda)$}
& \makecell{\textbullet\ Convexity, continuity, seminorm (Theorem \ref{thm:capacity-properties});\\ \textbullet\   Data-processing inequalities (Theorem \ref{thm:data-processing}); \\\textbullet\ Geometric and operational interpretations  (Corollary \ref{cor:capacity-geometry}, \ref{cor:CM-polar}, Sec.~\ref{sec:op-interp}); \\ \textbullet\ For linear $\mathcal{G}$: operator-norm form (Thm.~\ref{thm:capacity-G-linear}), free/resourceful decomposition \\ (Theorem \ref{thm:dissect-channel}, Corollary \ref{cor:gen-free-compatibility}). }\\
\hline
\makecell{Unpaired benchmark:\\ $\Pi_M(\Lambda)$}& \makecell{\textbullet\ Breaks the pairing $\rho\mapsto\mathcal{G}(\rho)$;\\ \textbullet\ Bounded from above by $\mathcal{C}_M(\Lambda)$ (Corollary \ref{cor:CM-PiM}) } \\
\hline
\makecell{Instantaneous rate \\$\Gamma_M(t)$}& \makecell{\textbullet\ Variation bound and vanishing criteria (Theorem \ref{thm:variation-bound}, Lemma \ref{lem:Gamma-M-zero}) \\ $\Downarrow$ \\ $\bullet$ Time \& feasibility bounds (Corollary \ref{cor:time-bound})} \\  
\hline
\makecell{Generalized functionals:\\ $\mathfrak{C}_M(\Lambda)$, $\mathfrak{T}_M(t)$} & \makecell{Generalization of $\mathcal{C}_M(\Lambda), \Gamma_M(t)$ to QRTs without resource-destroying maps} \\
\hline
\end{tabular}
\caption{Overview of the quantities introduced in this work and their key properties. \label{tab:summary}}
\end{table*}

Motivated by the goal of quantifying when quantum resources have an operational impact in chemical processes, we have introduced the resource impact functional $\mathcal{C}_M(\Lambda)$ (recall Eq.~\eqref{eq:capacity-sup}). To summarize our results, Table \ref{tab:summary} lists the central quantities introduced in this work together with their key properties.
For a fixed quantum channel $\Lambda$ and observable $M$, $\mathcal{C}_M(\Lambda)$ quantifies the maximal change in the expectation value of $M$ that can be induced by replacing an input state $\rho$ with its resource-free counterpart $\mathcal{G}(\rho)$ in a single application of $\Lambda$. 
This construction is based on a quantum-classical baseline provided by a resource-destroying map $\mathcal{G}$, which pairs each state $\rho$ with the free state $\mathcal{G}(\rho)$ obtained by removing the chosen resource. 
We established various key properties of $\mathcal{C}_M(\Lambda)$, including convexity, continuity, and pre- and post-processing inequalities, and discussed their practical relevance. 
A particular emphasis was placed on comparing $\mathcal{C}_M(\Lambda)$ with alternative, ``unpaired'' constructions that optimise independently over all resourceful and all free inputs (e.g., see Eq.~\eqref{eq:def-PiML}). 
Together with the geometric and operational interpretations in Secs.~\ref{sec:CM-geometry} and \ref{sec:op-interp}, this clarifies the scope of $\mathcal{C}_M(\Lambda)$: for a given process $(\Lambda,M)$, it quantifies the largest resource-induced change in the target signal relative to the paired free baseline $\mathcal{G}(\rho)$. In this way, $\mathcal{C}_M(\Lambda)$ characterises how the resource content of the input can affect the dynamics and, consequently, the chosen figure of merit. While experimentally motivated variants can be formulated by restricting the optimization to relevant classes of initial states, the present definition of $\mathcal{C}_M(\Lambda)$ retains the important advantage of providing a uniform state-independent upper bound.
Since not all resource theories admit a resource-destroying map, we also explained in Sec.~\ref{sec:RT-wo-RDM} how the definition of $\mathcal{C}_M(\Lambda)$ can be generalised in such cases, and how this slightly shifts its operational interpretation.

The identification of resource measures with maximal advantages in suitable discrimination tasks is by now standard in quantum resource theory (see, e.g., Refs.~\cite{NBCPJA16, TR19, LBL20}). 
In our setting, $\mathcal{C}_M(\Lambda)$ is precisely twice the maximal single-shot bias above random guessing in distinguishing, using the fixed process $(\Lambda,M)$, whether the input was $\rho$ or its resource-free counterpart $\mathcal{G}(\rho)$. 
This follows directly from standard results in binary hypothesis testing and quantum state discrimination. 
While the functional $\mathcal{C}_M(\Lambda)$ itself appears to be new, its operational interpretation therefore fits naturally into this broader discrimination-based characterisation of quantum resources. 

This perspective also opens the door to further questions, such as trade-offs between different tasks. 
Given two observables $M_1$ and $M_2$, one may ask whether there exists a state that (approximately) maximises the yield gain for both tasks simultaneously. 
In general, this is not expected: for linear $\mathcal{G}$, the operators $B_{M,\Lambda}$ entering $\mathcal{C}_M(\Lambda) = \|B_{M,\Lambda}\|_{\infty}$ (Theorem~\ref{thm:capacity-G-linear}, Sec.~\ref{sec:op-interp}) need not have a common eigenbasis, so the states that optimise $\mathcal{C}_{M_1}(\Lambda)$ and $\mathcal{C}_{M_2}(\Lambda)$ may be incompatible. 
The tools developed here provide a natural language to analyse such incompatibilities of optimal resources across tasks. This is particularly relevant in chemical settings, where different figures of merit, e.g., competing reaction yields, selectivities, or signal contrasts, may be sensitive to different resource characteristics, and any functional role of quantum resources is likely to involve compromises or synergies between more than a single task or resource type.

To bound how quickly a resource can modify a yield, we introduced the instantaneous rate $\Gamma_M(t)$ and showed that integral bounds in terms of $\Gamma_M(t)$ constrain how fast the resource impact functional can change along a given dynamics (Sec.~\ref{subsec:dynamics}), independent of microscopic details. 
In particular, uniform bounds on $\Gamma_M(t)$ lead to lower bounds on the time required to achieve a prescribed change in $\mathcal{C}_M(\Lambda_t)$, providing ``activation time'' guarantees that play the role of resource-aware analogues of quantum speed limits (Sec.~\ref{subsec:time-bound}). It will be interesting in future work to explore these bounds in more complex models, and to investigate whether they can be saturated or tightened for specific classes of dynamics.

Moreover, we introduced a mathematically well-defined and physically motivated decomposition of time-local master equations into free and resourceful components and showed that only the resourceful component $\Lambda_{\mathrm{res}}$ of a channel $\Lambda$ contributes to the resource impact functional, i.e., $\mathcal{C}_M(\Lambda)=\mathcal{C}_M(\Lambda_{\mathrm{res}})$ (see Sec.~\ref{sec:dissect-dyn}). 
This is significant because it singles out, within an arbitrary open system evolution, those parts of the dynamics that are responsible for exploiting the resource in the sense of changing the measurement statistics for $M$. 
This may be particularly valuable in more complicated scenarios where it is not obvious a priori which terms in the generator are responsible for an enhanced yield or a non-zero $\mathcal{C}_M(\Lambda)$.

The donor-acceptor example in Sec.~\ref{sec:DAM} provides a closed form evaluation of the resource impact functional $\mathcal{C}_M(\Lambda_t)$ and the instantaneous rate $\Gamma_M(t)$ within a standard GKSL description of exciton dynamics. In this setting, $\mathcal{C}_M(\Lambda_t)$ yields a state-independent upper bound on the maximal coherence-induced deviation of a specified readout at time $t$ from its dephased baseline, thereby identifying the time windows in which coherence can be operationally relevant for the task at hand.
Excitation energy transfer models are a particularly instructive because experimentally relevant parameters often lie between the fully coherent and fully incoherent limits: in many molecular aggregates and excitonic network models, the electronic coupling and the effective system-bath interaction are comparable, and bath relaxation can occur on timescales comparable to the transfer dynamics \cite{CF09, IF09-PNAS, IF09-Redfield, IF09-JCP, SIFW10, IF12, CS15, LLHC16, MOAGGS17, JZMTMD26}. The framework developed here is well suited to such intermediate coupling and finite memory regimes because its central quantities are defined directly in terms of reduced dynamical maps and therefore do not rely on a Markovian GKSL description. It is thus compatible, in principle, with non-Markovian approaches such as the hierarchical equations of motion, which will be an important direction for future work.

Finally, non-Markovian dynamics and multi-time processes are crucial for realistic chemical applications, where a sequence of control operations, such pulse sequences in ultrafast spectroscopy, interacts with a quantum system. 
To capture such multi-time behaviour, one may extend the present framework from single-step channels to process tensors \cite{PRFPM18, PMRFP18, MPM19, JP19, MM20, dWKLC25, CG25, KSBR25, SCGISN25}. 
As a concrete example, one could construct the process tensor of a model reaction using a numerically exact time-dependent method such as time-dependent density matrix renormalisation group (TD-DMRG), where the process tensor admits a compact matrix product operator representation \cite{JP19, CG25, KSBR25}. 
This would allow one to generalise the resource impact functional to a functional of the process tensor, thereby quantifying how much a given quantum resource can influence not just a single application of $\Lambda$ but an entire class of control sequences acting on the system.
In parallel, the same process tensor could be reconstructed (approximately) from experimental data, thus, extending ideas of quantum process tomography applied to ultrafast spectroscopy signals \cite{YZAG11,YZKMAG11,CM13,PMAG15,GH19}. 
A comparison between the TD-DMRG derived and experimentally reconstructed process tensors would then provide a way of benchmarking model Hamiltonians and open system descriptions of vibrationally coupled electron transfer, such as those studied in Ref.~\cite{WBLCS25}. 
In particular, one could (i) quantify how strongly quantum resources, e.g., coherence and memory effects, contribute to the experimentally observed dynamics, (ii) identify which model ingredients are essential to reproduce the resource impact seen in experiment, and (iii) assess the limitations of approximate Markovian or coarse grained descriptions from the perspective of their ability to capture resource-induced effects.

\begin{acknowledgments}
This research was funded by the National Science Foundation under Grant No.~CHE-$2537080$.
\end{acknowledgments}

\appendix
\section{Recap of quantum resource theories\label{app:QRT}}

In the following, we briefly summarize various concepts from quantum resource theory (QRT) which are used in the main text. There are many good references for QRTs, and we refer in particular to Refs.~\cite{CG19, Gour2025QRT} for more information.

\subsection{Free states and free operations\label{app:free-set-op}}

Suppose only a subset of quantum states can be prepared at negligible cost, referred to as the \textit{free states}. Then, the general idea of a QRT is to partition physically allowed transformations into two classes: those that are implementable essentially for free, called \textit{free operations}, and all other, costly transformations. Free operations are required to be compatible with the free states in the sense that they never generate a resource from a free input. One can specify a QRT starting either from a physically motivated set of free states and then identifying the operations that preserve it, or from a constrained class of operations and then defining the corresponding free states. In particular, the set of free operations may be chosen to be smaller than the maximum set of admissible free operations, referred to as \textit{resource non-generating operations}, that leaves the set of free states invariant. The appropriate starting point depends on the context, and in all cases the sets of free states and free operations must be mutually consistent.

Let $\mathcal{H}$ be a finite-dimensional complex Hilbert space of dimension $d=\mathrm{dim}(\mathcal{H})$, and denote by $\mathcal{D}(\mathcal{H})$ the set of density operators. Since in finite dimensions, positivity of an operator implies its Hermiticity, $\rho\in \mathcal{D}(\mathcal{H})$ if and only if $\Tr[\rho]=1$ and $\rho\geq 0$. Thus,
\begin{equation}
\mathcal{D}(\mathcal{H}) := \left\{ \rho \in \mathcal{B}(\mathcal{H}) \middle\vert \rho \geq 0,\ \Tr[\rho]=1 \right\}\,,
\end{equation}
where $\mathcal{B}(\mathcal{H})$ denotes the set of bounded linear operators on $\mathcal{H}$. 

For systems $A$ and $B$ with Hilbert spaces $\mathcal{H}_A$ and $\mathcal{H}_B$, we write $\mathcal{D}_A \equiv \mathcal{D}(\mathcal{H}_A)$ and $\mathcal{D}_B \equiv \mathcal{D}(\mathcal{H}_B)$. A set of free states is a subset $\mathcal{F}_A \subseteq \mathcal{D}_A$ (and similarly $\mathcal{F}_B \subseteq \mathcal{D}_B$).
Free operations are modelled as completely positive trace-preserving (CPTP) linear maps between operator algebras. Let
\begin{equation}
\mathrm{CPTP}_{A\to B} := \left\{ \Lambda:\mathcal{B}(\mathcal{H}_A)\to\mathcal{B}(\mathcal{H}_B)\ \text{linear, CPTP}\right\}\,.
\end{equation}
A collection $\mathcal{O}_{A\to B}\subseteq \mathrm{CPTP}_{A\to B}$ is a valid class of free operations (relative to the sets of free states $\mathcal{F}_A,\mathcal{F}_B$) if it satisfies:
\begin{itemize}
\item[(i)] inclusion of the identity map: $\mathrm{id}_A \in \mathcal{O}_{A\to A}$, 
\item[(ii)] closure under composition: if $\Lambda_1\in\mathcal{O}_{A\to C}$ and $\Lambda_2\in\mathcal{O}_{C\to B}$, then $\Lambda_2\circ \Lambda_1\in\mathcal{O}_{A\to B}$,
\item[(iii)] preservation of the set of free states: $\Lambda(\rho)\in\mathcal{F}_B$ for all $\rho\in\mathcal{F}_A$ and all $\Lambda\in\mathcal{O}_{A\to B}$.
\end{itemize}
A QRT is then specified by the pair $(\mathcal{F},\mathcal{O})$ of free states and free operations.
In addition, most resource theories also impose that the trace map $\Tr: \mathcal{B}(\mathcal{H})\to\CC$ is contained in the set of free operations. Furthermore, the set of resource non-generating operations is defined as
\begin{equation}
\mathcal{O}^{(\max)}_{A\to B}:= \left\{
\Lambda\in\mathrm{CPTP}_{A\to B}\,\middle\vert\,\forall\rho\in\mathcal{F}_A: \Lambda(\rho)\in\mathcal{F}_B
 \right\}
\end{equation}
and we may choose $\mathcal{O}_{A\to B}\subseteq \mathcal{O}^{(\max)}_{A\to B}$. Conversely, the minimal set $\mathcal{F}_{A}^{(\min)}$ of free states compatible with a fixed class $\mathcal{O}_{A\to B}$ is given by
\begin{equation}
\mathcal{F}_{A}^{(\min)}:= \left\{\rho\in\mathcal{D}_A\,\middle\vert\,\forall\sigma\in \mathcal{D}_A\,\exists\Lambda\in \mathcal{O}_{A\to A}: \rho = \Lambda(\sigma)\right\}\,.
\end{equation}

In addition to the above minimal properties, many $(\mathcal{F},\mathcal{O})$ admit additional structures. Following Ref.~\cite{CG19}, a QRT $(\mathcal{F},\mathcal{O})$ is called \emph{convex} if, for any systems $A,B$, the set of free operations $\mathcal{O}(A\to B)$ is convex, i.e., $p\Phi+(1-p)\Lambda\in\mathcal{O}(A\to B)$ for all $\Phi,\Lambda\in\mathcal{O}(A\to B)$ and $p\in[0,1]$.
In the formulation where free states are defined as free preparations $\mathcal{F}(\mathcal{H})=\mathcal{O}(\CC\to \mathcal{H})$ from a trivial system $\CC$, convexity of $\mathcal{O}$ implies that each $\mathcal{F}(\mathcal{H})$ is convex, while the converse need not hold.
Moreover, a resource theory $(\mathcal{F},\mathcal{O})$ is called \textit{affine} if, for any systems $A,B$, every CPTP map that can be expressed as an affine combination of free operations is itself free, i.e., if $\Phi=\sum_i \alpha_i \Phi_i$ with $\Phi_i\in\mathcal{O}(A\to B)$ and $\sum_i\alpha_i=1$, and if $\Phi$ is CPTP, then $\Phi\in\mathcal{O}(A\to B)$.
This condition implies that the free-state set is closed under affine combinations that yield a valid density operator, but this property alone does not imply that the QRT is affine.
Many widely studied QRTs, including entanglement, coherence, asymmetry, and athermality, are convex, whereas entanglement theory provides a canonical example of a convex but non-affine QRT.

\subsection{Resource destroying maps\label{app:RDP}}

Many resource theories come with resource destroying maps, which play a major role in the main text as they allow us to define a clear classical baseline with respect to which we can compare the effect of a given resource on the yield for a process described by a quantum channel $\Lambda$. 

Let $(\mathcal{F}, \mathcal{O})$ describe a QRT with $\mathcal{F} = \mathcal{F}(\mathcal{H})$. Then, a map $\mathcal{G}:\mathcal{B}(\mathcal{H})\to\mathcal{B}(\mathcal{H})$ is called resource destroying if 
\begin{itemize}
\item[(i)] $\mathcal{G}(\rho)=\rho$ for all $\rho\in\mathcal{F}$, and
\item[(ii)] $\mathcal{G}(\rho) \in\mathcal{F}$ for all $\rho\in\mathcal{D}(\mathcal{H})$.
\end{itemize}
While $\mathcal{G}$ may be linear for convex resource theories, it must be non-linear for non-convex ones \cite{LHL17}. Consequently, in non-convex QRTs $\mathcal{G}$ cannot be a quantum channel. Necessary and sufficient conditions for the existence and uniqueness of a \textit{linear} resource-destroying map have been derived in Ref.~\cite{Gour17}. In particular, for a linear $\mathcal{G}$ to exist, the set of free states $\mathcal{F}$ has to be affine, while the converse does not hold and not every affine resource theory has a linear resource-destroying map \cite{Gour17}. 
Furthermore, condition (ii) implies that $\mathcal{G}$ is idempotent, i.e., $\mathcal{G}\circ\mathcal{G} = \mathcal{G}$. 

Two canonical examples of resource destroying maps are the unitary group twirl for asymmetry and the dephasing map for coherence. 
For asymmetry, the resource destroying map is given by the corresponding group twirling operation. For a compact group $G$ with unitary representation $\Pi: G\to \text{U}(\mathcal{H})$ on a Hilbert space $\mathcal{H}$ and uniform Haar measure $\mathrm{d}\mu(g)$, 
\begin{equation}
\mathcal{G}_{\text{asy}}(\cdot) = \int\mathrm{d}\mu(g)\, U(g)(\cdot) U^\dagger(g)\,.
\end{equation}
For finite groups, the above integral is replaced by a discrete sum according to
\begin{equation}
\mathcal{G}_{\text{asy}}(\cdot) =\frac{1}{|G|}\sum_{g\in G} U(g)(\cdot) U^\dagger(g)\,.
\end{equation}

For coherence, we fix an orthonormal basis $\mathcal{B}=\{\ket{i}\}_{i=1}^d$ of $\mathcal{H}$ and take the free set $\mathcal{F}(\mathcal{H})$ to be the diagonal density operators in that basis. The associated resource destroying map is the dephasing channel, which acts on $\rho\in\mathcal{D}(\mathcal{H})$ as
\begin{equation}
\mathcal{G}_{\text{coh}}(\rho) = \sum_{i=1}^d  \rho_{ii}\ket{i}\!\bra{i}\,,\quad \rho_{ii} = \bra{i}\rho\ket{i}\,\forall i\in [d]\,.
\end{equation}
The map $\mathcal{G}_{\mathrm{coh}}$ is CPTP, unital, idempotent, and Hermitian with respect to the Hilbert–Schmidt inner product. 

While the two resource-destroying maps $\mathcal{G}_{\mathrm{asym}}$ and $\mathcal{G}_{\mathrm{coh}}$ are both self-adjoint, not every resource-destroying map satisfies this property. An example of a resource theory with a linear resource-destroying map that is not self-adjoint is the singleton free set $\mathcal{F} = \{\sigma\}$, for which the replacement channel
\begin{equation}\label{eq:replacement-channel}
\mathcal{G}(\rho)\equiv \Lambda_\sigma(\rho) := \Tr[\rho]\,\sigma
\end{equation}
is a valid (and unique) resource-destroying map. 
Its Hilbert-Schmidt adjoint is $\Lambda_\sigma^\dagger(A) = \Tr[A\sigma]\,\mathbbm{1}$, such that $\Lambda_\sigma$ is self-adjoint if and only if $\sigma = \mathbbm{1}/d$ with $d=\dim(\mathcal{H})$.

\section{Recap of time-local master equations\label{app:time-local}}

In this section, we briefly recall time-local master equations and indicate when they exist. A (reduced) quantum dynamical evolution is described by a family of linear maps $\{\Lambda_t\}_{t\geq 0}$ acting on density operators with $\Lambda_0=\mathrm{id}$ and $\rho(t) = \Lambda_t(\rho(0))$. The evolution is called time local if the density operator $\rho(t)$ of the system satisfies
\begin{equation}
\frac{\mathrm{d}\rho(t)}{\mathrm{d}t}=\mathcal{L}_t\left(\rho(t)\right)\,,
\end{equation}
with a (possibly time-dependent) time-local generator $\mathcal{L}_t$. Two standard Markovian regimes are then distinguished \cite{RHP10, CK10, HCLA14, BLPV16, dVA17, SLHK19}:
\begin{itemize}
\item[(i)] a time-independent Lindblad generator $\mathcal{L}_t\equiv\mathcal{L}$, yielding a quantum dynamical semigroup
\begin{equation}
\Lambda_{t,0}=e^{t\mathcal{L}}\,;
\end{equation}
\item[(ii)] a general time-local generator $\mathcal{L}_t$ of Gorini-Kossakowski-Sudarshan-Lindblad (GKSL) form at each time $t$, commonly referred to as time-inhomogeneous Markovian dynamics.
\end{itemize}

A convenient microscopic setting starts from a closed system-environment Hamiltonian
\begin{equation}
H = H_S+H_E+\kappa H_I\in\mathcal{B}(\mathcal{H}_S\otimes\mathcal{H}_E)\,,
\end{equation}
where $H_S$ and $H_E$ are the system and environment Hamiltonians acting only non-trivially on the system and environment Hilbert spaces $\mathcal{H}_S, \mathcal{H}_E$, respectively, $H_I$ describes their interaction, and $\kappa$ is a coupling parameter. In the weak-coupling (Born-Markov) regime together with a secular (rotating wave) approximation, one obtains the time-independent GKSL master equation for the reduced state $\rho(t)$ (setting $\hbar=1$):
\begin{equation}\label{eq:GKLS}
\frac{\mathrm{d}\rho(t)}{\mathrm{d}t} = -i[\widetilde{H},\rho]
+\sum_k \gamma_k\left(L_k\rho L_k^\dagger-\frac{1}{2}\left\{L_k^\dagger L_k,\rho\right\}\right)\,,
\end{equation}
where $\widetilde{H}:=  H_S+H_L$ includes the Lamb-shift Hamiltonian $H_L$ which, in the weak-coupling limit, commutes with $H_S$ and encodes the environment-induced renormalization of the system's energy levels. Moreover, $\gamma_k\geq 0$ are the so-called decoherence rates, and the Lindblad operators $L_k$ are determined by $H_I$ and the bath correlation functions. Furthermore, Eq.~\eqref{eq:GKLS} is commonly expressed in the form
\begin{equation}
\frac{\mathrm{d}\rho(t)}{\mathrm{d}t} = -i[\widetilde{H},\rho(t)]+\mathcal{D}(\rho(t)) = \mathcal{L}(\rho(t)),
\end{equation}
where $\mathcal{D}(\cdot)$ is the dissipator and $\mathcal{L}(\cdot)$ the time-independent Lindblad generator. Moreover, $\mathcal{L}$ is the Liouvillian superoperator acting on the vectorized density operator $|\rho\rangle\!\rangle$, where we use the convention (row-stacking)
\begin{equation}\label{eq:row-stacking}
\mathrm{vec}\!\left(A\rho B\right)=\left(A\otimes B^{\mathrm T}\right)|\rho\rangle\!\rangle\,.
\end{equation}
Then, the Lindblad generator in Liouvillian space, $\mathcal{H}_S\otimes\mathcal{H}_S$, reads
\begin{align}
\mathcal{L} &= -i\left(\widetilde{H}\otimes \mathbbm{1} - \mathbbm{1}\otimes \widetilde{H}^{\mathrm{T}} \right)+ \sum_k \gamma_k\bigg[
L_k \otimes L_k^{*} \nonumber\\
&\qquad  - \frac{1}{2}\left(L_k^\dagger L_k\otimes \mathbbm{1}
+\mathbbm{1}\otimes \left(L_k^\dagger L_k \right)^\mathrm{T}\right)
\bigg]\,.
\end{align}
This operator representation allows one to directly read off the Heisenberg adjoint $\mathcal{L}^\dagger(\cdot)$ with respect to the Hilbert-Schmidt inner product as used in the main text. 

In general, the reduced time evolution of a system is described by a family of \textit{dynamical maps}. Under the assumption that the joint system-environment density operator $\rho_{SE}\in\mathcal{D}(\mathcal{H}_S\otimes\mathcal{H}_E)$ at time $t=0$ is a product state, i.e., $\rho_{SE}(0) = \rho(0)\otimes \rho_E(0)$, the dynamical map $\Lambda_t$ can be defined by tracing out the environmental degrees of freedom of $\rho_{SE}(t)$ under unitary time evolution leading to
\begin{equation}
\rho(t) = \Lambda_t\left(\rho(0)\right)\,,
\end{equation}
and $\Lambda_t$ is trace preserving and completely positive and, thus, a valid quantum channel. 
A dynamical process on a time interval $[0,T]$ is thus described by a one-parameter family $\{\Lambda_t\}_{t\in [0, T]}$ of dynamical maps $\Lambda_t$. 
It can be shown that differentiability of $t\mapsto \Lambda_t$ and existence of the inverse map $\Lambda_t^{-1}$ for all $t\in[0,T]$ ensure the existence of a time-local master equation with generator \cite{HCLA14}
\begin{equation}
\mathcal{L}_t = \left(\frac{\mathrm{d}}{\mathrm{d}t}\Lambda_t\right)\circ \Lambda_t^{-1}\,,
\end{equation}
i.e.,
\begin{align}\label{eq:time-local-master}
\frac{\mathrm{d}\rho(t)}{\mathrm{d}t} & = \mathcal{L}_t(\rho(t))\\
&= -i[\widetilde{H}(t),\rho(t)]
+\sum_k \gamma_k(t)\Big(L_k(t)\rho L_k^\dagger(t)\nonumber\\
&\quad-\frac{1}{2}\left\{L_k^\dagger(t) L_k(t),\rho(t)\right\}\Big)\,,\nonumber
\end{align}
where in contrast to Eq.~\eqref{eq:GKLS}, $\widetilde{H}, \gamma_k, L_k$ are in general time-dependent. Unlike the time-independent GKSL case \eqref{eq:GKLS}, the general time-local master equation \eqref{eq:time-local-master} does not necessarily correspond to a dynamical semigroup anymore since $\mathcal{L}_t$ is time-dependent. Moreover, the decoherence rates $\gamma_k(t)$ may be non-positive for intermediate times. 

A commonly used stronger notion of quantum Markovianity is \emph{complete-positivity (CP) divisibility}.
A process $\{\Lambda_t\}_{t\in[0,T]}$ is called CP divisible if for all $t\ge s$ there exists a CPTP map $V_{t,s}$ such that
\begin{equation}
\Lambda_t = V_{t,s}\circ \Lambda_s .
\end{equation}
If $\Lambda_s$ is invertible as a linear map, then $V_{t,s}$ is uniquely determined and can be written as
$V_{t,s}=\Lambda_t\circ \Lambda_s^{-1}$. Moreover, invertibility alone does not guarantee that $V_{t,s}$ is completely positive, since $\Lambda_s^{-1}$ need not be completely positive. 
For sufficiently regular finite-dimensional dynamics, CP divisibility is equivalent to the time-local generator $\mathcal{L}_t$ being of GKSL form at each time $t$, i.e., to the non-negativity of the decoherence rates $\gamma_k(t)$ for all $t$ \cite{BLPV16}.

\section{Proof of Theorem \ref{thm:capacity-G-linear} \label{app:proof-thm-1}}

\begin{proof}
For $\mathcal{G}$ linear, $\Delta Y = |\Tr[\widetilde{X}\rho]|$ with $\widetilde{X}:=(\text{id}-\mathcal{G}^\dagger)(\Lambda^\dagger(M))$ is convex in $\rho$ (but not necessarily affine, as otherwise the proof of (i) would be trivial). 
Since 
\begin{equation}
|z| =\sup_{\theta\in[0, 2\pi)}\Re\left(\mathrm{e}^{-i\theta}z\right)
\end{equation}
for any complex number number $z\in \CC$, we have (recall that $\mathcal{H}$ finite-dimensional)
\begin{align}
\mathcal{C}_M(\Lambda) &= \max_{\rho\in\mathcal{D}(\mathcal{H})}\left\vert\Tr\left[\widetilde{X}\rho\right]\right\vert \nonumber\\
&= \sup_{\theta\in [0, 2 \pi)} \max_{\rho\in\mathcal{D}(\mathcal{H})} \Re\left(\Tr\left[\mathrm{e}^{-i\theta}\widetilde{X}\rho\right]\right) \nonumber\\
&= \sup_{\theta\in [0, 2 \pi)} \max_{\rho\in\mathcal{D}(\mathcal{H})} \Tr\left[\widetilde{X}_{\theta}\rho\right]\,,
\end{align}
where $\widetilde{X}_\theta: = \Re(\mathrm{exp}(-i\theta) \widetilde{X})$ is Hermitian and we used the linearity of the trace. Since for any fixed angle $\theta\in [0, 2 \pi)$ the expectation $\Tr[\widetilde{X}_{\theta}\rho] $ is linear, it follows that
\begin{align}
\mathcal{C}_M(\Lambda) &=  \sup_{\theta\in [0, 2 \pi)}\max_{\substack{[\Psi]\in\mathbbm{P}(\mathcal{H})\,\\ \|\Psi\|=1}} \bra{\Psi}\widetilde{X}_{\theta}\ket{\Psi}\nonumber\\
&= \max_{\substack{[\Psi]\in\mathbbm{P}(\mathcal{H})\,\\ \|\Psi\|=1}}\left\vert \bra{\Psi}\widetilde{X}\ket{\Psi}\right\vert\,,
\end{align}
where we used in the last line that $\theta$ can always be chosen to align with the phase of $\ket{\Psi}$.

Next, we prove (ii). Since $\mathcal{G}$ is Hermiticity preserving and $M\in\mathcal{B}_{\mathrm{sa}}(\mathcal{H})$, the operator $\widetilde{X}$ is Hermitian (for $\mathcal{G}$ CPTP this can also be shown explicitly by using the Kraus representation of $\mathcal{G}$). We denote the maximal and minimal eigenvalues of $\widetilde{X}$ by $\lambda_{\max}(\widetilde{X})$ and $\lambda_{\min}(\widetilde{X})$, respectively. Then, 
\begin{align}
\mathcal{C}_M(\Lambda) 
& = \max\left\{\left\vert\lambda_{\max}(\widetilde{X})\right\vert, \left\vert\lambda_{\min}(\widetilde{X})\right\vert\right\} = \|\widetilde{X}\|_{\infty}\,,
\end{align}
where we used that for $\widetilde{X}$ Hermitian, the operator ($\infty-$Schatten) norm equals the spectral norm \cite{Bhatia-MatrixAnalysis}. 

It remains to show (iii). Since $\mathcal{G}$ is idempotent, which holds for any resource-destroying map, $\mathcal{B}(\mathcal{H})$ can be decomposed into a direct sum
\begin{equation}
\mathcal{B}(\mathcal{H})=\mathrm{im}(\mathcal{G}^\dagger)\oplus \mathrm{ker}(\mathcal{G}^\dagger)\,,
\end{equation}
where $\mathrm{im}(\mathcal{G}^\dagger)$ and $ \mathrm{ker}(\mathcal{G}^\dagger)$ denote the image (range) and kernel of $\mathcal{G}^\dagger$, respectively. Thus, for $\Lambda^\dagger(M)\in\mathrm{im}(\mathcal{G}^\dagger)$, $\mathcal{C}_M(\Lambda)=0$ as $\mathrm{id}-\mathcal{G}^\dagger$ is the projector onto the kernel of $\mathcal{G}^\dagger$. 
\end{proof}

\section{$\Pi_M(\Lambda)$ as a binary discrimination task\label{app:state-disc}}

In this section, we show how $\Pi_M(\Lambda)$ as defined in Eq.~\eqref{eq:rel-Pi-C} can be identified with an additive advantage of a particular binary state discrimination task, as identified via the respective success probability. For a detailed discussion of state discrimination tasks and an associated operational advantage provided by resource states, see Refs.~\cite{TR19, KTAY24}.

As explained in the first paragraph of Sec.~\ref{subsec:capacity}, the yield $Y(\rho)$ \eqref{eq:yield} is naturally defined with respect to a POVM element $M$, i.e., $0\leq M\leq \mathbbm{1}$, such that $Y(\rho)\in [0, 1]$ corresponds to a success probability. In most theorems in the main text we skipped the constraint $0\leq M\leq \mathbbm{1}$ and considered a general self-adjoint $M\in\mathcal{B}_{\mathrm{sa}}(\mathcal{H})$ whenever the mathematical derivation of the results did not require $0\leq M\leq \mathbbm{1}$. However, the distinction between $M$ being a POVM element and a general self-adjoint operator becomes important when we would like to relate $\Pi_M(\Lambda)$ to the advantage as phrased in terms of state discrimination success probabilities in Ref.~\cite{KTAY24}. 

We consider the two cases separately and discuss $M$ being a POVM element first. While the goal of channel discrimination protocols as described in Ref.~\cite{KTAY24} is to determine which quantum channel $\Lambda_i$ of an ensemble $\{p_i, \Lambda_i\}$ with probabilities $p_i$ had been applied to a quantum state, we work with a single quantum channel $\Lambda$ in the definition of $\Pi_M(\Lambda)$. Thus, the channel discrimination becomes an a priori trivial task. Nevertheless, it is still instructive to relate $\Pi_M(\Lambda)$ to the advantage described in Ref.~\cite{KTAY24}. For this, we identify $p_{\mathrm{succ}}(\rho, \Lambda, M)\equiv Y(\rho) $ such that for $0\leq M\leq \mathbbm{1}$:
\begin{equation}\label{eq:Pi-prob}
\Pi_M(\Lambda) \equiv \sup_{\rho\in \mathcal{D}(\mathcal{H})}p_{\mathrm{succ}}(\rho, \Lambda, M) - \sup_{\sigma\in \mathcal{F}}p_{\mathrm{succ}}(\sigma, \Lambda, M)\,,
\end{equation}
i.e., instead of the ratio between $p_{\mathrm{succ}}(\rho, \Lambda, M) $ for a fixed input state $\rho$ and the optimal success probability for states in the set $\mathcal{F}$ of free states (see Ref.~\cite{KTAY24} for details), $\Pi_M(\Lambda) $ quantifies the maximal possible yield change in agreement with $\mathcal{C}_M(\Lambda)$ using an ``additive'' notion of advantage rather than a ``multiplicative'' one. 

In case $M$ is not a POVM element but a general observable, we first have to modify $\Pi_M(\Lambda)$ accordingly before we can identify its constituents with two success probabilities for an input state $\rho$ and a free state $\sigma$ as in Eq.~\eqref{eq:Pi-prob}. For this, we build a binary POVM $\{\widetilde{M}, \mathbbm{1} - \widetilde{M}\}$ from M as follows: first, denote the maximal and minimal eigenvalues of $M$ by $\lambda_{\max}^{(M)}$ and $\lambda_{\min}^{(M)}$, respectively. Then, we construct a linear function of $M$ according to (assume $M\neq \alpha \mathbbm{1}, \alpha\in \RR$)
\begin{equation}
\widetilde{M}  = \frac{M - \lambda_{\min}^{(M)}\mathbbm{1}}{\lambda_{\max}^{(M)}-\lambda_{\min}^{(M)}}\,,
\end{equation}
which is well-defined for finite-dimensional Hilbert spaces $\mathcal{H}$ as this ensures that $M$ is bounded with $\lambda_{\min}^{(M)}, \lambda_{\max}^{(M)}\in (-\infty, \infty)$. In addition, we excluded the degenerate case $\lambda_{\min}^{(M)}=\lambda_{\max}^{(M)}$, which implies $M\neq \alpha \mathbbm{1}, \alpha\in \RR$ as in that case $\Pi_M(\Lambda)=0$. Note that $\widetilde{M}$ is only one particular way to construct a POVM element from an observable $M$. Next, we abbreviate $\Delta\lambda^{(M)}:= \lambda_{\max}^{(M)}-\lambda_{\min}^{(M)}$. Then, $\Pi_M(\Lambda)$ can be written as
\begin{align}
\Pi_M(\Lambda)  &= \Delta\lambda^{(M)}\Bigg[\sup_{\rho\in \mathcal{D}(\mathcal{H})}\left| p_{\mathrm{succ}}(\rho, \Lambda, \widetilde{M})  +\frac{\lambda_{\min}^{(M)}}{\Delta\lambda^{(M)}}\right|\nonumber\\
&\quad - \sup_{\sigma\in \mathcal{F}} \left| p_{\mathrm{succ}}(\sigma, \Lambda, \widetilde{M}) +\frac{\lambda_{\min}^{(M)}}{\Delta\lambda^{(M)}}\right| \Bigg]\,.
\end{align}
In particular, if $M\geq 0$ then $\Tr[M\Lambda(\rho)]\geq 0$ for all states $\rho$, so the absolute values in the definition of $\Pi_M(\Lambda)$ are redundant. In this case the constant shift $\lambda_{\min}^{(M)}/\Delta\lambda^{(M)}$ cancels between the two suprema, and we obtain
\begin{align}
\Pi_M(\Lambda)
= \Delta\lambda^{(M)}\Big[\sup_{\rho\in\mathcal{D}(\mathcal{H})}p_{\mathrm{succ}}(\rho,\Lambda,\widetilde M)\nonumber\\
- \sup_{\sigma\in\mathcal{F}}p_{\mathrm{succ}}(\sigma,\Lambda,\widetilde M)\Big]\,.
\end{align}
Comparing with Eq.~\eqref{eq:Pi-prob}, where $M$ itself is a POVM element, we see that also for a general $M\geq 0$ the quantity $\Pi_M(\Lambda)$ can be expressed entirely in terms of success probabilities $p_{\mathrm{succ}}(\cdot,\Lambda,\widetilde M)$ of a suitably normalised effect $\widetilde M$, and thus fits into the operational advantage framework of channel discrimination tasks in Ref.~\cite{KTAY24}, up to the overall scaling factor $\Delta\lambda^{(M)}$.

\section{Proof of Corollary \ref{cor:CM-PiM}\label{app:proof-cor-2}}

\begin{proof}
For any density operator $\rho\in\mathcal{D}(\mathcal{H})$, we have
\begin{equation}
\left\vert \Tr[M\Lambda(\rho)]\right\vert \leq \left\vert \Tr[M\Lambda(\mathcal{G}(\rho))]\right\vert + \left\vert \Tr[M\Lambda(\rho - \mathcal{G}(\rho))]\right\vert\,.
\end{equation}
Taking on both sides the supremum over all $\rho\in\mathcal{D}(\mathcal{H})$ and using that $\mathcal{G}(\rho) \in \mathcal{F}$ for all $\rho\in \mathcal{D}(\mathcal{H})$ yields Eq.~\eqref{eq:rel-Pi-C}. 
\end{proof}

\section{Proof of Theorem \ref{thm:capacity-properties} \label{app:proof-thm-3}}

\begin{proof}
We abbreviate $\Delta_{\Lambda_k}: = \Lambda_k(\rho) -\Lambda_k(\mathcal{G}(\rho))$. Then, convexity of $\mathcal{C}_M(\Lambda)$ in $\Lambda$ follows from
\begin{align}
\mathcal{C}_M\left(\sum_{k}p_k \Lambda_k\right)  & = \max_{\rho\in\mathcal{D}(\mathcal{H})}\left\vert\Tr\left[M\sum_k p_k \Delta_{\Lambda_k}(\rho)  \right]\right\vert\nonumber\\
&\leq  \max_{\rho\in\mathcal{D}(\mathcal{H})}\sum_kp_k \left\vert\Tr\left[M \Delta_{\Lambda_k}(\rho)  \right]\right\vert\nonumber\\
&\leq \sum_kp_k\mathcal{C}_M(\Lambda_k)\,.
\end{align}

Let $\Delta\mathcal{C}_{M}^{(\Lambda_1, \Lambda_2)}: = \left\vert \mathcal{C}_M(\Lambda_1) -\mathcal{C}_M(\Lambda_2)\right\vert $. To show Lipshitz continuity, we first observe that
\begin{align}
\Delta\mathcal{C}_{M}^{(\Lambda_1, \Lambda_2)} &\leq \max_{\rho\in\mathcal{D}(\mathcal{H})}\big\vert \left\vert \Tr[M\Delta_{\Lambda_1}] \right\vert -  \left\vert \Tr[M\Delta_{\Lambda_2}(\rho)  ] \right\vert\big\vert\nonumber\\
&\leq \max_{\rho\in\mathcal{D}(\mathcal{H})} \left\vert \Tr\left[M\left(\Delta_{\Lambda_1}(\rho) -\Delta_{\Lambda_2}(\rho)  \right)\right]    \right\vert\nonumber\\
&\leq \max_{\rho\in\mathcal{D}(\mathcal{H})} \big\{\left\vert\Tr\left[M(\Lambda_1-\Lambda_2)(\rho)\right]\right\vert \nonumber\\
&\quad +\left\vert \Tr\left[M(\Lambda_1-\Lambda_2)(\mathcal{G}(\rho))\right]\right\vert\big\}\nonumber\\
&\leq \|M\|_{\infty} \max_{\rho\in\mathcal{D}(\mathcal{H})} \big[ \|(\Lambda_1-\Lambda_2)(\rho)\|_1\nonumber\\
&\quad  + \|(\Lambda_1-\Lambda_2)(\mathcal{G}(\rho))\|_1\big]\,,
\end{align}
where we used the (reverse) triangle inequality and H\"older's inequality for Schatten norms with $p=1, q=\infty$.
The diamond norm $\|\cdot\|_{\diamond}$ on superoperators (and, thus, quantum channels) is defined in Refs.~\cite{Kitaev97, Watrous04, Watrous09, BS10}. Since for any quantum channel $\Lambda$
\begin{equation}
\|\Lambda\|_{\diamond} : = \sup_{k\geq 1}\|\Lambda \otimes\mathrm{id}_k\|_1\,,
\end{equation}
and $\mathcal{G}$ maps density operators to density operators, we can choose $k=1$ to obtain
\begin{align}
\Delta\mathcal{C}_{M}^{(\Lambda_1, \Lambda_2)} &\leq  \|M\|_{\infty} \|\Lambda_1-\Lambda_2\|_{\diamond} \left(\|\rho\|_1+ \|\mathcal{G}(\rho)\|_1\right)\nonumber\\
& = K\|\Lambda_1-\Lambda_2\|_{\diamond}
\end{align}
with Lipschitz constant $K: = 2\|M\|_{\infty}$. In case $\mathcal{G}$ is linear, we can use Eq.~\eqref{eq:capacity} to show that
\begin{equation}
\Delta\mathcal{C}_{M}^{(\Lambda_1, \Lambda_2)}\leq \|\mathrm{id}-\mathcal{G}\|_{\infty\to \infty} \|M\|_{\infty}\|\Lambda_1-\Lambda_2\|_{\diamond}\,,
\end{equation}
where the superoperator norm $\|\mathcal{G}\|_{\infty\to \infty}$ is defined as
\begin{equation}
\|\mathcal{G}\|_{\infty\to \infty} := \sup_{\mathcal{B}(\mathcal{H})\ni M\neq 0}\frac{\|\mathcal{G}(M)\|_{\infty}}{\|M\|_\infty}\,.
\end{equation}
If $\mathcal{G}$ is an idempotent, unital CPTP map, it is a pinching quantum channel \cite{Tomamichel2016, Watrous2018} such that for any $O\in\mathcal{B}(\mathcal{H})$
\begin{equation}
\mathcal{G}(O) = \sum_iP_iOP_i\,,
\end{equation}
where $\{P_i\}_i$ is a set of orthogonal projectors. Then, $\mathcal{G}^\dagger=\mathcal{G}$ because for any two $A,B\in\mathcal{B}(\mathcal{H})$,
\begin{equation}
\Tr[A\mathcal{G}(B)] = \Tr[\mathcal{G}^\dagger(A) B] = \Tr[\mathcal{G}(A)B]\,.
\end{equation}
Moreover, $\|O - \mathcal{G}(O)\|_{\infty}\leq \|O\|_{\infty}$ \cite{BK08} such that $\|\mathrm{id}-\mathcal{G}\|_{\infty\to\infty}\leq 1$. 

Property (iii) immediately follows from
\begin{align}
\mathcal{C}_{a M_1+bM_2}(\Lambda) &\leq \sup_{\rho\in\mathcal{D}(\mathcal{H})}\left( \left\vert a\Tr[M_1\Delta_{\Lambda}(\rho)]\right\vert + \left\vert b\Tr[M_2\Delta_{\Lambda}(\rho)]\right\vert\right)\nonumber\\
&\leq |a|\sup_{\rho\in\mathcal{D}(\mathcal{H})}\left\vert \Tr[M_1\Delta_{\Lambda}(\rho)]\right\vert \nonumber\\
&\quad + |b|\sup_{\rho\in\mathcal{D}(\mathcal{H})}\left\vert \Tr[M_2\Delta_{\Lambda}(\rho)]\right\vert\nonumber\\
& = |a|\mathcal{C}_{M_1}(\Lambda) +|b|\mathcal{C}_{M_2}(\Lambda)\,.
\end{align}
\end{proof}

\section{Proof of Theorem \ref{thm:data-processing} \label{app:proof-thm-4}}

\begin{proof}
Eq.~\eqref{eq:CM-composition} is a direct consequence of the linearity of $\Lambda_2$ and the definition of the adjoint of a quantum channel, 
\begin{align}
\mathcal{C}_M(\Lambda_2\circ \Lambda_1) &= \max_{\rho\in\mathcal{D}(\mathcal{H})}\left\vert \Tr\left[\Lambda_2^\dagger(M)\left(\Lambda_1(\rho) - \Lambda_1(\mathcal{G}(\rho))\right)\right]\right\vert\nonumber\\
&=\mathcal{C}_{\Lambda_2^\dagger(M)}(\Lambda_1)\,.
\end{align}
Let $\Delta_{\Lambda_1}(\rho) := \Lambda_1(\rho) -\Lambda_1(\mathcal{G}(\rho))$. Then, adding $\pm M\Delta_{\Lambda_1}(\rho)$ inside the trace yields
\begin{align}
\mathcal{C}_{\Lambda_2^\dagger(M)}(\Lambda_1)&\leq \mathcal{C}_M(\Lambda_1)  + \max_{\rho\in\mathcal{D}(\mathcal{H})}\left\vert\Tr\left[\Delta M\left(\Delta_{\Lambda_1}(\rho)\right)\right]\right\vert\nonumber\\
& = \mathcal{C}_M(\Lambda_1)  + \mathcal{C}_{\Delta M}(\Lambda_1)\,.
\end{align}
Moreover, it follows from H\"older's inequality for Schatten norms with $p=1, q=\infty$ that
\begin{align}
\mathcal{C}_{\Delta M}(\Lambda_1) &\leq \|\Delta M\|_{\infty} \max_{\rho\in\mathcal{D}(\mathcal{H})}\|\Delta_{\Lambda_1}\rho\|_1\nonumber\\
&\leq \|\Delta M\|_{\infty} \underbrace{\max_{\rho\in\mathcal{D}(\mathcal{H})}\|\rho - \mathcal{G}(\rho)\|_1}_{=:R_{\mathcal{G}}}\,,
\end{align}
where we used in the last line that the trace distance is contractive under CPTP maps \cite{NielsenChuang}. 

We show (ii) next. Adding $\pm \Lambda_2(\mathcal{G}(\Lambda_1(\rho)))$ inside the trace in Eq.~\eqref{eq:capacity-sup} and using the triangle inequality yields
\begin{align}
\mathcal{C}(\Lambda_2\circ\Lambda_1)&\leq \sup_{\rho\in\mathcal{D}(\mathcal{H})}\big\{ \left\vert \Tr[M\Delta_{\Lambda_2}(\Lambda_1(\rho))]\right\vert \nonumber\\
&\quad + \left\vert \Tr[M\Lambda_2(\mathcal{G}(\Lambda_1(\rho)) - \Lambda_1(\mathcal{G}(\rho)))]\right\vert\big\}\nonumber\\
&\leq \sup_{\rho\in \mathrm{im}(\Lambda_1)}\left\vert \Tr[M\Delta_{\Lambda_2}(\rho)]\right\vert\nonumber\\
&\quad  + \sup_{\rho\in\mathcal{D}(\mathcal{H})}\left\vert \Tr[M\Lambda_2(\mathcal{G}(\Lambda_1(\rho)) - \Lambda_1(\mathcal{G}(\rho)))]\right\vert
\end{align}
The second term in the last inequality,
\begin{equation}
Y_{\mathcal{G}, \Lambda_1}:= \sup_{\rho\in\mathcal{D}(\mathcal{H})}\left\vert \Tr[M\Lambda_2(\mathcal{G}(\Lambda_1(\rho)) - \Lambda_1(\mathcal{G}(\rho)))]\right\vert\,,
\end{equation}
vanishes if $\Lambda_1$ is $\mathcal{G}$-covariant. Moreover, we have $\mathrm{im}(\Lambda_1)\subseteq\mathcal{D}(\mathcal{H})$ and it follows that
\begin{equation}
\mathcal{C}(\Lambda_2\circ\Lambda_1) \leq \mathcal{C}_M(\Lambda_2)
\end{equation}
whenever $\mathcal{G}\circ \Lambda_1 = \Lambda_1\circ\mathcal{G}$. 

The classical coarse-graining inequality (iii) follows from (recall $\Delta_{\Lambda}(\rho) :=\Lambda(\rho) - \Lambda(\mathcal{G}(\rho))$)
\begin{align}
\mathcal{C}_{\widetilde{M}_j}(\Lambda) &\leq \sup_{\rho\in\mathcal{D}(\mathcal{H})}\sum_{j}V_{ji}\left\vert \Tr\left[M_i \Delta_{\Lambda}(\rho)\right]\right\vert\nonumber\\
&\leq \sum_{j}V_{ji} \sup_{\rho\in\mathcal{D}(\mathcal{H})}\left\vert \Tr\left[M_i \Delta_{\Lambda}(\rho)\right]\right\vert\nonumber\\
&= \sum_i V_{ji}\mathcal{C}_{M_i}(\Lambda)\,.
\end{align}
\end{proof}

\section{Proof of Corollary \ref{cor:capacity-geometry} \label{app:proof-cor-5}}

\begin{proof}
For any linear real-valued function $f:X\to\RR$ we have
\begin{align}
\sup_{x\in X}|f(x)| &= \max\left\{ \sup_{x\in X}f(x), \sup_{x\in X}-f(x)\right\} \nonumber\\
&= \sup_{x\in X\cup (-X)}f(x)\nonumber\\
&=\sup_{x\in\mathrm{conv}( X\cup (-X))}f(x)\,.
\end{align}
Applied to $\mathcal{C}_M(\Lambda)$ this yields
\begin{align}
\mathcal{C}_M(\Lambda) = \sup_{O\in \mathcal{M}_{\mathcal{G}, \Lambda}}|\langle M, O\rangle_{\mathrm{HS}}| =  \sup_{O\in \widetilde{\mathcal{M}}_{\mathcal{G}, \Lambda}}\langle M, O\rangle_{\mathrm{HS}}\,.
\end{align}
The last equation is nothing else than the definition of a support function $h_K(y)$ of a convex set $K$ \cite{Rockafellar97} and, thus, $\mathcal{C}_M(\Lambda)=h_{\widetilde{\mathcal{M}}_{\mathcal{G}, \Lambda}}(M)$. 
\end{proof}

\section{Proof of Corollary \ref{cor:CM-polar} \label{app:cor-CM-polar}}

\begin{proof}
The statement is a direct consequence of the following well-known theorem (see, e.g., \cite{Rockafellar97})
\begin{thm}\label{thm:support-gauge}
Let $\mathcal{X}\subset\mathcal{B}_{\mathrm{sa}}(\mathcal{H})$ be non-empty and contain the origin. Then, for every $M\in\mathcal{B}_{\mathrm{sa}}(\mathcal{H})$,
\begin{equation}\label{eq:support-as-gauge}
h_{\mathcal{X}}(M)= \inf\left\{t\geq 0\,\middle\vert\, M\in t\,\mathcal{X}^{\circ}\right\}\,.
\end{equation}
In particular, $h_{\mathcal{X}}(M)\leq 1$ holds if and only if $M\in\mathcal{X}^{\circ}$.
\end{thm}
Theorem \ref{thm:support-gauge} follows from the definition of the polar set in Eq.~\eqref{eq:polar-def} as $M\in t\,\mathcal{X}^{\circ}$ is equivalent to $\langle X,O\rangle_{\mathrm{HS}}\leq t$ for all $O\in\mathcal{X}$, and the smallest such $t$ is the supremum defining $h_{\mathcal{X}}(X)$. Applying Theorem \ref{thm:support-gauge} to $\mathcal{X} = \widetilde{M}_{\mathcal{G}, \Lambda}$ in combination with Corollary \ref{cor:capacity-geometry} yields Corollary \ref{cor:CM-polar}. 
\end{proof}

\section{Total error probability for multi-shot hypothesis testing \label{app:CM-operat-int}}

In this appendix, we derive an upper bound on the total error probability
\begin{equation}
P_{\text{err}}^{(n)}(\rho) := \frac{1}{2}\left[\Pr(\text{error}\mid H_0) + \Pr(\text{error}\mid H_1)\right]\,,
\end{equation}
i.e., the average probability of error over the two hypotheses $H_0$ and $H_1$ introduced in Sec.~\ref{sec:op-interp} when they occur with equal prior probability $1/2$.

For a fixed input state $\rho$, consider $n$ independent and identically distributed (i.i.d.) repetitions of the single-shot hypothesis-testing experiment described in Sec.~\ref{sec:op-interp}. 
Let $X_k\in\{0,1\}$ denote the outcome in run $k$ (with $X_k = 1$ corresponding to the outcome associated with $M$), and $x^n := (x_1, \dots, x_n)\in\{0,1\}^n$ the observed string. 
The product distributions of $x^n$ under $H_0$ and $H_1$ are then
\begin{align}
P_0(x^n) &= \prod_{k=1}^n p_0^{x_k}(1-p_0)^{1-x_k}\,,\\
P_1(x^n) &= \prod_{k=1}^n p_1^{x_k}(1-p_1)^{1-x_k}\,,
\end{align}
where $p_i \equiv p_i(\rho)$ for brevity. Any (deterministic) decision rule is a function $f:\{0,1\}^n\to\{0,1\}$, where the output $0$ stands for $H_0$ and $1$ for $H_1$. 
By the Neyman-Pearson Lemma \cite{NP33}, the optimal decision rule in our setting is the likelihood-ratio test, which in the i.i.d.~Bernoulli case depends on $x^n$ only through the empirical mean $\overline{X} := \sum_{k=1}^n X_k/n$ and is therefore equivalent to a threshold test on $\overline{X}$.

Suppose $p_1\geq p_0$ (otherwise swap the labels), and consider the symmetric threshold $\tau := (p_0+p_1)/2$. 
Define a decision rule that decides in favour of $H_1$ if $\overline{X}\ge\tau$ and in favour of $H_0$ otherwise. 
Under $H_0$, we have $\mathbbm{E}_0[\overline{X}] = p_0$, and an error occurs when $\overline{X}\geq\tau$, i.e., when $\overline{X}-p_0 \geq (p_1-p_0)/2$. 
Similarly, under $H_1$, $\mathbbm{E}_1[\overline{X}] = p_1$ and an error occurs when $\overline{X}<\tau$, i.e., when $p_1-\overline{X} \geq (p_1-p_0)/2$. 
Then, Hoeffding's inequality \cite{Hoeffding63} implies that
\begin{align}
P_0\left(\overline{X} - p_0 \geq \frac{p_1-p_0}{2}\right)&\leq \exp\left(-\frac{n}{2}(p_1-p_0)^2\right)\,,\\
P_1\left(p_1 - \overline{X} \geq \frac{p_1-p_0}{2}\right)&\leq \exp\!\left(-\frac{n}{2}(p_1-p_0)^2\right)\,.
\end{align}
Therefore, the total error probability for this explicit decision rule is bounded by
\begin{equation}\label{eq:Perr-app}
P_{\text{err}}^{(n)}(\rho)\leq \exp\left(-\frac{n}{2}(p_1-p_0)^2\right)= \exp\left(-\frac{n}{2}\left(\Delta Y(\rho)\right)^2\right)\,.
\end{equation}
Finally, the decision rule just constructed is not necessarily optimal for all finite $n$, but it provides a valid upper bound on the minimal achievable error probability. 
Hence, the optimal total error probability in discriminating $\rho$ from $\mathcal{G}(\rho)$ using $n$ independent runs of $(\Lambda,M)$ satisfies Eq.~\eqref{eq:Perr-app}.

\section{Proof of Theorem \ref{thm:variation-bound} \label{app:proof-thm-6}}

\begin{proof}
To show the variation bound, we use that for any two functions $f_1, f_2: X\to \RR$,
\begin{equation}\label{eq:sup-identity}
\big|\sup_{x\in X}|f_1(x)| - \sup_{x\in X}|f_2(x)|\big| \leq \sup_{x\in X}|f_1(x)-f_2(x)|\,.
\end{equation}
By assumption, $\{\Lambda_t\}_{t\in\mathcal{I}}$ is absolutely continuous in integral form,
\begin{equation}\label{eq:integral-Lambda}
\Lambda_{t_2}-\Lambda_{t_1} = \int_{t_1}^{t_2} \dot{\Lambda}_s\,ds
\qquad (t_1\leq t_2)\,.
\end{equation}
In particular, $t\mapsto \dot{\Lambda}_t$ is Bochner integrable on every compact subinterval of $\mathcal{I}$.
Next, we define
\begin{equation}
f_t(\rho):= \Tr \left[M\left(\Lambda_t(\rho)   - \Lambda_t(\mathcal{G}(\rho))\right)\right]
\end{equation}
such that
\begin{equation}
\mathcal{C}_M(\Lambda_t) = \sup_{\rho\in \mathcal{D}}(\mathcal{H}) |f_t(\rho)|\,.
\end{equation}

Using Eq.~\eqref{eq:sup-identity}, we obtain
\begin{align}\label{eq:CM-diff}
\big|\mathcal{C}_M(\Lambda_{t_2})  - \mathcal{C}_M(\Lambda_{t_1})\big| \leq \sup_{\rho\in\mathcal{D}(\mathcal{H})} \big|f_{t_2}(\rho) - f_{t_1}(\rho)\big|\,.
\end{align}
Then, for each fixed $\rho$, Eq.~\eqref{eq:integral-Lambda} implies that
\begin{align}
f_{t_2}(\rho)-f_{t_1}(\rho) &= \Tr \left[ M\left(\Lambda_{t_2}(\rho-\mathcal{G}(\rho)) - \Lambda_{t_1}(\rho-\mathcal{G}(\rho))\right)\right] \nonumber\\
&= \int_{t_1}^{t_2}\mathrm{d}s\,\Tr \left[ M\,\dot\Lambda_s\left(\rho - \mathcal{G}(\rho)\right) \right]\,,
\end{align}
which leads to
\begin{align}
\big|f_{t_2}(\rho) - f_{t_1}(\rho)\big| \leq \int_{t_1}^{t_2}\mathrm{d}s\, \left| \Tr \left[M\dot\Lambda_s\left(\rho - \mathcal{G}(\rho)\right)\right]\right|\,.
\label{eq:f-diff-bound}
\end{align}
Since both $\rho$ and $\mathcal{G}(\rho)$ are density operators, one has
\begin{equation}
\|\rho  - \mathcal{G}(\rho)\|_1\leq 2\,.
\end{equation}
Hence, for almost every $s\in\mathcal{I}$,
\begin{align}
\Gamma_M(s) &= \sup_{\rho\in\mathcal{D}(\mathcal{H})} \left| \Tr \left[ M \dot\Lambda_s\left(\rho - \mathcal{G}(\rho)\right) \right] \right| \nonumber\\
&\leq\sup_{\rho\in\mathcal{D}(\mathcal{H})} \|M\|_\infty\,\|\dot\Lambda_s\|_{1\to 1}\,\|\rho - \mathcal{G}(\rho)\|_1
\nonumber\\
&\leq 2\|M\|_\infty\,\|\dot\Lambda_s\|_{1\to 1}\,.
\end{align}
Therefore $t\mapsto \Gamma_M(t)$ is locally integrable on $\mathcal{I}$.

Taking the supremum over $\rho$ in \eqref{eq:f-diff-bound} and using monotonicity of the integral yields
\begin{align*}
\sup_{\rho\in\mathcal{D}(\mathcal{H})} \big|f_{t_2}(\rho)-f_{t_1}(\rho)\big| \leq \int_{t_1}^{t_2}\mathrm{d} s\,\Gamma_M(s)\,.
\end{align*}
Combining this with \eqref{eq:CM-diff} gives the variation bound. 

Next, we show Eq.~\eqref{eq:Dini}.
The right Dini derivatives for a function $f:\RR\to \RR$ are defined as \cite{KannanKrueger1996Dini}
\begin{align}
D^+f(t) &:= \limsup_{h\to 0^+}\frac{f(t+h)-f(t)}{h}\,, \nonumber\\
D_+f(t) &:= \liminf_{h\to 0^+}\frac{f(t+h)-f(t)}{h}\,.
\end{align}
Here, $f(t) = \mathcal{C}_M(\Lambda_t)$.
From \eqref{eq:variation-bound} we have, for every $t$ and $h>0$,
\begin{align}
\left|\mathcal{C}_M(\Lambda_{t+h})-\mathcal{C}_M(\Lambda_t)\right|&\leq \int_t^{t+h}\mathrm{d}s\,\Gamma_M(s)\,,
\end{align}
and hence
\begin{align}
\frac{\mathcal{C}_M(\Lambda_{t+h}) - \mathcal{C}_M(\Lambda_t)}{h}
&\leq \frac{1}{h}\int_t^{t+h}\mathrm{d}s\,\Gamma_M(s)\,,
\label{eq:upper-quotient} \\
\frac{\mathcal{C}_M(\Lambda_{t+h}) - \mathcal{C}_M(\Lambda_t)}{h} &\geq -\frac{1}{h}\int_t^{t+h}\mathrm{d}s\,\Gamma_M(s)\,.
\label{eq:lower-quotient}
\end{align}
Since $t\mapsto \Gamma_M(t)$ is locally integrable, the Lebesgue differentiation theorem implies that for Lebesgue almost every $t\in\mathcal{I}$,
\begin{equation}
\lim_{h\to 0^+}\frac{1}{h}\int_t^{t+h}\mathrm ds\,\Gamma_M(s)=\Gamma_M(t)\,.
\end{equation}
Taking the limit superior of \eqref{eq:upper-quotient} as $h\to 0^+$ then gives
\begin{equation}
D^+\mathcal{C}_M(\Lambda_t) \leq  \limsup_{h\to 0^+}\frac{1}{h}\int_t^{t+h}\mathrm{d}s\,\Gamma_M(s)  =  \Gamma_M(t)
\end{equation}
for almost every $t$.
Similarly, from Eq.~\eqref{eq:lower-quotient} and the identity
$\liminf_{h\to 0^+}(-g(h))=-\limsup_{h\to 0^+}g(h)$, we obtain
\begin{equation}
D_+\mathcal{C}_M(\Lambda_t) \geq -\limsup_{h\to 0^+}\frac{1}{h}\int_t^{t+h}\mathrm{d}s\,\Gamma_M(s)= -\Gamma_M(t)
\end{equation}
for almost every $t$.
The analogous bounds for the left Dini derivatives follow by applying the same argument on intervals $[t-h,t]$ instead of $[t,t+h]$.
\end{proof}

\section{Proof of Lemma \ref{lem:Gamma-M-zero} \label{app:proof-lem-7}}
\begin{proof}
Suppose $\dot{\Lambda}_t^\dagger (M)\in \mathcal{V}_{\mathcal{G}}^\perp$. Then, by definition, we have 
\begin{equation}\label{eq:proof-lem7}
\Tr\left[\dot{\Lambda}_t^\dagger(M) X\right]=0 \quad \forall\,X \in \mathcal{V}_{\mathcal{G}}\,,
\end{equation}
so the above equality holds in particular for any $\rho - \mathcal{G}(\rho)$. This implies $\Gamma_M(t)=0$. To show the reverse direction, we start with $\Gamma_M(t)=0$. Then, 
\begin{equation}
\sup_{\rho\in\mathcal{D}(\mathcal{H})} \left| \Tr\left[\dot{\Lambda}_t^\dagger(M) (\rho - \mathcal{G}(\rho))\right] \right|=0\,.
\end{equation}
Since every term in the supremum is positive and the trace is linear, it follows that Eq.~\eqref{eq:proof-lem7} holds for any $X \in \mathcal{V}_{\mathcal{G}}$. Thus, $\dot{\Lambda}_t^\dagger\in \mathcal{V}_{\mathcal{G}}^\perp$. 
\end{proof}

\section{Proof of Corollary \ref{cor:time-bound} \label{app:proof-cor-8}}

\begin{proof}
We first bound the instantaneous advantage rate $\Gamma_M(s)$. We use H\"older's inequality for Schatten norms with $p=1$ and $q=\infty$ to show that the instantaneous rate $\Gamma_M(t)$ as defined in Eq.~\eqref{eq:Gamma-M-sup-general} is bounded from above by
\begin{align}
\Gamma_M(s) &= \sup_{\rho\in\mathcal{D}(\mathcal{H})} \left|\Tr \left[ M \dot\Lambda_s  \left(\rho-\mathcal{G}(\rho)\right)\right]\right|\nonumber\\
&\leq \|M\|_\infty \|\dot\Lambda_s\|_{1\to 1} \sup_{\rho\in\mathcal{D}(\mathcal{H})} \|\rho-\mathcal{G}(\rho)\|_1 \nonumber\\
&=\|M\|_\infty R_{\mathcal{G}} \|\dot\Lambda_s\|_{1\to 1}  \,,
\end{align}
where we used in the second line the definition of the induced operator norm, and $R_{\mathcal{G}}$ is the resource radius defined in Eq.~\eqref{eq:res-radius}. Thus, it follows that 
\begin{equation}\label{eq:GammaM-dotLambda-bound}
\Gamma_M(s)\le c_{M,\mathcal{G}}\|\dot\Lambda_s\|_{1\to 1} \quad\text{for a.e. }s\in\mathcal{I}\,,
\end{equation}
with $c_{M,\mathcal{G}} = \|M\|_\infty R_{\mathcal{G}}\leq 2\|M\|_\infty$.

Inserting Eq.~\eqref{eq:GammaM-dotLambda-bound} into the variation bound from Theorem \ref{thm:variation-bound} yields
\begin{align}
\left|\Delta\mathcal{C}_M(\Lambda_{t_2},\Lambda_{t_1})\right| &\leq c_{M,\mathcal{G}}\int_{t_1}^{t_2}\mathrm ds\,\|\dot\Lambda_s\|_{1\to 1} \nonumber\\
&\leq c_{M,  \mathcal{G}} \Delta \tau   \operatorname*{ess\,sup}_{s\in[t_1,t_2]}\|\dot\Lambda_s\|_{1\to 1}\,,
\label{eq:uniform-Lip}
\end{align}
where $\Delta\tau:=t_2-t_1$. This proves Eq.~\eqref{eq:uniform-bound}.

If $\|\dot\Lambda_s\|_{1\to 1}\leq K_{\max}$ for almost every $s\in\mathcal{I}$, then Eq.~\eqref{eq:uniform-Lip} implies
\begin{equation}
\left|\Delta\mathcal{C}_M(\Lambda_{t_2},\Lambda_{t_1})\right| \leq c_{M,\mathcal{G}}\,K_{\max}\,\Delta\tau\,.
\end{equation}
Hence, for any non-zero change $|\Delta\mathcal{C}_M|$ and $K_{\max}>0$,
\begin{equation}
\Delta\tau \geq\frac{\left|\Delta\mathcal{C}_M(\Lambda_{t_2},\Lambda_{t_1})\right|}{c_{M,\mathcal{G}}K_{\max}},
\end{equation}
which is the time bound \eqref{eq:time-QSL}.

Finally, for an external time budget $T_{\max}$ with $0\le \Delta\tau\le T_{\max}$ and $\|\dot\Lambda_s\|_{1\to 1}\leq K_{\max}$ for almost every $s\in[t_1,t_1+T_{\max}]$, Eq.~\eqref{eq:uniform-Lip} implies that
\begin{align}
\left|\Delta\mathcal{C}_M(\Lambda_{t_2},\Lambda_{t_1})\right| &\leq c_{M,\mathcal{G}} \Delta\tau K_{\max} \nonumber\\
&\leq T_{\max}\,c_{M,\mathcal{G}} K_{\max} \,,
\end{align}
which is the feasibility bound \eqref{eq:feasibility-bound}.
If, in addition, $\Lambda_t$ satisfies a time-local master equation
$\dot\Lambda_t=\mathcal{L}_t\circ\Lambda_t$ almost everywhere, then $\|\dot\Lambda_t\|_{1\to 1}\leq  \|\mathcal{L}_t\|_{1\to 1}$.  Hence the generator-based version of the bounds follows as a special case.
\end{proof}

\section{Proof of Theorem \ref{thm:dissect-channel} \label{app:proof-thm-9}}

\begin{proof}
Since $\mathcal{G}$ is a linear and idempotent, $\mathcal{B}(\mathcal{H})$ decomposes into a direct sum
\begin{equation}\label{eq:B-dec}
\mathcal{B}(\mathcal{H}) = \text{im}(\mathcal{G}) \oplus \mathrm{ker}(\mathcal{G})\,,
\end{equation}
where $\text{im}(\mathcal{G})$ denotes the image of $\mathcal{G}$ and $\mathrm{ker}(\mathcal{G})$ its kernel. 
As a consequence, we can uniquely decompose any $X\in \mathcal{B}(\mathcal{H})$ as
\begin{equation}
X = \underbrace{\mathcal{G}(X)}_{\in  \text{im}(\mathcal{G})} +\underbrace{\mathcal{G}^\perp (X)}_{\in\mathrm{ker}(\mathcal{G})}\,,
\end{equation}
where $\mathcal{G}^\perp:=\text{id} - \mathcal{G}$ with $(\mathcal{G}^\perp)^2 =\mathcal{G}^\perp$ and $\mathrm{im}(\mathcal{G}^\perp) = \ker(\mathcal{G})$, such that $\mathcal{G}^\perp$ is the projection onto the kernel of $\mathcal{G}$. 

Moreover, every linear map $T: \mathcal{B}(\mathcal{H})\to\mathcal{B}(\mathcal{H})$ can be written as 
\begin{equation}\label{eq:T-decomp}
T = \mathcal{G}\circ T\circ\mathcal{G} +  \mathcal{G}^\perp\circ T\circ\mathcal{G}^\perp+ \mathcal{G}^\perp\circ T\circ\mathcal{G}+ \mathcal{G}\circ T\circ\mathcal{G}^\perp\,.
\end{equation}
Applied to a quantum channel $\Lambda_t$, this motivates the definitions
\begin{equation}
\Lambda_{t,\mathrm{free}} := \mathcal{G}\circ\Lambda_t\circ\mathcal{G}\,,\quad
\Lambda_{t,\mathrm{res}} := \Lambda_t - \Lambda_{t,\mathrm{free}}\,,
\end{equation}
and, similarly in the time-local case, for the generator,
\begin{equation}
\mathcal{L}_{t,\mathrm{free}} := \mathcal{G}\circ\mathcal{L}_t\circ\mathcal{G}\,,\quad 
\mathcal{L}_{t,\mathrm{res}} := \mathcal{L}_t - \mathcal{L}_{t,\mathrm{free}}\,.
\end{equation}
By construction, we have
\begin{equation}
\Lambda_t = \Lambda_{t,\mathrm{free}} + \Lambda_{t,\mathrm{res}}\,,\quad 
\mathcal{L}_t = \mathcal{L}_{t,\mathrm{free}} + \mathcal{L}_{t,\mathrm{res}}\,,
\end{equation}
which proves the decompositions \eqref{eq:Lambda-t-dec} and \eqref{eq:Lt-dec}. Using $\mathcal{G}^2=\mathcal{G}$ and $(\mathcal{G}^\perp)^2=\mathcal{G}^\perp$, we obtain
\begin{equation}
\mathcal{G}\circ\Lambda_{t,\mathrm{free}}
= \mathcal{G}\circ\mathcal{G}\circ\Lambda_t\circ\mathcal{G}
= \mathcal{G}\circ\Lambda_t\circ\mathcal{G}
= \Lambda_{t,\mathrm{free}}\,,
\end{equation}
and
\begin{equation}
\Lambda_{t,\mathrm{free}}\circ\mathcal{G} = \mathcal{G}\circ\Lambda_t\circ\mathcal{G}\circ\mathcal{G} = \mathcal{G}\circ\Lambda_t\circ\mathcal{G} = \Lambda_{t,\mathrm{free}}\,,
\end{equation}
so $\Lambda_{t,\mathrm{free}}$ satisfies \eqref{eq:decomp-properties-1}. Moreover,
\begin{equation}
\mathcal{G}\circ\Lambda_{t,\mathrm{res}}\circ\mathcal{G} = \mathcal{G}\circ\Lambda_t\circ\mathcal{G} - \mathcal{G}\circ(\mathcal{G}\circ\Lambda_t\circ\mathcal{G})\circ\mathcal{G} = 0\,.
\end{equation}
The same relations hold for $\mathcal{L}_{t,\mathrm{free}}$ and $\mathcal{L}_{t,\mathrm{res}}$ by the same argument. From $\mathcal{G}\circ\Lambda_{t,\mathrm{free}} = \Lambda_{t,\mathrm{free}} = \Lambda_{t,\mathrm{free}}\circ\mathcal{G}$ (and analogously for $\mathcal{L}_{t,\mathrm{free}}$) it also follows that
$[\mathcal{G},\Lambda_{t,\mathrm{free}}]=[\mathcal{G},\mathcal{L}_{t,\mathrm{free}}]=0$.

If $\Lambda_t$ is CPTP and $\mathcal{G}$ is CPTP, then $\Lambda_{t,\mathrm{free}} = \mathcal{G}\circ\Lambda_t\circ\mathcal{G}$ is a composition of CPTP maps, hence CPTP as claimed in part (iii).

It remains to show (ii). For any $\rho\in\mathcal{D}(\mathcal{H})$,
\begin{equation}
\Lambda_{t,\mathrm{free}}(\rho-\mathcal{G}(\rho))
= \mathcal{G}\left(\Lambda_t(\mathcal{G}(\rho)-\mathcal{G}^2(\rho))\right)
= \mathcal{G}\left(\Lambda_t(0)\right) = 0\,,
\end{equation}
such that
\begin{align}
\Lambda_t(\rho-\mathcal{G}(\rho))&= \Lambda_{t,\mathrm{free}}(\rho-\mathcal{G}(\rho)) + \Lambda_{t,\mathrm{res}}(\rho-\mathcal{G}(\rho))\nonumber\\
&= \Lambda_{t,\mathrm{res}}(\rho-\mathcal{G}(\rho))\,.
\end{align}
This implies that 
\begin{equation}
\Tr\left[M \Lambda_t(\rho-\mathcal{G}(\rho))\right] = \Tr\left[M\Lambda_{t,\mathrm{res}}(\rho-\mathcal{G}(\rho))\right]\,\,\forall\,\rho\in\mathcal{D}(\mathcal{H})\,,
\end{equation}
and taking the supremum over $\rho$ yields
\begin{equation}
\mathcal{C}_M(\Lambda_t) = \mathcal{C}_M(\Lambda_{t,\mathrm{res}})\,.
\end{equation}
\end{proof}

\section{Proof of Corollary \ref{cor:gen-free-compatibility} \label{app:proof-cor-10}}

\begin{proof}
The equivalence follows from the chain rule
\begin{equation}
\frac{d}{dt}\Lambda_{t, \text{free}} = \mathcal{L}_{t, \mathrm{free}}\circ \Lambda_{t, \mathrm{free}} + (\mathcal{G}\circ \mathcal{L}_t\circ \mathcal{G}^\perp)\circ(\mathcal{G}^\perp\circ\Lambda_t\circ \mathcal{G})
\end{equation}
with the initial condition $\Lambda_{0, \mathrm{free}} = \mathcal{G}$ originating from the definition of $\Lambda_{t, \mathrm{free}}$ in Eq.~\eqref{eq:Lambda-t-dec}.
\end{proof}

\section{Derivation of $\mathcal{C}_M(\Lambda)$ and $\Gamma_M(t)$ for the donor-acceptor model\label{app:DAM}}

In this section we provide the detailed derivation of the resource impact functional $\mathcal{C}_M(\Lambda)$, the instantaneous rate $\Gamma_M(t)$, and the associated variation and time bounds for the donor--acceptor model.

We recall the three-step process $\Lambda = \Lambda_{\mathrm{AD}}\circ\Lambda_{\mathrm{PD}}\circ\Lambda_\theta$ from Eq.~\eqref{eq:Lambda-conc}, where $\Lambda_\theta$ is the unitary channel generated by the excitonic Hamiltonian in Eq.~\eqref{eq:Hex} and describes coherent mixing between the singly excited donor and acceptor states, $\Lambda_{\mathrm{PD}}$ is a phase-damping channel acting on the excited-state manifold, and $\Lambda_{\mathrm{AD}}$ describes spontaneous emission from $\ket{D}$ and $\ket{A}$ to the ground state $\ket{g}$ via amplitude damping. Each map is CPTP, and so is their concatenation.

Let $K_i^{\mathrm{PD}}$ and $K_k^{\mathrm{AD}}$ denote Kraus operators of $\Lambda_{\mathrm{PD}}$ and $\Lambda_{\mathrm{AD}}$, with $i\in\{0,1\}$ and $k\in\{0,1,2\}$. A Kraus representation of $\Lambda_{\mathrm{AD}}\circ\Lambda_{\mathrm{PD}}$ is then given by $K_{ik} := K_i^{\mathrm{AD}}K_k^{\mathrm{PD}}$, and one convenient choice (in the basis $\{\ket{g}, \ket{D}, \ket{A}\}$) is given by
\begin{align}
K_{00} &= \ket{g}\!\bra{g} + \sqrt{\frac{(1+\eta)(1-p_D)}{2}}\ket{D}\!\bra{D}\\
&\quad+  \sqrt{\frac{(1+\eta)(1-p_A)}{2}}\ket{A}\!\bra{A}  \,,\nonumber\\
K_{10} &=  \sqrt{\frac{(1-\eta)(1-p_D)}{2}}\ket{D}\!\bra{D}\nonumber\\
&\quad  - \sqrt{\frac{(1-\eta)(1-p_A)}{2}}\ket{A}\!\bra{A}    \,,\nonumber\\
K_{01} &=  \sqrt{\frac{p_D(1+\eta)}{2}} \ket{g}\!\bra{D}\,,\,\,\,\,\,\,
K_{02} =  \sqrt{\frac{p_A(1+\eta)}{2}} \ket{g}\!\bra{A} \,,\nonumber\\
K_{11} &=   \sqrt{\frac{p_D(1-\eta)}{2}} \ket{g}\!\bra{D}\,,\,\,\,\,\,
K_{12} =  - \sqrt{\frac{p_A(1-\eta)}{2}} \ket{g}\!\bra{A}\,.\nonumber
\end{align}
Here $\eta$ is the phase damping parameter, and $p_D$, $p_A$ are the spontaneous emission probabilities from $\ket{D}$ and $\ket{A}$, respectively. One readily checks that $\sum_j K_j^\dagger K_j = \mathbbm{1}$. A Kraus representation for the full channel $\Lambda$ is obtained by composing these with the unitary $U_\theta$ in Eq.~\eqref{eq:Utheta}, that is, $\{K_j U_\theta\}_j$.

Since $\mathcal{G}$ in Eq.~\eqref{eq:G-DAM} is linear and Hermitian, we work in the following in the Heisenberg picture with $\mathcal{C}_M(\Lambda)$ given by Eq.~\eqref{eq:capacity}. The adjoint of the quantum channel $\Lambda$ acts on $M$ as 
\begin{equation}
\Lambda^\dagger (M) = \Lambda_{\theta}^\dagger \circ \Lambda_{\text{PD}}^\dagger \circ \Lambda_{\text{AD}}^\dagger (M)\,,
\end{equation}
and for a POVM element $M$ of the form in Eq.~\eqref{eq:M-DAM} this yields the expression for $\mathcal{C}_M(\Lambda)$ shown in Eq.~\eqref{eq:CM-Ltheta-main}.

To compute the advantage rate $\Gamma_M(t)$ we propagate $\Lambda_t^\dagger(M)$ under the adjoint Lindblad generator $\mathcal{L}^\dagger$ associated with $\Lambda$. In the time-independent case, $\Lambda_t^\dagger = e^{t\mathcal{L}^\dagger}$ satisfies the Heisenberg equation
\begin{align}\label{eq:time-local-MA-DA}
\frac{\mathrm{d}}{\mathrm{d}t}\Lambda_t^\dagger(M) &= \mathcal{L}^\dagger(\Lambda_t^\dagger(M)) \\
&=i\left[H_{\text{ex}}, \Lambda_t^\dagger(M) \right]+\sum_{k=1}^4\Big( L_k^\dagger\Lambda_t^\dagger(M) L_k \nonumber\\
&\quad - \frac{1}{2}\Big\{L_k^\dagger L_k, \Lambda_t^\dagger(M) \Big\} \Big)\,,
\end{align}
where the jump operators are
\begin{align}\label{eq:Lk-DA-app}
L_1 &= \sqrt{\gamma_{\varphi}}\ket{D}\!\bra{D}\,,\quad L_2 = \sqrt{\gamma_{\varphi}}\ket{A}\!\bra{A}\,, \nonumber\\
L_3 &= \sqrt{\gamma_D}\ket{g}\!\bra{D}\,,\quad L_4 = \sqrt{\gamma_A}\ket{g}\!\bra{A}\,.
\end{align}
The dephasing rate $\gamma_\varphi$ and the relaxation rates $\gamma_D,\gamma_A$ are related to the finite-time parameters by $\eta(\Delta t) = e^{-\gamma_\varphi\Delta t}$ and $p_i(\Delta t) = 1 - e^{-\gamma_i\Delta t}$, $i\in\{D,A\}$.

For an observable $M$ of the form in Eq.~\eqref{eq:M-DAM}, the dynamics \eqref{eq:time-local-MA-DA} does not generate coherences between $\ket{g}$ and $\{\ket{D}, \ket{A}\}$, so $\Lambda_t^\dagger(M)$ preserves the block-diagonal form of $M$ according to
\begin{align}
\Lambda_t^\dagger(M) &= x_g(t)\ket{g}\!\bra{g} + x_D(t)\ket{D}\!\bra{D} + x_A(t)\ket{A}\!\bra{A}\nonumber\\
&\quad + (y(t)\ket{D}\!\bra{A}+ \text{h.c.}),,
\end{align}
with $x_g(t),x_D(t),x_A(t)\in\RR$ and $y(t)\in\CC$. Inserting this ansatz into Eq.~\eqref{eq:time-local-MA-DA} yields the following set of coupled linear differential equations
\begin{align}\label{eq:DA-CDE}
\frac{\mathrm{d}}{\mathrm{d}t}x_D(t)& = 2J\Im[y(t)] + \gamma_D\left[x_g(t)-x_D(t)\right]\nonumber   \,,\\
\frac{\mathrm{d}}{\mathrm{d}t}x_A(t)& = -2J\Im[y(t)] + \gamma_A\left[x_g(t)-x_A(t)\right] \nonumber  \,,\\
\frac{\mathrm{d}}{\mathrm{d}t}y(t)& = i\left[\Delta y(t) + J(x_A(t) - x_D(t))\right]-\xi y(t)\,,
\end{align}
where $\xi := \gamma_\varphi + (\gamma_D+\gamma_A)/2$. The initial conditions are $x_g(0)=\mu_g$, $x_D(0)=\mu_D$, $x_A(0)=\mu_A$, and $y(0)=\nu$, where $\mu_g,\mu_D,\mu_A,\nu$ are the matrix elements of $M$.

Writing $y(t) = u(t)+iv(t)$ and collecting the variables into
\begin{equation}
\bd z(t) := (u(t),v(t),x_D(t),x_A(t))^{\mathrm T}
\end{equation}
the system in \eqref{eq:DA-CDE} takes the form
\begin{equation}
\frac{\mathrm{d}}{\mathrm{d}t} \bd z(t)= A\bd z(t) + \bd b
\end{equation}
with
\begin{align}\label{eq:parameters-DAM}
\bd b &:= (0,0,\gamma_D\mu_g,\gamma_A\mu_g)^{\mathrm T},\\
A &:= \begin{pmatrix}
-\xi & -\Delta & 0 & 0 \\
\Delta & -\xi & -J & J \\
0 & 2J & -\gamma_D & 0 \\
0 & -2J & 0 & -\gamma_A
\end{pmatrix}\,.
\end{align}
The inhomogeneous term can be removed by introducing shifted variables
\begin{equation}
\widetilde{x}_D := x_D- \mu_g\,,\quad
\widetilde{x}_A := x_A- \mu_g\,,
\end{equation}
which leaves $A$ unchanged. In terms of $\widetilde{\bd z} := (u,v,\widetilde{x}_D,\widetilde{x}_A)^{\mathrm T}$, this leads to
$\mathrm{d}\bd z(t)/\mathrm{d}t= A\widetilde{\bd z}(t)$ with initial condition
\begin{equation}\label{eq:z0-init}
\widetilde{\bd z}(0) = \left(\Re(\nu),\Im(\nu),\mu_D-\mu_g,\mu_A-\mu_g\right)^{\mathrm T}
\end{equation}
and, hence, $\widetilde{\bd z}(t) = e^{At}\widetilde{\bd z}(0)$. 
Evaluating $\widetilde{\bd z}(t)$ analytically for general parameters $(\gamma_\varphi,\gamma_D,\gamma_A,\Delta,J)$ is cumbersome. In the following, we therefore specialise to the cases (i) $\gamma_\varphi=0$ with $\gamma_D=\gamma_A$, and (ii) $\Delta=0$ with $\gamma_D=\gamma_A$, for which closed-form expressions and explicit time bounds can be obtained.

\subsection{Zero dephasing and symmetric lifetimes\label{sec:zero-deph-DA}}

For $\gamma_\varphi = 0$ and symmetric lifetimes $\gamma_D=\gamma_A=\gamma$ we have $\xi=\gamma$, and the eigenvalues of the matrix $A$ in Eq.~\eqref{eq:parameters-DAM} are
\begin{equation}
\lambda_1= \lambda_2 = -\gamma,\quad \lambda_{3,4} = -\gamma \pm i\Omega\,,
\end{equation}
where $\Omega = \sqrt{4J^2+\Delta^2}$ is the generalised Rabi frequency. It is convenient to introduce the population difference
\begin{equation}
s(t):= x_A(t)-x_D(t)
\end{equation}
and the total excited-state population
\begin{equation}
N(t) := x_A(t)+x_D(t)\,.
\end{equation}
With the initial conditions in Eq.~\eqref{eq:z0-init}, one obtains
\begin{align}\label{eq:solution-DA-nodeph}
u(t) &= \mathrm{e}^{-\gamma t}\bigg[\Re(\nu)-\frac{\Delta\Im(\nu)}{\Omega}\sin(\Omega t)\nonumber\\
&\quad  +\frac{\Delta[\Delta\Re(\nu)+J(\mu_A-\mu_D)]}{\Omega^2}(\cos(\Omega t)-1)\bigg]\,,\nonumber\\
v(t) &= \mathrm{e}^{-\gamma t}\bigg[\Im(\nu)\cos(\Omega t)\nonumber\\
&\quad +\frac{\Delta\Re(\nu)+J(\mu_A-\mu_D)}{\Omega}\sin(\Omega t)\bigg]\,,\nonumber\\
s(t) &=  \mathrm{e}^{-\gamma t}\bigg[ \mu_A-\mu_D -\frac{4J\Im(\nu)}{\Omega}\sin(\Omega t)\nonumber\\
&\quad  + \frac{4 J[\Delta\Re(\nu)+J(\mu_A-\mu_D)]}{\Omega^2}(\cos(\Omega t)-1)\bigg]\,,\nonumber\\
N(t) &= 2\mu_g +(\mu_A+\mu_D-2\mu_g)\mathrm{e}^{-\gamma t}\,.
\end{align}
For $\mathcal{G}$ the dephasing map in the site basis $\{|g\rangle,|A\rangle,|D\rangle\}$, the task-response functional and the instantaneous rate take the simple form
\begin{equation}\label{eq:CM-GM-DM}
\mathcal{C}_M(\Lambda_t) = |y(t)|\,,\quad \Gamma_M(t) = \left|\frac{\mathrm{d}y(t)}{\mathrm{d}t}\right|\,,
\end{equation}
since $(\mathrm{id}-\mathcal{G})\Lambda_t^\dagger(M)$ is supported only on the $\ket{D}\!\bra{A}$/$\ket{A}\!\bra{D}$ block with eigenvalues $\pm|y(t)|$. Thus $\mathcal{C}_M(\Lambda_t)$ follows directly from $u(t),v(t)$ in Eq.~\eqref{eq:solution-DA-nodeph}, and $\Gamma_M(t)$ is obtained by differentiating $y(t)=u(t)+iv(t)$.

For the measurement operator $M=|A\rangle\!\langle A|$, this yields the instantaneous rate
\begin{widetext}
\begin{equation}\label{eq:GM-zerodeph}
\Gamma_{|A\rangle\!\langle A|}(t) = e^{-\gamma t}|J|\sqrt{\Delta^2\left(\frac{\gamma\left(\cos(\Omega t)-1\right)}{\Omega^2}+ \frac{\sin(\Omega t)}{\Omega}\right)^2+ \left(\cos(\Omega t) - \frac{\gamma}{\Omega}\sin(\Omega t)\right)^2}\,.
\end{equation}
\end{widetext}
We now apply Theorem~\ref{thm:variation-bound} to bound the change
\begin{align}
|\Delta\mathcal{C}_{|A\rangle\!\langle A|}|&\equiv \left|\mathcal{C}_{|A\rangle\!\langle A|}(\Lambda_{t_2}) - \mathcal{C}_{|A\rangle\!\langle A|}(\Lambda_{t_1})\right| \nonumber\\
&\leq \int_{t_1}^{t_2}\mathrm{d}s\,\Gamma_{|A\rangle\!\langle A|}(s)\,,
\end{align}
where the implicit time-dependence of $|\Delta\mathcal{C}_{|A\rangle\!\langle A|}|$ is omitted to simplify the notation. 
For generic $\Delta$, the integral in Eq.~\eqref{eq:GM-zerodeph} is best evaluated numerically and we illustrate both the actual change $|\Delta \mathcal{C}_{|A\rangle\!\langle A|}|$ evaluated from Eq.~\eqref{eq:CM-GM-DM} and the variation bound $\int_0^t\mathrm{d}s\,\Gamma_{|A\rangle\!\langle A|}(t)$ for $t_1=0$ and $t_2\geq t_1$ and a chosen set of parameters in Fig.~\ref{fig:DAM-bounds} in the main text.

We continue with a discussion of the uniform time and feasibility bounds from Corollary \ref{cor:time-bound}. 
For this, we recall that Corollary \ref{cor:time-bound} implies that for any $t_1\leq t_2$
\begin{equation}
\left|\mathcal{C}_M(\Lambda_{t_2}) - \mathcal{C}_M(\Lambda_{t_1})\right|\leq \Delta\tau c_{M,\mathcal{G}}\sup_{s\in[t_1,t_2]}\|\mathcal{L}_s\|_{1\to1}\,,
\end{equation}
where $\Delta\tau := t_2-t_1$, with $c_{M,\mathcal{G}} = \|M\|_\infty R_{\mathcal{G}}$ and $R_{\mathcal{G}}$ the resource radius defined in Eq.~\eqref{eq:res-radius}. In the donor-acceptor model considered in this section, the generator $\mathcal{L}_s\equiv\mathcal{L}$ is time-independent, such that 
\begin{equation}
L_{\max} := \sup_{s\in[t_1,t_2]}\|\mathcal{L}_s\|_{1\to1} = \|\mathcal{L}\|_{1\to1}\,,
\end{equation}
and the bound becomes
\begin{equation}\label{eq:DA-uniform-Lmax-generic}
\left|\mathcal{C}_M(\Lambda_{t_2}) - \mathcal{C}_M(\Lambda_{t_1})\right|\leq \Delta\tau c_{M,\mathcal{G}}\|\mathcal{L}\|_{1\to1}\,.
\end{equation}

While $\|\mathcal{L}\|_{1\to1}$ can in principle be computed exactly for fixed $(\Delta, J, \gamma)$ via semidefinite programming, there is no simple closed-form expression in terms of these parameters.
An analytic estimate of the induced trace norm $\|\mathcal{L}\|_{1\to1}$ is obtained by
splitting $\mathcal{L}$ into its Hamiltonian and dissipative parts (recall $L_1=L_2=0$ for $\gamma_{\varphi}=0$),
\begin{equation}\label{eq:L-DAM-app}
\mathcal{L}(\rho) = -i[H_{\mathrm{ex}},\rho] + \underbrace{\sum_{k=3,4}\left(L_k \rho L_k^\dagger - \frac{1}{2}\{L_k^\dagger L_k,\rho\}\right)}_{=:\mathcal{D}(\rho)}\,,
\end{equation}
with $H_{\mathrm{ex}}$ and $L_k$ given by Eqs.~\eqref{eq:Hex} and \eqref{eq:Lk-DA-app}.
Then, for any operator $X\in\mathcal{B}_{\mathrm{sa}}(\mathcal{H})$ we have
\begin{equation}
\|-i[H_{\mathrm{ex}},X]\|_1 \leq 2\|H_{\mathrm{ex}}\|_\infty \|X\|_1\,,
\end{equation}
where we used H\"older's inequality for Schatten $p$-norms, such that $\| -i[H_{\mathrm{ex}},\cdot]\|_{1\to1}\leq 2\|H_{\mathrm{ex}}\|_\infty$. Since the eigenvalues of $H_{\mathrm{ex}}$ on the $\{\ket{D},\ket{A}\}$ subspace are $\pm\Omega/2$ with $\Omega=\sqrt{4J^2+\Delta^2}$, we obtain
$\|H_{\mathrm{ex}}\|_\infty=\Omega/2$ and hence $\| -i[H_{\mathrm{ex}},\cdot]\|_{1\to1}\leq\Omega$.

To bound the second term in Eq.~\eqref{eq:L-DAM-app}, we observe that each jump operator $L_k$,
\begin{align}
\|L_k X L_k^\dagger\|_1 &\leq \|L_k\|_\infty^2\|X\|_1\,,\nonumber\\
\left\|\frac{1}{2}\{L_k^\dagger L_k,X\}\right\|_1 &\leq\|L_k^\dagger L_k\|_\infty\|X\|_1 = \|L_k\|_\infty^2\|X\|_1\,,
\end{align}
which implies that 
\begin{equation}
\left\|L_k X L_k^\dagger - \frac{1}{2}\{L_k^\dagger L_k,X\}\right\|_1
\le 2\|L_k\|_\infty^2\|X\|_1\,.
\end{equation}
and, thus, the dissipator $\mathcal{D}(\cdot)$ in Eq.~\eqref{eq:L-DAM-app} satisfies
\begin{align}
\left\|\mathcal{D}\right\|_{1\to1} &:= \left\|\sum_{k=3,4} L_k(\cdot)L_k^\dagger-\frac{1}{2}\{L_k^\dagger L_k,\cdot\}\right\|_{1\to1}\nonumber\\
&\leq 2\sum_{k=3,4}  \|L_k\|_\infty^2\,.
\end{align}
In the zero-dephasing, symmetric-lifetime case, we have $L_3=\sqrt{\gamma}\,|g\rangle\!\langle D|$ and $L_4=\sqrt{\gamma}\,|g\rangle\!\langle A|$, such that $\|L_3\|_\infty^2=\|L_4\|_\infty^2=\gamma$ and
\begin{equation}
\|\mathcal{D}\|_{1\to1} \leq 4\gamma.
\end{equation}
Combining these estimates via the triangle inequality eventually yields
\begin{equation}\label{eq:L1to1-DA}
\|\mathcal{L}\|_{1\to1}\leq \Omega + 4\gamma = \sqrt{4J^2 + \Delta^2} + 4\gamma.
\end{equation}

Inserting Eq.~\eqref{eq:L1to1-DA} into Eq.~\eqref{eq:DA-uniform-Lmax-generic} gives the explicit
uniform variation bound
\begin{equation}\label{eq:DA-uniform-Lmax-final}
\left|\mathcal{C}_M(\Lambda_{t_2}) - \mathcal{C}_M(\Lambda_{t_1})\right|\leq \Delta\tau  c_{M,\mathcal{G}}\left(\sqrt{4J^2 + \Delta^2} + 4\gamma\right)\,.
\end{equation}
For the donor–acceptor yield observable $M = |A\rangle\!\langle A|$, we have $\|M\|_\infty = 1$, so $c_{M,\mathcal{G}}=R_{\mathcal{G}}$, which is a fixed constant depending only on the dephasing map $\mathcal{G}$ in the site basis with $R_{\mathcal{G}}<2$. Here it is important to note that the bound in Eq.~\eqref{eq:DA-uniform-Lmax-final} will be in general an upper bound to the (numerically determined) true value of $\|\mathcal{L}\|_{1\to 1}$, which is not directly accessible in a closed analytic form for the chosen Lindblad generator describing the dynamics of our donor-acceptor dimer. 

On resonance, $\Delta=0$ and $\Omega=2|J|$ (see also Appendix \ref{sec:zero-det-DA} for a general discussion of zero detuning with non-zero dephasing), the expressions for the resource impact functional and the instantaneous rate simplify to
\begin{align}
\mathcal{C}_{|A\rangle\!\langle A|}(\Lambda_t) &= \tfrac12 e^{-\gamma t}\left|\sin(2Jt)\right|
\end{align}
and 
\begin{align}
\Gamma_{|A\rangle\!\langle A|}(t)&= e^{-\gamma t}|J|\left|\cos(2Jt) - \frac{\gamma}{2J}\sin(2Jt)\right|\nonumber\\
&= e^{-\gamma t}|J|\,R\,\left|\cos(2Jt+\phi)\right|\,,
\end{align}
where
\begin{align}
R := \sqrt{1+\frac{\gamma^2}{4J^2}}\,,\quad \phi := \arctan\left(\frac{\gamma}{2J}\right)\,.
\end{align}
Using $|\cos(\cdot)|\leq 1$, we obtain the uniform variation bound
\begin{equation}\label{eq:DA-ex-vb}
\left|\Delta\mathcal{C}_{|A\rangle\!\langle A|}\right|\leq  \frac{|J|}{\gamma}\sqrt{1+\frac{\gamma^2}{4J^2}} \left(e^{-\gamma t_1}-e^{-\gamma t_2}\right)\,.
\end{equation}
For a fixed initial time $t_1$, Eq.~\eqref{eq:DA-ex-vb} yields a lower bound on the minimal time $\Delta\tau^* = t_2-t_1$ (recall Eq.~\eqref{eq:min-time}) required to achieve a target change $|\Delta\mathcal{C}_{\ket{A}\!\bra{A}}^*|$:
\begin{equation}\label{eq:Dtau-DA}
\Delta\tau^* \geq \frac{1}{\gamma} \ln\left(1 - \frac{\gamma e^{\gamma t_1}\left|\Delta\mathcal{C}_{|A\rangle\!\langle A|}^*\right|}{|J|\sqrt{1+\gamma^2/(4J^2)}}\right)^{-1}.
\end{equation}
Feasibility (i.e.,the existence of a finite $\Delta\tau>0$) requires the argument of the logarithm to be at least one, which is equivalent to
\begin{equation}
0 \leq \frac{\gamma e^{\gamma t_1}\left|\Delta\mathcal{C}_{|A\rangle\!\langle A|}^*\right|}{|J|\sqrt{1+\gamma^2/(4J^2)}} < 1,
\end{equation}
or
\begin{equation}\label{eq:feasibility-bound-Zdep}
\left|\Delta\mathcal{C}_{|A\rangle\!\langle A|}\right| < \frac{|J|}{\gamma}e^{-\gamma t_1}\sqrt{1+\frac{\gamma^2}{4J^2}}\quad \text{(feasibility bound)}.
\end{equation}
Thus the maximal achievable change in the task capacity over any interval starting at $t_1$ is directly bounded in terms of the coupling $J$ and decay rate $\gamma$. If the desired $|\Delta\mathcal{C}_{|A\rangle\!\langle A|}|$ exceeds this bound in Eq.~\eqref{eq:feasibility-bound-Zdep}, it cannot be achieved within this model: no amount of coherence generated by $\Lambda_\theta$ will produce such a change in finite time for the given $J$ and $\gamma$.

Finally, we return to the looser uniform bound arising bounding $\|\mathcal{L}_s\|_{1\to 1}\leq L_{\max}$ in Eq.~\eqref{eq:DA-uniform-Lmax-final}. Specialising to resonance, $\Delta=0$ and, thus, $\Omega=2|J|$, the bound in Eq.~\eqref{eq:DA-uniform-Lmax-final} becomes
\begin{equation}
\left|\Delta \mathcal{C}_{|A\rangle\langle A|}\right|\leq \Delta\tau R_{\mathcal{G}}(2|J| + 4\gamma)\,.
\end{equation}
As expected, this $L_{\max}$-based uniform bound is looser than the resonant bound in Eq.~\eqref{eq:DA-ex-vb}, which exploits the explicit form of $\Gamma_{|A\rangle\langle A|}(t)$ and its exponential decay. In particular, the right-hand side of Eq.~\eqref{eq:DA-ex-vb} saturates at a finite value as $t_2\to\infty$, whereas Eq.~\eqref{eq:DA-uniform-Lmax-final} grows linearly with $\Delta\tau$.

\subsection{Zero detuning and symmetric lifetimes\label{sec:zero-det-DA}}

We now consider the case of zero detuning, $\Delta = 0$, and symmetric lifetimes $\gamma_A=\gamma_D=\gamma$, while the pure dephasing rate $\gamma_\varphi$ remains arbitrary. For $\Delta=0$, the coupled linear differential equations in Eq.~\eqref{eq:DA-CDE} simplify to
\begin{align}
\frac{\mathrm{d}}{\mathrm{d}t}u(t) &= -\xi  u(t)\,,\quad\frac{\mathrm{d}}{\mathrm{d}t}v(t) =-\xi v(t)+J s(t)\,,\nonumber\\
\frac{\mathrm{d}}{\mathrm{d}t}s(t)& =  -4 Jv(t)-\gamma s(t)\,,\quad \frac{\mathrm{d}}{\mathrm{d}t}N(t):= -\gamma[N(t) -2\mu_g]
\end{align}
with
\begin{align}
\xi &:= \gamma_\varphi + \gamma\,.
\end{align}
The equations for $u$ and $N$ integrate immediately to
\begin{align}
u(t) &= \Re(\nu)\mathrm{e}^{-\xi t}\,,\\
N(t) &= (\mu_A+\mu_D -2\mu_g)\mathrm{e}^{-\gamma t}+ 2\mu_g\,.
\end{align}

The remaining coupled equations for $v(t)$ and $s(t)$ can be rearranged into a single damped harmonic oscillator equation. Then, eliminating $s(t)$ yields
\begin{equation}
\frac{\mathrm{d}^2}{\mathrm{d}t^2}v(t)+ (\gamma +\xi) \frac{\mathrm{d}}{\mathrm{d}t} v(t +\left(4J^2+\gamma\xi\right) v(t) =0\,,
\end{equation}
with characteristic roots
\begin{equation}\label{eq:lpm-DMA}
\lambda_{\pm}
= -\frac{1}{2}\left(\gamma+\xi\pm \sqrt{\gamma_\varphi^2-16 J^2}\right).
\end{equation}
In the following, we distinguish the underdamped regime $\gamma_\varphi<4|J|$, the overdamped regime
$\gamma_\varphi>4|J|$, and the critical case $\gamma_\varphi=4|J|$.

In the underdamped case, $\gamma_\varphi<4|J|$, the roots are complex, $\lambda_\pm=-\zeta\pm i\omega$, with
\begin{equation}
\zeta := \frac{\gamma+\xi}{2}\,,\quad \omega := \frac{1}{2}\sqrt{16J^2-\gamma_\varphi^2}\,.
\end{equation}
Imposing the initial conditions $v(0)=\Im(\nu)$ and $s(0)=\mu_A-\mu_D$ leads to
\begin{widetext}
\begin{align}\label{eq:vt-ZD}
v(t) &= \mathrm{e}^{-\zeta t}\left[
\Im(\nu)\cos(\omega t)
+\frac{2J(\mu_A-\mu_D)-\gamma_\varphi \Im(\nu)}{2\omega}\sin(\omega t)
\right],\\[1ex]
\label{eq:st-ZD}
s(t) &= \mathrm{e}^{-\zeta t}\left[
(\mu_A-\mu_D)\cos(\omega t)
+\frac{\gamma_\varphi (\mu_A-\mu_D)-8 J \Im(\nu)}{2\omega}\sin(\omega t)
\right].
\end{align}
In the overdamped case, $\gamma_\varphi>4|J|$, the roots are real and negative, $\lambda_\pm=-\zeta\pm\kappa$, with
\begin{equation}
\kappa := \frac{1}{2}\sqrt{\gamma_\varphi^2-16J^2}>0\,,
\end{equation}
and the solutions become
\begin{align}
v(t) &= \mathrm{e}^{-\zeta t}\left[\Im(\nu)\cosh(\kappa t)
+\frac{2J(\mu_A-\mu_D)-\gamma_\varphi \Im(\nu)}{2\kappa}\sinh(\kappa t)
\right]\,,\label{eq:v-overdamped}\\
s(t) &= \mathrm{e}^{-\zeta t}\left[(\mu_A-\mu_D)\cosh(\kappa t) +\frac{\gamma_\varphi (\mu_A-\mu_D)-8 J \Im(\nu)}{2\kappa}\sinh(\kappa t) \right]\,.\label{eq:s-overdamped}
\end{align}
\end{widetext}
The critical case $\gamma_\varphi=4|J|$ is obtained as the limit $\omega\to 0$ (or equivalently $\kappa\to 0$) and yields polynomial prefactors multiplying $\mathrm{e}^{-(\gamma+2|J|)t}$.
In all three regimes the task advantage norm and instantaneous advantage rate follow as
\begin{align}
\mathcal{C}_M(\Lambda_t) &= \sqrt{u^2(t)+v^2(t)}\,,\\
\Gamma_M(t) &=  \sqrt{\xi^2 u^2(t) + [Js(t)-\xi v(t)]^2}\,,\label{eq:GM-det-zero}
\end{align}
in agreement with Eq.~\eqref{eq:CM-GM-DM}.

To obtain analytic variation and time bounds, we first use
\begin{align}
\Gamma_M(t) &= \sqrt{\xi^2 u^2(t) + [Js(t)-\xi v(t)]^2}\nonumber\\
&\leq |\xi u(t)| + |Js(t)-\xi v(t)|
\end{align}
and then bound the two contributions separately. The first term gives
\begin{equation}
\int_{t_1}^{t_2}\mathrm{d}s\,|\xi u(s)| = |\Re(\nu)|\left(\mathrm{e}^{-\xi t_1}-\mathrm{e}^{-\xi t_2}\right)\,,
\end{equation}
while the second term leads to exponentially decaying envelopes in the underdamped and overdamped regimes. In the underdamped case, Eqs.~\eqref{eq:vt-ZD}–\eqref{eq:st-ZD} imply
\begin{equation}
Js(t)-\xi v(t) = \mathrm{e}^{-\zeta t}\left[a_u\cos(\omega t) + b_u\sin(\omega t)\right]\,,
\end{equation}
with
\begin{align}
a_u &:= J(\mu_A-\mu_D) - \xi\,\Im(\nu)\,,\\
b_u &:= \frac{\Im(\nu)(\xi\gamma_\varphi - 8J^2)- J(\mu_A-\mu_D)(\gamma+\xi)}{2\omega}\,,
\end{align}
such that
\begin{equation}
|Js(t)-\xi v(t)| \le \mathrm{e}^{-\zeta t}R_u\,,\quad
R_u := \sqrt{a_u^2+b_u^2}\,.
\end{equation}
In the overdamped case one similarly finds
\begin{equation}
Js(t)-\xi v(t) = \mathrm{e}^{-\zeta t}\left[a_u\cosh(\kappa t) + b_o\sinh(\kappa t)\right]\,,
\end{equation}
with 
\begin{equation}
b_o := \frac{\Im(\nu)(\xi\gamma_\varphi - 8J^2)  - J(\mu_A-\mu_D)(2\gamma+\gamma_\varphi)}{2\kappa}\,.
\end{equation}
Since $k, t\geq 0$, we can use in the overdamped case $\sinh(x)^2+ \cosh(x)^2 = \cosh(2x)\leq \mathrm{e}^{2x}$ for $x\geq 0$ and apply Cauchy-Schwarz inequality to obtain
\begin{align}
|a\cosh(x) + b\sinh(x)| &\leq \sqrt{a^2+b^2}\sqrt{\cosh(2x)}\nonumber\\
&\leq \mathrm{e}^{x} \sqrt{a^2+b^2}\,,
\end{align}
which leads to
\begin{equation}
|Js(t)-\xi v(t)| \le \mathrm{e}^{(\kappa-\zeta)t}R_o,\qquad
R_o := \sqrt{a_u^2 + b_o^2}.
\end{equation}
Below, we skip the critical case as it follows again from the other two cases by taking a suitable limit $\omega\to 0/\kappa\to 0$. 
We then eventually obtain for the variation bound
\begin{widetext}
\begin{align}\label{eq:DC-Zdet}
|\Delta\mathcal{C}_{M}| \leq \int_{t_1}^{t_2}\mathrm{d}s\,\Gamma_M(s)\leq \Re(\nu)\left(\mathrm{e}^{-\xi t_1} -\mathrm{e}^{-\xi t_2}\right)+  \begin{cases}
\frac{R_u}{\zeta} \left(  \mathrm{e}^{-\zeta t_1}- \mathrm{e}^{-\zeta t_2} \right) \quad &\text{ if }\gamma_\varphi < 4 |J|    \\
\frac{R_o}{\zeta-\kappa}\left(\mathrm{e}^{(\kappa-\zeta)t_1} - \mathrm{e}^{(\kappa-\zeta)t_2}\right) \quad &\text{ if }\gamma_\varphi > 4 |J|   
\end{cases}\,.
\end{align}
The variation bound on $|\Delta\mathcal{C}_M|$ \eqref{eq:DC-Zdet} is valid for any $M$ of the form in Eq.~\eqref{eq:M-DAM}. 

To obtain uniform time bounds analytically, we consider in the following again the special case $M=\ket{A}\!\bra{A}$ as then the first term on the right-hand side in Eq.~\eqref{eq:DC-Zdet} vanishes and we can solve for an anticipated $|\Delta\mathcal{C}_{|A\rangle\!\langle A|}^*|$ for the minimal time $\Delta \tau^* \geq 0$:
\begin{align}
\Delta\tau^* \geq \begin{cases}
\frac{1}{\zeta}\ln\left(1 - \frac{\zeta\left|\Delta\mathcal{C}_{|A\rangle\!\langle A|}^*\right| \mathrm{e}^{\zeta t_1}}{R_u}\right)^{-1}\,\,&\text{ if }\gamma_\varphi < 4 |J|    \\
\frac{1}{\zeta - \kappa}\ln\left(1 - \frac{(\zeta-\kappa)\left|\Delta\mathcal{C}_{|A\rangle\!\langle A|}^*\right| \mathrm{e}^{(\zeta-\kappa)t_1}}{R_o}\right)^{-1}\,\, &\text{ if }\gamma_\varphi > 4 |J|   
\end{cases}\,.
\end{align}
\end{widetext} 

In analogy to Sec.~\ref{sec:zero-deph-DA}, the time bound eventually yields a hard constraint on the maximal possible magnitude of the yield change in terms of the respective system parameters, namely
\begin{align}\label{eq:feasibility-bound-Zdet}
|\Delta\mathcal{C}_{|A\rangle\!\langle A|}| <\begin{cases}
\frac{R_u}{\zeta}\mathrm{e}^{-\zeta t_1} \, &\text{ if }\gamma_\varphi < 4 |J|    \\
\frac{R_o}{\zeta-\kappa}\mathrm{e}^{-(\zeta-\kappa)t_1}\, &\text{ if }\gamma_\varphi > 4 |J|   
\end{cases}\,,
\end{align}
which constitutes a strict feasibility bound for finite times $t_2\geq 0$.

\bibliography{Refs2-RT}

\end{document}